\documentclass[11pt]{article}

\usepackage{amsmath,amssymb,xspace,amsthm,ctable,url,citesort,multirow,wrapfig,txfonts,fge}
\usepackage{pstricks,pst-node,pst-plot,pst-coil,pst-text,multido,pst-3d,pst-tree,pst-grad}
\usepackage{mathptmx} 
\usepackage{picins,trfsigns,graphics,pspicture,nicefrac,longtable} 
\usepackage{trfsigns,graphics,pspicture,nicefrac,longtable} 

\newtheorem{theorem}{Theorem}[section]
\newtheorem{proposition}[theorem]{Proposition}
\newtheorem{lemma}[theorem]{Lemma}

\newtheorem{definition}[theorem]{Definition}

\newtheorem{remark}[theorem]{Remark}

\newcommand{\FI}[1]{Fig.~\ref{#1}\xspace}
\newcommand{\EG}{{\it e.g.}\xspace}
\newcommand{\IE}{{\em i.e.}\xspace}
\newcommand{\tx}{^{\rm th}}

\newcommand{\bank}{\mathsf{Bank}}

\newcommand{\din}{{\mathrm{deg}_{\mathrm{in}}}}
\newcommand{\dout}{{\mathrm{deg}_{\mathrm{out}}}}
\newcommand{\vi}{\mathsf{SI}}
\newcommand{\snsvi}{\mathsf{SI}_{\textsc{SaNS}}}

\newcommand{\ssvi}{\mathsf{SI}_{\textsc{SaS}}}

\newcommand{\dvi}{\mathsf{DSI}}
\newcommand{\ssdvi}{\mathsf{DSI}_{\textsc{SaS}}}
\newcommand{\snsdvi}{\mathsf{DSI}_{\textsc{SaNS}}}
\newcommand{\infl}{\mathsf{infl}}

\newcommand{\apx}{\mathsf{APX}}
\newcommand{\NP}{\mathsf{NP}}
\newcommand{\level}{\mathsf{level}}

\newcommand{\eps}{\varepsilon}
\newcommand{\cU}{\mathcal{U}}
\newcommand{\cS}{\mathcal{S}}
\newcommand{\opt}{\mathsf{OPT}}
\newcommand{\nbr}{\mathsf{Nbr}}
\newcommand{\vs}{{V_{\mathrm{shock}}}}
\newcommand{\vt}{{v_{\mathrm{top}}}}
\newcommand{\vsi}{{v_{\mathrm{side}}}}
\newcommand{\galive}{{G_{\mathrm{alive}}}}
\newcommand{\valive}{{V_{\mathrm{alive}}}}
\newcommand{\ealive}{{F_{\mathrm{alive}}}}

\newcommand{\nam}{{{\sc Stab$_{T,\Phi}$}}}
\newcommand{\dnam}{{{\sc Dual-Stab$_{T,\Phi,\kappa}$}}}

\newcommand{\E}{\mathcal{E}}

\newcommand{\comment}[1]{}

\newcommand{\iz}{\mathsf{iz}}

\newcommand{\SC}{{\sc Set}-{\sc Cover}}
\newcommand{\minrep}{{\sc Minrep }}
\newcommand{\vl}{V^{\mathrm{left}}}
\newcommand{\vr}{V^{\mathrm{right}}}
\newcommand{\GS}{G_{\mathrm{super}}}
\newcommand{\DS}{{\mathsf{DS}}}
\newcommand{\VS}{V_{\mathrm{super}}}
\newcommand{\ES}{F_{\mathrm{super}}}
\newcommand{\ovr}[1]{\overrightarrow{#1}}

\renewcommand{\O}{\mathcal{O}}

\title{On the Computational Complexity of Measuring Global Stability of Banking Networks\thanks{Talks based on these results were given or will be given 
at the $4\tx$ annual New York Computer Science and Economics Day, New York University, September 16, 2011, 
at the Industrial-Academic Workshop on Optimization in Finance and Risk Management, October 3-4, 2011, Fields Institute, Toronto, Canada, and 
at the Mathematical Finance theme, 2012 Annual Meeting of the Canadian Applied and Industrial Mathematics Society, July 24-28, 2012.}}

\author{
Piotr Berman \\
Department of Computer Science \& Engineering \\
Pennsylvania State University \\
University Park, PA 16802 \\
Email: {\tt berman@cse.psu.edu}
\and 
Bhaskar DasGupta \& Lakshmi Kaligounder \\
Department of Computer Science \\
University of Illinois at Chicago \\
Chicago, IL 60607-7053 \\
Email: {\tt dasgupta@cs.uic.edu, lkalig2@uic.edu} \\
\and 
Marek Karpinski \\
Department of Computer Science \\
University of Bonn \\
53117 Bonn, Germany \\
Email: {\tt marek@cs.uni-bonn.de}
}

\begin{document}

\maketitle

\begin{abstract}
Threats on the stability of a financial system may severely affect the functioning of the entire economy, and thus considerable emphasis is placed on the analyzing the cause and 
effect of such threats. The financial crisis in the current and past decade has shown that one important cause of instability in global markets is the 
so-called {\em financial contagion}, namely the spreadings of instabilities or failures of {\em individual} components of the network to other, perhaps healthier, components.
This leads to a natural question of whether the regulatory authorities could have predicted and perhaps mitigated the current economic crisis by effective computations of some stability 
measure of the banking networks. Motivated by such observations, we consider the problem of defining and evaluating stabilities of both homogeneous and heterogeneous banking networks against 
propagation of {\em synchronous idiosyncratic shocks} given to a subset of banks. We formalize the homogeneous banking network model of Nier {\em et al.}~\cite{NYYA07} and its 
corresponding heterogeneous version, formalize the synchronous shock propagation procedures outlined in~\cite{NYYA07,E04}, define two appropriate stability 
measures and investigate the computational complexities of evaluating these measures for various network topologies and parameters of interest.
Our results and proofs also shed some light on the properties of topologies and parameters of the network that may lead to higher or lower stabilities.
\end{abstract}

\section{Introduction and Motivation}

In market-based economies, financial systems perform important financial intermediation functions of borrowing from surplus units and lending to deficit units. 
Financial stability is the ability of the financial systems to absorb shocks and perform its key functions, even in stressful situations. 
Threats on the stability of a financial system may severely affect the functioning of the entire economy, and thus considerable emphasis is placed on the analyzing the cause and effect of 
such threats. The concept of instability of a market-based financial system due to factors such as debt financing of investments can be traced back to earlier works of the economists such as 
Irving Fisher~\cite{F33} and John Keynes~\cite{K36} during the 1930's Great Depression era. Subsequently, some economists such as Hyman Minsky~\cite{M77} have argued that:

\begin{quote}
{\em Such instabilities are inherent in many modern capitalist economies}.
\end{quote}

\noindent
In this paper, we investigate systemic instabilities of the banking networks, an important component of modern capitalist economies of many countries.
The financial crisis in the current and past decade has shown that an important component of instability in global financial markets is the so-called {\em financial contagion}, namely 
the spreadings of instabilities or failures of {\em individual} components of the network to other, perhaps healthier, components.
The general topic of interest in this paper, motivated by the global economic crisis in the current and the past decade, is the phenomenon of financial contagion 
in the context of {\em banking networks}, and is related to the following natural extension of the question posed by Minsky and others:
\begin{itemize}
\item
What is the true characterization of such instabilities of banking networks, \IE, 
\begin{itemize}
\item
Are such instabilities systemic, 
\EG, 
caused by a repeal of Glass-Steagall act with subsequent development of specific properties of banking networks that allowed a ripple effect~\cite{CG91}? 

\item 
Or, are such instabilities caused just by a few banks that were ``too big to fail'' and/or ``a few individually greedy executives'' ? 
\end{itemize}
\end{itemize}
To investigate these types of questions, one must first settle the following issues:
\begin{itemize}
\item
What is the {\em precise} model of the banking network that is studied?

\item
How {\em exactly} failures of individual banks propagated through the network to other banks? 

\item
What is an {\em appropriate stability measure} and what are the computational properties of such a measure? 
\end{itemize}
As prior researchers such as Allen and Babus~\cite{AG08} pointed out,graph-theoretic concepts provide a conceptual framework within which various patterns of 
connections between banks can be described and analyzed in a meaningful way by modeling banking networks as a {\em directed} network in which nodes 
represent the banks and the links represent the direct exposures between banks.
Such a network-based approach to studying financial systems is particularly important for assessing financial stability, and in capturing the externalities 
that the risk associated with a single or small group of institutions may create for the entire system. 
Conceptually, links between banks have two {\em opposing} effects on contagion: 
\begin{itemize}
\item
More interbank links increase the opportunity for spreading failures to other banks~\cite{GK08}: 
when one region of the network suffers from a crisis, another region also incurs a loss because 
their claims on the troubled region fall in value and, if this spillover effect is strong enough, it can cause a crisis in adjacent regions.

\item 
More interbank links provide banks with a form of {\em coinsurance} against uncertain liquidity flows~\cite{AG2000},  
\IE, banks can insure against the liquidity shocks by exchanging deposits through links in the network. 
\end{itemize}

\section{The Banking Network Model}

\subsection{Rationale Behind the Model}

As several prior researchers such as~\cite{AG08,NYYA07,E04,PYR09} have already commented, graph-theoretic frameworks 
may provide a powerful tool for analyzing stability of banking and other financial networks.
We provide and use a mathematically precise abstraction of a banking network model as outlined in~\cite{NYYA07} and elsewhere.
The same or very similar version of the graph-theoretic loss propagation model used in this paper has also been 
extensively used by prior researchers in finance, economics and banking industry to study various properties and research questions involving banking systems
similar to what is studied in this paper (\EG, see~\cite{F03,UW04,M07,ACM11,CMS10}, to name a few). As commented by researchers such as~\cite{NYYA07,ACM11}:

\begin{quote}
{\em the modelling challenge in studying banking networks lies not so much in analyzing a model that is flexible enough to represent all types of insolvency cascades, but in 
studying a model that can mimic the empirical properties of these different types of networks}.
\end{quote}

\noindent
A loss propagation model such as the one discussed here and elsewhere such as in~\cite{F03,UW04,M07,ACM11,CMS10}
conceptualises the main characteristics of a financial system using network theory
by relating the cascading behavior of financial networks both to the local properties of the nodes and to the underlying topology of the network, 
allowing us to vary continuously the key parameters of the network.

\subsection{Homogeneous Networks: Balance Sheets and Parameters for Banks}

We provide a precise abstraction of the model as outlined in~\cite{NYYA07} which builds up on the 
works of Eboli~\cite{E04}. The network is modeled by a weighted directed graph $G=(V,F)$
of $n$ nodes and $m$ directed edges, where each node $v\in V$ corresponds to a bank ($\bank_v$) and each directed
edge $(v,v')\in F$ indicates that $\bank_v$ has an agreement to lend money to $\bank_{\,v'}$.
Let $\din(v)$ and $\dout(v)$ denote the in-degree and the out-degree of node $v$. 
The model has the following parameters:

\vspace*{0.1in}
\hspace*{-0.4in}
{\small
\begin{tabular}{c}\toprule
$E=$ total external asset,  \hspace*{0.2in} $I=$ total inter-bank exposure,  \hspace*{0.2in} $A=I+E=$ total asset \\
\hspace*{-0.15in}
$[0,1]\ni\gamma=$ percentage of equity to asset, \hspace*{0.05in} $w=w(e)=\frac{I}{m}=$ weight of edge $e\in F$, \hspace*{0.05in} $\Phi=$ severity of shock ($1\geq\Phi>\gamma$) \\
\bottomrule
\end{tabular}
}

\vspace*{0.1in}
\noindent
Now, we describe the balance sheet for a node $v\in V$ (\IE, for $\bank_v$)\footnote{This model assumes that all the depositors are insured
for their deposits, \EG, in United States the Federal Deposit Insurance Corporation provides such an insurance up to a maximum level. Thus, 
{\bf we will omit the parameters $d_v$ for all $v$ in the rest of the paper when using the model}. Similarly, $\ell_v$ quantities (which depend on the $d_v$'s) are also only 
necessary in writing the balance sheet equation and will not be used subsequently.}:

\vspace*{0.1in}
\hspace*{-0.6in}
{\small 
\begin{tabular}{l|l} \toprule
\multicolumn{1}{c}{\bf Assets} & \multicolumn{1}{c}{\bf Liabilities} \\ \hline
$
\begin{array}{rll}
\iota_v & = & \dout(v)\times w = \mbox{interbank asset} \\
e_v  & = & (b_v-\iota_v) + \frac{E-\sum_{v\in V} {(b_v-\iota_v)}}{n}=(b_v-\iota_v) + \frac{E}{n} \\
     & = & \mbox{share of total external asset $E$} \\
a_v=e_v+\iota_v & = & b_v+\frac{E}{n}= \mbox{total asset} \\
\end{array}
$
&
$
\begin{array}{rll}
b_v     & = & \din(v)\times w  = \mbox{interbank borrowing} \\
c_v = \gamma \times a_v & = & \mbox{net worth (equity)} \\
d_v     & = & \mbox{customer deposits} \\
\ell_v = b_v+c_v+d_v & = & \mbox{total liability} \\
\end{array}
$
\\ \hline
\multicolumn{2}{c}{$a_v=\ell_v$ (balance sheet equation)} \\
\bottomrule
\end{tabular}
}

\vspace*{0.1in}
\noindent
Note that the homogeneous model is completely described by the $4$-tuple of parameters $\pmb{\langle} G,\gamma,I,E \pmb{\rangle}$.

\subsection{Balance Sheets and Parameters for Heterogeneous Networks}

The heterogeneous version of the model is the same as its' homogeneous counterpart as described above, except 
that the shares of interbank exposures and external assets for different banks may be different.
Formally, the following modifications are done in the homogeneous model:
\begin{itemize}
\item
$w(e)>0$ denotes the weight of the edge $e\in E$ along with the constraint that $\sum_{e\in F} w(e)=I$. 

\item
\vspace*{-6pt}
$\iota_v=\sum_{e=(v,v')\in F} w(e)$, and $b_v=\sum_{e=(v',v)\in F} w(e)$.

\item
\vspace*{-4pt}
$e_v= (b_v-\iota_v) + \alpha_v\times \big(E-\sum_{v\in V} (b_v-\iota_v)\big)$ for some $\alpha_v>0$ 
along with the constraint $\sum_{v\in V}\alpha_v=1$.
Since $\sum_{v\in V} (b_v-\iota_v)=0$, this gives  $e_v= (b_v-\iota_v) + \alpha_v E$. 
Consequently, $a_v$ now equals $b_v+\alpha_v E$.
\end{itemize}
Denoting the $m$-dimensional vector of $w(e)$'s by $\mathbf{w}$ and 
the $n$-dimensional vector of $\alpha_v$'s by $\pmb{\alpha}$, 
the heterogeneous model is completely described by the $6$-tuple of parameters $\pmb{\langle} G,\gamma,I,E,\mathbf{w},\pmb{\alpha}\pmb{\rangle}$.

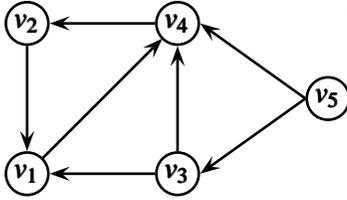
\begin{figure}[htbp]
\begin{center}
\begin{minipage}[b]{1in}
\begin{pspicture}(-0.7,0)(4,3)
\psset{xunit=1cm,yunit=1cm}
\pscircle[linewidth=1pt,origin={0,0},fillstyle=none,fillcolor=lightgray](0,0){0.3}
\rput(0,0){$\pmb{v_1}$}
\pscircle[linewidth=1pt,origin={0,2},fillstyle=none,fillcolor=lightgray](0,0){0.3}
\rput(0,2){$\pmb{v_2}$}
\pscircle[linewidth=1pt,origin={2,0},fillstyle=none,fillcolor=lightgray](0,0){0.3}
\rput(2,0){$\pmb{v_3}$}
\pscircle[linewidth=1pt,origin={2,2},fillstyle=none,fillcolor=lightgray](0,0){0.3}
\rput(2,2){$\pmb{v_4}$}
\pscircle[linewidth=1pt,origin={4,1},fillstyle=none,fillcolor=lightgray](0,0){0.3}
\rput(4,1){$\pmb{v_5}$}
\psline[linewidth=1pt,arrowsize=1.5pt 4,linecolor=black]{<-}(0.3,0)(1.7,0)
\psline[linewidth=1pt,arrowsize=1.5pt 4,linecolor=black]{<-}(0,0.3)(0,1.7)
\psline[linewidth=1pt,arrowsize=1.5pt 4,linecolor=black]{->}(0.2,0.2)(1.8,1.8)
\psline[linewidth=1pt,arrowsize=1.5pt 4,linecolor=black]{<-}(0.3,2)(1.7,2)
\psline[linewidth=1pt,arrowsize=1.5pt 4,linecolor=black]{->}(2,0.3)(2,1.7)
\psline[linewidth=1pt,arrowsize=1.5pt 4,linecolor=black]{<-}(2.3,0)(3.7,1)
\psline[linewidth=1pt,arrowsize=1.5pt 4,linecolor=black]{<-}(2.3,2)(3.7,1)
\end{pspicture}
\end{minipage}
\hspace*{0.5in}
\begin{minipage}[b]{4in}
\caption{\label{exx}An example of our banking network model.}

\hspace*{0.5in}
$n=\mbox{number of nodes}=5$ 

\hspace*{0.5in}
$m=\mbox{number of edges}=7$ 

\hspace*{0.5in}
$I=\mbox{total inter-bank exposure}=m=7$ 

\hspace*{0.5in}
$E=\mbox{total external asset}=14$, $\gamma=0.1$
\end{minipage}
\end{center}
\end{figure}

\paragraph{Illustration of calculations of balance sheet parameters}
We illustrate the calculation of relevant parameters of the balance sheet of banks for the simple banking network shown in \FI{exx}.

\vspace*{0.1in}
\noindent
({\em a}) Homogeneous version of the network
\begin{itemize}
\item
$w=\mbox{weight of every edge} = \nicefrac{I}{m}=1$.

\item 
$\iota_{v_1}=\dout(v_1)\times w=1, \,\,\iota_{v_2}=\dout(v_2)\times w=1, \,\,\iota_{v_3}=\dout(v_3)\times w=2, \,\,\iota_{v_4}=\dout(v_4)\times w=1, \,\,\iota_{v_5}=\dout(v_5)\times w=2$. 

\item
$b_{v_1}=\din(v_1)\times w=2, \,\,b_{v_2}=\din(v_2)\times w=1, \,\,b_{v_3}=\din(v_3)\times w=1, \,\,b_{v_4}=\din(v_4)\times w=3, \,\,b_{v_5}=\din(v_5)\times w=0$. 

\item
$e_{v_1}=b_{v_1}-\iota_{v_1}+\frac{E}{n}=3.8, \,\, e_{v_2}=b_{v_2}-\iota_{v_2}+\frac{E}{n}=2.8, \,\, e_{v_3}=b_{v_3}-\iota_{v_3}+\frac{E}{n}=1.8, 
\,\, e_{v_4}=b_{v_4}-\iota_{v_4}+\frac{E}{n}=4.8, \,\, e_{v_5}=b_{v_5}-\iota_{v_5}+\frac{E}{n}=0.8$. 

\item
$a_{v_1}=b_{v_1}+\frac{E}{n} =4.8, \,\, a_{v_2}=b_{v_2}+\frac{E}{n} =3.8, \,\, a_{v_3}=b_{v_3}+\frac{E}{n} =3.8, 
\,\, a_{v_4}=b_{v_4}+\frac{E}{n} =5.8, \,\, a_{v_5}=b_{v_5}+\frac{E}{n} =2.8$. 

\item
$c_{v_1}=\gamma\,a_{v_1}=0.48$, 
$c_{v_2}=\gamma\,a_{v_2}=0.38$, 
$c_{v_3}=\gamma\,a_{v_3}=0.38$, 
$c_{v_4}=\gamma\,a_{v_4}=0.58$,
$c_{v_5}=\gamma\,a_{v_5}==0.28$.
\end{itemize}
({\em b}) Heterogeneous version of the network

\vspace*{0.1in}
Suppose that 95\% of $E$ is distributed equally on the two banks $v_1$ and $v_2$, and the rest 5\% of $E$ is distributed equally on the remaining three banks.
Thus:
\[
\textstyle
\alpha_{v_1}E= \frac{0.95\,E}{2}=6.65, 
\,\, \alpha_{v_2}E =\frac{0.95\,E}{2}=6.65, 
\,\, \alpha_{v_3}E =\frac{0.05\,E}{3}\approx 0.233, 
\,\, \alpha_{v_4}E =\frac{0.05\,E}{3}\approx 0.233, 
\,\, \alpha_{v_5}E =\frac{0.05\,E}{3}\approx 0.233
\]
Suppose that 95\% of $I$ is distributed equally on the three edges $f_1=(v_2,v_1),f_2=(v_1,v_4),f_3=(v_4,v_2)$, 
and the remaining 5\% of $I$ is distributed equally on the remaining four edges $f_4=(v_3,v_1),f_5=(v_3,v_4),f_6=(v_5,v_4),f_7=(v_5,v_3)$.
Then, 
\[
\textstyle
w(f_1)=w(f_2)=w(f_3)=\frac{0.95\,I}{3}\approx 2.216, 
\,\,\, w(f_4)=w(f_5)=w(f_6)=w(f_7)=\frac{0.05\,I}{4}=0.08725
\]
\[
\begin{array}{rl}
\mbox{for bank $v_1$:} 
&
\begin{array}{l}
b_{v_1}=w(f_1)+w(f_4)\approx 2.30325, \,\, \iota_{v_1}=w(f_2)=2.216
\\
e_{v_1}=(b_{v_1}-\iota_{v_1})+\alpha_{v_1}E\approx 6.7365, \,\,
a_{v_1}=b_{v_1}+\alpha_{v_1}E=8.9525, \,\,
c_{v_1}=\gamma\,a_{v_1}=0.8925
\end{array}
\\[0.2in]
\mbox{for bank $v_2$:} 
&
\begin{array}{l}
b_{v_2}=w(f_3)\approx 2.216, \,\, \iota_{v_2}=w(f_1)\approx 2.216
\\
e_{v_2}=(b_{v_2}-\iota_{v_2})+\alpha_{v_2}E=6.65, \,\,
a_{v_2}=b_{v_2}+\alpha_{v_2}E\approx 8.866, \,\,
c_{v_2}=\gamma\,a_{v_2}\approx 0.8666
\end{array}
\\[0.2in]
\mbox{for bank $v_3$:} 
&
\begin{array}{l}
b_{v_3}=w(f_7)=0.08725, \,\, 
\iota_{v_3}=w(f_4)+w(f_5)=0.1745
\\
e_{v_3}=(b_{v_3}-\iota_{v_3})+\alpha_{v_3}E\approx 0.14575, \,\,
a_{v_3}=b_{v_3}+\alpha_{v_3}E\approx 0.32025, \,\,
c_{v_3}=\gamma\,a_{v_3}\approx 0.032035
\end{array}
\\[0.2in]
\mbox{for bank $v_4$:} 
&
\begin{array}{l}
b_{v_4}=w(f_2)+w(f_5)+w(f_6)\approx 2.39050, \,\,
\iota_{v_4}=w(f_3)\approx 2.216
\\
e_{v_4}=(b_{v_4}-\iota_{v_4})+\alpha_{v_4}E\approx 0.4075, \,\,
a_{v_4}=b_{v_4}+\alpha_{v_4}E\approx 2.6235, \,\,
c_{v_4}=\gamma\,a_{v_4}\approx 0.26235
\end{array}
\\[0.2in]
\mbox{for bank $v_5$:} 
&
\begin{array}{l}
b_{v_5}=0, \,\, 
\iota_{v_5}=w(f_6)+w(f_7)=0.1745
\\
e_{v_5}=(b_{v_5}-\iota_{v_5})+\alpha_{v_5}E\approx 0.0585, \,\,
a_{v_5}=b_{v_5}+\alpha_{v_5}E\approx 0.233, \,\,
c_{v_5}=\gamma\,a_{v_5}\approx 0.0233
\end{array}
\end{array}
\]

\subsection{Idiosyncratic Shock~\cite{NYYA07,E04}}

As in~\cite{NYYA07}, our initial failures are caused by {\em idiosyncratic shocks} which can occur due to {\em operations risks} (frauds) 
or {\em credit risks}, and has the effect of reducing the external assets of a selected subset of banks perhaps causing them
to default. While {\em aggregated} or {\em correlated} shocks affecting all banks simultaneously is relevant in practice, 
idiosyncratic shocks are a cleaner way to study the {\em stability} of the topology of the banking network.
Formally, we select a non-empty subset of nodes (banks) $\emptyset\subset\vs\subseteq V$. For all nodes $v\in\vs$, we 
simultaneously decrease their external assets from $e_v$ by $s_v=\Phi\,e_v$, where the parameter $\Phi\in(0,1]$ 
determines the ``severity'' of the shock. As a result, the new net worth of $\bank_v$ becomes 
$c_v'=c_v-s_v$. The effect of this shock is as follows:
\begin{itemize}
\item
If $c_v'\geq 0$, $\bank_v$ continues to operate but with a lower net worth of $c_v'$.

\item
If $c_v'<0$, $\bank_v$ {\em defaults} (\IE, stops functioning).
\end{itemize}

\begin{center}
\begin{table}[ht]
\begin{center}
\begin{tabular}{l} \toprule
$t=1$ ; $\valive(1)=V$ 
\\
(* start the shock at $t=1$ on nodes in $\vs$ *) \\
$\pmb{\forall\,v\in V}\colon$ {\bf if} $v\in\vs$ \hspace*{5pt} {\bf then} $c_v(1)=c_v-\Phi\,e_v$ \hspace*{5pt} {\bf else} $c_v(1)=c_v$ 
\\
\\
(* shock propagation at times $t=2,3,\dots,T$ *) 
\\
{\bf while} $\,\,\,(t\leq T)\,\,\,\bigwedge\,\,\,(\valive(t)\neq\emptyset)\,\,\,$ {\bf do} 
\\
\\
\hspace*{0.2in}
(* transmit loss to next time step *) 
\\
\hspace*{0.2in}
$\displaystyle\pmb{\forall\,u\in \valive(t)}\colon c_u(t+1) = c_u(t)-\hspace*{-0.5in}\sum_{{v\colon c_v(\,t\,)<0\,\,\,\, \pmb{\&}\,\,\,\, (u,v)\in\ealive(\,t\,)}} \hspace*{-0.2in}\frac{\min\big\{\,\big|\,c_v(t)\,\big|\,,\,\,b_v\,\big\}}{\din(v,\,t)}$
\\
\\
\hspace*{0.2in}
(* remove $\bank_v$ from network if it is to fail at this step *) 
\\
\hspace*{0.2in}
$\valive(t+1)=\valive(t)\setminus \left\{\,v\,\big| \, v\in\valive(t)\,\,\pmb{\&}\,\,c_v(t)<0 \, \right\}$
\\
\\
\hspace*{0.2in}
$t=t+1$ 
\\
{\bf endwhile} 
\\
\bottomrule
\end{tabular}
\end{center}
\caption{Discrete-time idiosyncratic shock propagation for $T$ steps.}
\label{t1}
\end{table}
\end{center}

\subsection{Propagation of an Idiosyncratic Shock~\cite{NYYA07,E04}}

We use the notation $c_v(t)$ to denote $c_v$ at time $t$, and $t_0^+$ to denote any $t>t_0$. 
Let $\valive(t)\subseteq V$ be the set of nodes that have not failed at time $t$
and let $\galive(t)=(\valive(t),\ealive(t))$ be the corresponding node-induced subgraph of $G$ at time $t$ with 
$\din(v,t)$ and $\dout(v,t)$ denote the in-degree and out-degree of a node $v\in\valive(t)$ in the graph $\galive(t)$. 
In a continuous-time model, the shock propagates as follows:
\begin{itemize}
\item
$\valive(1)=V$, $c_v(1)=c_v-s_v$ if $v\in\vs$, and $c_v(1)=c_v$ otherwise.

\item
If a banks equity ever becomes negative, it fails subsequently, \IE, 
$\forall\,t_0\geq 1\colon c_v(t_0)<0 \Rightarrow v\notin\valive(t_0^+)$.

\item
A failed bank $\bank_v$ at time $t=t_0$ affects the net worth (equity) of all banks that gave loan to $\bank_v$ 
in the following manner. For each edge $(u,v)\in\ealive(t_0)$ in the network at time $t_0$, the equity $c_u(t_0)$ is decreased 
by an amount\footnote{If $|\,c_v(t_0)\,|>b_v$ then the depositors incur a loss of $b_v-|\,c_v(t_0)\,|$, but as already mentioned before
this model assumes that all the depositors are insured for their deposits.}
of $\frac{\min\left\{\,\left|\,c_v(t_0)\,\right|,\,\,b_v\,\right\}}{\din(v,\,t_0)}$.
Thus, the shock propagation is defined by the following differential equation: 
\[
{\forall\,t\geq 1}\,\,\,{\forall\, u\in\valive(t)\colon} \,\,{\frac{\partial\,c_u(t)}{\partial\,{t}}} { \,= 
-\hspace*{-0.2in}\sum_{{v\colon c_v(t)<0 \,\,\pmb{\&}\,\,(u,v)\in\ealive(t)}} \hspace*{-0.2in}\frac{\min\big\{\,\big|\,c_v(t)\,\big|,\,\,b_v\,\big\}}{\din(v,\,t)}}
\]
\end{itemize}
An intuitive explanation of the two quantities inside the summation in the above equation is as follows.
The term $\frac{|\,c_v(t)|}{\din(v,\,t)}$ distributes the loss of equity of a bank equitably among its creditors that have not failed yet.
The term $\frac{b_v}{\din(v,\,t)}$ ensures that the total loss propagated is no more than the total interbank exposure of the failed bank.

A {\em discrete-time} version of the above can be obtained by the obvious method of 
quantizing time and replacing the partial differential equations
by ``difference equations''. With appropriate normalizations, the discrete-time model
for shock propagation is described by 
a synchronous iterative procedure shown in Table~\ref{t1} where $t=1,2,\dots,T$ denotes
the discrete time step at which the synchronous update is done ($T\leq n$).

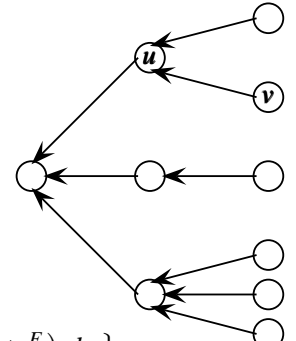
\begin{wrapfigure}[9]{R}{0pt}
\scalebox{0.7}[0.7]{
\begin{pspicture}(-0.4,-2.6)(5,2)
\psset{xunit=0.75cm,yunit=0.75cm}
\pscircle[origin={0,0},linewidth=1pt](0,0){0.3}
\pscircle[origin={3,-3},linewidth=1pt](0,0){0.3}
\pscircle[origin={3,0},linewidth=1pt](0,0){0.3}
\pscircle[origin={3,3},linewidth=1pt](0,0){0.3}
\rput(3,3){\Large$\pmb{u}$}
\pscircle[origin={6,-4},linewidth=1pt](0,0){0.3}
\pscircle[origin={6,-3},linewidth=1pt](0,0){0.3}
\pscircle[origin={6,-2},linewidth=1pt](0,0){0.3}
\pscircle[origin={6,0},linewidth=1pt](0,0){0.3}
\pscircle[origin={6,2},linewidth=1pt](0,0){0.3}
\rput(6,2){\Large$\pmb{v}$}
\pscircle[origin={6,4},linewidth=1pt](0,0){0.3}
\psline[origin={0,0},linewidth=1pt,linecolor=black,arrowsize=1.5pt 8]{->}(2.7,-3)(0,-0.3)
\psline[origin={0,0},linewidth=1pt,linecolor=black,arrowsize=1.5pt 8]{->}(2.7,0)(0.3,0)
\psline[origin={0,0},linewidth=1pt,linecolor=black,arrowsize=1.5pt 8]{->}(2.7,3)(0,0.3)
\psline[origin={0,0},linewidth=1pt,linecolor=black,arrowsize=1.5pt 8]{->}(5.7,-3)(3.3,-3)
\psline[origin={0,0},linewidth=1pt,linecolor=black,arrowsize=1.5pt 8]{->}(5.7,-2)(3,-2.7)
\psline[origin={0,0},linewidth=1pt,linecolor=black,arrowsize=1.5pt 8]{->}(5.7,-4)(3,-3.3)
\psline[origin={0,0},linewidth=1pt,linecolor=black,arrowsize=1.5pt 8]{->}(5.7,0)(3.3,0)
\psline[origin={0,0},linewidth=1pt,linecolor=black,arrowsize=1.5pt 8]{->}(5.7,2)(3,2.7)
\psline[origin={0,0},linewidth=1pt,linecolor=black,arrowsize=1.5pt 8]{->}(5.7,4)(3,3.3)
%
\end{pspicture}
}
\caption{\label{in-arb-fig}An in-arborescence graph.}
\end{wrapfigure}

\paragraph{A simplified illustration of the effect of idiosyncratic shocks}
Consider the case when the model is {\em homogeneous} and 
the topology of the graph $G$ is {\em in-arborescence}, \IE, a directed rooted tree where all edges are oriented towards the root.
Consider two nodes $u,v\in V$ such that $(v,u)\in F$ and $\din(v)=0$ (see \FI{in-arb-fig}).
Suppose that at time $t=1$ the node $u$ is shocked and consequently it {\em defaults}. The amount of shock 
transmitted from $u$ to $v$ is 
\begin{multline*}
\Delta = 
\frac{\min \left\{ |c_u'(1)|, b_u \right\}}{\din(u)}
= 
\frac{\min \left\{ \Phi e_u - c_u(1), b_u \right\}}{\din(u)}
=
\frac{\min \left\{ \Phi \left(b_u-\iota_u+\frac{E}{n} \right)  - \gamma \left(b_u+\frac{E}{n} \right), b_u \right\}}{\din(u)}
\\
=
\frac{\min \left\{ \Phi \left(\din(u)-1+\frac{E}{n} \right)  - \gamma \left(\din(u)+\frac{E}{n} \right), \din(u) \right\}}{\din(u)}
\\
=
\min 
\left\{
\left(\Phi-\gamma\right) \times \left( 1 + \frac{E}{n\,\din(u)} \right) - \frac{1}{\din(u)}, \, 1
\right\}
\end{multline*}
Since $c_v(1)=\gamma\times \nicefrac{E}{n}$, we have 
\[
c_v(2)=c_v(1)-\Delta
=
\gamma\times \frac{E}{n}
-
\min 
\left\{
\left(\Phi-\gamma\right) \times \left( 1 + \frac{E}{n\,\din(u)} \right) - \frac{1}{\din(u)}, \, 1
\right\}
\]
Assuming $\Phi-\gamma  \ll 1 + \frac{E}{n\,\din(u)}$ and $\din(u)\gg 1$, the above expression simplifies to 
\[
c_v(2)
\approx
\gamma\times \frac{E}{n}
-
\left(\Phi-\gamma\right) \times \left( 1 + \frac{E}{n\,\din(u)} \right) 
\]
Suppose that $\gamma=\nicefrac{\Phi}{4}$. Then, 
$c_v(2)\approx 
\gamma \times \left( \frac{E}{n} - 3 - \frac{3\,E}{n\,\din(u)} \right)
$.
Consequently, one can observe the following:
\begin{itemize}
\item
If $\nicefrac{E}{n}\leq 2$, then $c_v(2)<0$ and node $v$ will surely fail at time $t=2$.

\item
If $\nicefrac{E}{n}\geq 4$ and $\din(u)>10$ then $c_v(2)>0$ and node $v$ will surely not fail at time $t=2$.
\end{itemize}

\subsection{Parameter Simplification}
We can assume without loss of generality that in the homogeneous shock propagation model $w=1$. 
To observe this, if $w=\nicefrac{I}{m}\neq 1$, then we can divide each of the quantities 
$\iota_v$, $b_v$, $E$ and $d_v$ by $w$; it is easy to see that the outcome of the shock propagation 
procedure in Table~\ref{t1} remains the same. Moreover, we will ignore the balance sheet equation 
since $d_v$ has no effect in shock propagation.

\section{Related Prior Works on Financial Networks}

Although there is a large amount of literature on stability of financial systems in general and banking systems in particular, much of the prior research 
is on the empirical side or applicable to small-size networks. Two main categories of prior researches can be summarized as follows.
The particular model used in this paper is the model of Nier {\em et al.}~\cite{NYYA07}.
As stated before, definition of a precise stability measure and analysis of its computational complexity issues for stability calculation 
were not provided for these models before.

\paragraph{Network formation}
Babus~\cite{B07} proposed a model in which banks form links with each other as an insurance mechanism to 
reduce the risk of contagion. In contrast, Castiglionesi and Navarro~\cite{FN09} studied decentralization of the network of banks that is optimal from the
perspective of a social planner. In a setting in which banks invest on behalf of depositors and there are positive network externalities on the investment returns, 
fragility arises when ``not sufficiently capitalized'' banks gamble with depositors' money. When the probability of bankruptcy is low, the decentralized solution 
well-approximates the first objective of Babus. 

\paragraph{Contagion spread in networks}
Although ordinarily one would expect the risk of contagion to be larger in a highly interconnected banking system, 
some empirical simulations indicate that shocks have an {\em extremely complex} effect on the network stability in the sense that
higher connectivity among banks may sometimes lead to {\em lower} risk of contagion.

Allen and Gale~\cite{AG2000} studied how a banking system may respond to contagion when banks are connected under different network structures, and found that, 
in a setting where consumers have the liquidity preferences as introduced by Diamond and Dybvig~\cite{DP83} and have random liquidity needs, banks perfectly insure against 
liquidity fluctuations by exchanging interbank deposits, but the connections created by swapping deposits expose the {\em entire system} to contagion. 
Allen and Gale concluded that incomplete networks are {\em more} prone to contagion than networks with maximum connectivity since 
better-connected networks are more resilient via transfer of proportion of the losses in one bank's portfolio to more banks through interbank agreements. 
Freixas {\em et al.}~\cite{XBJ2000} explored the case of banks that face liquidity fluctuations due to the uncertainty about 
consumers withdrawing funds. Gai and Kapadia~\cite{GK08} argued that the higher is the connectivity among banks the more will be the contagion effect during crisis. 
Haldane~\cite{H09} suggested that contagion should be measured based on the interconnectedness of each institution within the financial system.
Liedorp {\em et al.}~\cite{Lie2010} investigated if interconnectedness in the interbank market is a channel through which banks affect each others riskiness, and argued that both large 
lending and borrowing shares in interbank markets increase the riskiness of banks active in the {\em dutch} banking market. 

Dasgupta~\cite{D04} explored how linkages between banks, represented by cross-holding of deposits, can be a source of contagious breakdowns by investigating 
how depositors, who receive a private signal about fundamentals of banks, may want to withdraw their deposits if they believe that enough other depositors will do the same. 
Lagunoff and Schreft~\cite{RS98} considered a model in which agents are linked in the sense that the return on an agents' portfolio depends on the portfolio allocations of other agents.  
Iazzetta and Manna~\cite{CM09} used network topology analysis on monthly data on deposits exchange to gain more insight into the way a liquidity crisis spreads. 
Nier {\em et al.}~\cite{NYYA07} explored the dependency of systemic risks on the structure of the banking system via network theoretic approach and  
the resilience of such a  system to contagious defaults.
Kleindorfer {\em et al.}~\cite{PYR09} argued that network analyses can play a crucial role in understanding many important phenomena in finance. 
Corbo and Demange~\cite{JG09} explored, given the exogenous default of set of banks,
the relationship of the structure of interbank connections to the contagion risk of defaults.
Babus~\cite{A05} studied how the trade-off between the benefits and the costs of being linked changes depending on the network structure, and 
observed that, when the network is maximal, liquidity can be redistributed in the system to make the risk of contagion minimal. 

\section{The Stability and Dual Stability Indices} 

A banking network is called {\em dead} if all the banks in the network have failed.
Consider a given homogeneous or heterogeneous banking network 
$\pmb{\langle}\, G,\gamma,I,E,\Phi \,\pmb{\rangle}$ or $\pmb{\langle}\, G,\gamma,I,E,\Phi,\mathbf{w},\pmb{\alpha}\,\pmb{\rangle}$.
For $\emptyset\subset V'\subseteq V$, let 
\begin{gather*}
\infl(V')=\big\{\, v\in V\,\big|\,v \mbox{ fails if all nodes in $V'$ are shocked}\,\big\}
\\
\vi(G,V',T)
=
\left\{
\begin{array}{cl}
\nicefrac{\big|V'\big|\,}{\,n}\,, & \mbox{if $\infl(V')=V$} \\
\infty\,,          & \mbox{otherwise} \\
\end{array}
\right.
\end{gather*}

\paragraph{The Stability Index}
The optimal {\em stability index} of a network $G$ is defined as 
\[
\boxed{
\vi^\ast(G,T)=\vi(G,\vs,T)=\min_{\,V'}\big\{\,\vi(G,V',T)\,\big\}
}
\]
For estimation of this measure, we assume that it is possible for the network to fail, \IE, 
$\vi^\ast(G,T)<\infty$. Thus, $0<\vi^\ast(G,T)\leq 1$, and the higher the stability index is, the better is the stability 
of the network against an idiosyncratic shock. We thus arrive at the natural computational problem \nam.
We denote an optimal subset of nodes that is a solution of Problem~\nam\ by $\vs$, \IE, $\vi^\ast(G,T)=\vi(G,\vs,T)$.
Note that if $T\geq n$ then the \nam\ finds a minimum subset of nodes which, when shocked, 
will {\em eventually} cause the death of the network in an arbitrary number of time steps.

\vspace*{10pt}
{\small
\hspace*{-0.5in}
\begin{tabular}{l|l}
\toprule
{\bf Input}:
a banking network with shocking parameter $\Phi$, 
&
{\bf Input:} 
a banking network with shocking parameter $\Phi$, 
\\
\hspace*{0.8in} 
and an integer $T>1$
&
\hspace*{0.8in} 
and two integers $T,\kappa>1$
\\
{\bf Valid solution}:
A subset $V'\subseteq V$ such that $\vi(G,V',T)<\infty$
&
{\bf Valid solution:} 
A subset $V'\subseteq V$ such that $|V'|=\kappa$
\\
{\bf Objective}:
{\em minimize} $|V'|$
&
{\bf Objective:} 
{\em maximize} $\,\,\,\big|\,\infl(V')/\kappa\,\big|$
\\ \midrule
\multicolumn{1}{c}{
{\bf Stability of banking network} ($\,\mathbf{\mbox{\nam}}\,$)
}
&
\multicolumn{1}{c}{
{\bf Dual Stability of banking network} ($\,\mathbf{\mbox{\dnam}}\,$)
}
\\
\bottomrule
\end{tabular}
}

\paragraph{The Dual Stability Index} 
Many covering-type minimization problems in combinatorics have a natural maximization dual in which one fixes a-priori the number of covering sets 
and then finds a maximum number of elements that can be covered with these many sets. For example, the usual dual
of the minimum set covering problem is the maximum coverage problem~\cite{KMN99}. Analogously, we define a dual
stability problem \dnam.
The {\em dual stability index} of a network $G$ can then be defined as 
\[
\boxed{
\dvi^\ast(G,T,\kappa)=\max_{V'\subseteq V\colon |V'|=\kappa} \big|\,{\infl(V')}\,/\,{\kappa}\,\big|
}
\]
The dual stability measure is of particular interest when $\vi^\ast(G,T)=\infty$, \IE, the entire network cannot be made to fail.
In this case, a natural goal is to find out if a significant portion of the nodes in the network can be failed by 
shocking a limited number of nodes of $G$; this is captured by the definition of $\dvi^\ast(G,T,\kappa)$.

\paragraph{Violent Death vs. Slow Poisoning}
In our results, we distinguish two cases of death of a network:
\begin{description}
\item[violent death ($T=2$)] 
The network is dead by the very next step after the shock.

\item[slow poisoning (any $T\geq 2$)]
The network may not be dead immediately but dies {\em eventually}.
\end{description}

\subsection{Rationale Behind the Stability Measures}

Although it is possible to think of many other alternate measures of stability for networks than the ones defined in this paper, 
the measures introduced here are in tune with the ideas that references~\cite{NYYA07,E04} directly (and, some other references such 
as~\cite{F03,UW04} implicitly) used to empirically study their networks by shocking only a few (sometimes one) node. Thus, 
a rationale in defining the stability measures in the above manner
is to follow the cue provided by other researchers in the banking industry in studying models such as in this paper instead of creating a 
completely new measure that may be out of sync with 
ideas used by prior researchers and therefore could be subject to criticisms.

\section{Comparison with Other Models for Attribute Propagation in Networks}
\label{comparison-sec}

\begin{wrapfigure}[8]{r}{1.9in}
\begin{pspicture}(0.1,-0.4)(3,1.1)
\psset{xunit=1in,yunit=1in}
\pscircle[linewidth=1pt,origin={0.2,0},fillstyle=none,fillcolor=lightgray](0,0){0.3}
\rput(0.2,0){$\pmb{b}$}
\pscircle[linewidth=1pt,origin={0.2,0.5},fillstyle=none,fillcolor=lightgray](0,0){0.3}
\rput(0.2,0.5){$\pmb{a}$}
\pscircle[linewidth=1pt,origin={1,0.25},fillstyle=none,fillcolor=lightgray](0,0){0.3}
\rput(1,0.25){$\pmb{c}$}
\pscircle[linewidth=1pt,origin={1.8,0},fillstyle=none,fillcolor=lightgray](0,0){0.3}
\rput(1.8,0){$\pmb{e}$}
\pscircle[linewidth=1pt,origin={1.8,0.5},fillstyle=none,fillcolor=lightgray](0,0){0.3}
\rput(1.8,0.5){$\pmb{d}$}
\psline[linewidth=1pt,arrowsize=1.5pt 4,linecolor=black]{<-}(0.3,0)(0.9,0.2)
\psline[linewidth=1pt,arrowsize=1.5pt 4,linecolor=black]{<-}(0.3,0.5)(0.9,0.3)
\psline[linewidth=1pt,arrowsize=1.5pt 4,linecolor=black]{<-}(1.1,0.2)(1.7,0)
\psline[linewidth=1pt,arrowsize=1.5pt 4,linecolor=black]{<-}(1.1,0.3)(1.7,0.5)
\rput[ml](0.2,-0.25){$\Phi=0.4\,\,\,\,\,\,\gamma=0.1\,\,\,\,\,\,E=5$}
\end{pspicture}
\caption{\label{ex1}A homogeneous network used in the discussion in Section~\ref{comparison-sec}.}
\end{wrapfigure}
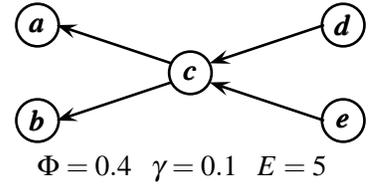

Models for propagation of beneficial or harmful attributes have been investigated in the past in several other contexts such as  
influence maximization in social networks~\cite{kempe1,CWY09,Chen08,borodin}, disease spreading in urban networks \cite{eubank04,coelho08,eubank05}, 
percolation models in physics and mathematics~\cite{SA94} and other types of contagion spreads~\cite{klein1,klein2}.
However, the model for shock propagation in financial network discussed in this paper is {\em fundamentally} very different from all these 
models. For example, the cascade models of failure considered in~\cite{klein1,klein2} are probabilistic models of failure propagation of a more generic nature, 
and thus not very useful to study failure propagation via interlocked balance sheets of financial institutions (as is the case in OTC derivatives markets).
Some distinguishing features of our model include: 

\vspace*{0.1in}
\noindent
$\pmb{(\mathfrak{a})}$ 
Almost all of these models include a trivial solution in which the attribute spreads to the entire network if we inject each node
individually with the attribute. This is not the case with our model: 
{\em a node may not fail when shocked, and the network may not be dead if all nodes are shocked}. 
For example, consider the network in~\FI{ex1}$\!\!${\bf (i)}. 
\begin{itemize}
\item
Suppose that all the nodes are shocked. Then, the following events happen.
\begin{itemize}
\item
Node $a$ (and similarly node $b$) fails at $t=1$ since $\Phi\left(\din(a)+\frac{E}{5}\right)>\gamma\left(\din(a)+\frac{E}{5}\right)$.

\item
Node $c$ also fails at $t=1$ since $\Phi\left(\din(c)-\dout(c)+\frac{E}{5}\right)=0.4>\gamma\left(\din(c)+\frac{E}{5}\right)=0.3$.

\item
Node $d$ (and similarly node $e$) do not fail at $t=1$ since 
$\Phi\left(-\dout(d)+\frac{E}{5}\right)=0<\gamma\times\frac{E}{5}=0.1$ and its equity stays at $0.1-0=0.1$.

\item
At $t=2$, node $d$ (and similarly node $e$) receives a shock from node $c$ of the amount $\frac{0.4-0.3}{2}=0.05<0.1$. 
Thus, nodes $d$ and $e$ do not fail. Since no new nodes fail during $t>2$, the network does not become dead.
\end{itemize}
\item
However, suppose that only nodes $a$ and $b$ are shocked. Then, the following events happen.
\begin{itemize}
\item
Node $a$ (and similarly node $b$) fails at $t=1$ since $\Phi\left(\din(a)+\frac{E}{5}\right)\!\!=\!0.8>\gamma\left(\din(a)+\frac{E}{5}\right)\!=\!0.2$.

\item
At $t=2$, node $c$ receives a shock of the amount $2\times (0.8-0.2)=1.2>\gamma\left(\din(c)+\frac{E}{5}\right)=0.3$.
Thus, node $c$ fails at $t=2$.

\item
At $t=3$, node $d$ (and similarly node $e$) receives a shock of the amount $\frac{1.2-0.3}{2}=0.45>\gamma\times\frac{E}{5}=0.1$.
Thus, both these nodes fail at $t=3$ and the entire network is dead.
\end{itemize}
\end{itemize}
As the above example shows, if shocking a subset of nodes makes a network dead, adding more nodes to this subset may {\em not} necessarily lead to the death 
of the network, and the stability measure is {\em neither monotone nor sub-modular}.
Similarly, it is also possible to exhibit banking networks such that to make the entire network fail:
\begin{itemize}
\item 
it may be necessary to shock a node even if it does not fail since shocking such a node ``weakens'' it by decreasing its equity, and 

\item 
it may be necessary to shock a node even if it fails due to shocks given to other nodes. 
\end{itemize}

\vspace*{0.1in}
\noindent
$\pmb{(\mathfrak{b})}$ 
The complexity of the computational aspects of many previous attribute propagation models arise due to the presence of cycles in 
the graph; for example, see~\cite{Chen08} for polynomial-time solutions of some of these problems when the underlying graph does not 
have a cycle. In contrast, our computational problems are may be hard {\em even when the given graph is acyclic}; instead, a key 
component of computational complexity arises due to two or more directed paths sharing a node.

\begin{table}[htbp]
{\small
\hspace*{-0.7in}
\scalebox{0.82}[0.82]{
\begin{tabular}{lccl} 
\cmidrule{2-4}
\multicolumn{2}{c}{\bf \begin{tabular}{c} \hspace*{0.5in} Network type, \\ \hspace*{0.5in} result type \end{tabular}} & \bf \begin{tabular}{c} Stability $\pmb{\vi^\ast(G,T)}$ \\ \hline bound, assumption (if any), \\ corresponding theorem \end{tabular} & \bf \begin{tabular}{c} Dual Stability $\pmb{\dvi^\ast(G,T,\kappa)}$ \\  \hline bound, assumption (if any), \\ corresponding theorem \end{tabular} \\ \cmidrule{2-4}
\multirow{8}*{{\bf \begin{tabular}{c} Homo- \\ geneous \end{tabular} }} & \bf \begin{tabular}{c}$\mathbf{T=2}$ \\ approximation hardness\end{tabular}    & \begin{tabular}{c} $\mathbf{(1-\eps)\,\ln n}$, \\ $\mathbf{\NP\not\subseteq\mathsf{DTIME}}\mathbf{\left(n^{\log\log n}\right)}$, \bf Theorem~\ref{th1} \end{tabular}           &   \\ \cmidrule{2-4}
                           & \begin{tabular}{c} $\mathbf{T=2}$, \bf approximation ratio \end{tabular} & \begin{tabular}{c} ${\mathrm{O}\left(\log \left( \dfrac{n\,\,\Phi\,\,E}{\gamma\,\,(\Phi-\gamma)\,\,\left| E-\Phi \right|} \right)\right)}$, \bf Theorem~\ref{log1} \end{tabular} &                                                  \\ \cmidrule{2-4}
                           & \begin{tabular}{c} {\bf Acyclic}, $\mathbf{\forall\,\,T>1}$, \\ \bf approximation hardness \end{tabular}  & \begin{tabular}{c} \bf $\mathbf{\apx}$-hard, Theorem~\ref{apx1} \end{tabular}    & \begin{tabular}{c} $\pmb{\left(1-\e^{-1}+\eps\right)}^{-1}$, \\ $\mathbf{\mathsf{P}\neq\NP}$, \bf Theorem~\ref{dual-np}(a) \end{tabular} \\ \cmidrule{2-4}
                           & \begin{tabular}{c} {\bf In-arborescence}, \\ $\mathbf{\forall\,\,T>1}$, \bf exact solution \end{tabular} & \begin{tabular}{c} $\mathbf{\mathrm{O}\left(n^2\right)}$ \bf time, every node fails \\ \bf when shocked, Theorem~\ref{poly1} \end{tabular} & \begin{tabular}{c} $\mathbf{\mathrm{O}\left(n^3\right)}$ \bf time, every node fails \\ \bf when shocked, Theorem~\ref{dual-np}(b) \end{tabular} \\ \cmidrule{1-4}
\multirow{8}*{{\bf \begin{tabular}{c} Hetero- \\ geneous \end{tabular} }} & \begin{tabular}{c} {\bf Acyclic}, $\mathbf{\forall\,\,T>1}$, \\ \bf approximation hardness \end{tabular} & \begin{tabular}{c} $\mathbf{(1-\eps)\,\ln n}$, $\,\,\mathbf{\NP\not\subseteq\mathsf{DTIME}}\mathbf{\left(n^{\log\log n}\right)}$, \\ \bf Theorem~\ref{hetero-thm1} \end{tabular}             & \begin{tabular}{c} $\pmb{{\left(1-\e^{-1}+\eps\right)}^{-1}}$, \\ $\mathbf{\mathsf{P}\neq\NP}$, \bf Theorem~\ref{dual-np}(a) \end{tabular} \\ \cmidrule{2-4}
                                 & \begin{tabular}{c} {\bf Acyclic}, $\mathbf{T=2}$, \bf approximation hardness \end{tabular} &   & \begin{tabular}{c} $\mathbf{n^\delta}$, \bf assumption ($\mathbf{\star}$)$^{\pmb{\dagger}}$, Theorem~\ref{dual-np-2} \end{tabular} \\ \cmidrule{2-4}
                           & \begin{tabular}{c} {\bf Acyclic}, $\mathbf{\forall\,\,T>3}$, \\ \bf approximation hardness \end{tabular} & \begin{tabular}{c} $\mathbf{2^{\log^{1-\eps}n}}$, $\mathbf{\NP\not\subseteq\mathsf{DTIME}}\mathbf{(n^{\,\mathrm{poly}(\log n)})}$, \\ \bf Theorem~\ref{hetero-thm2} \end{tabular} & \\ \cmidrule{2-4}
                           & \begin{tabular}{c} {\bf Acyclic}, $\mathbf{T=2}$, \\ {\bf approximation ratio}$^{\,\pmb{\ddagger}}$ \end{tabular} & 
\begin{tabular}{c}
$\displaystyle{\mathrm{O}\left( \log \dfrac{n\,\,\overline{E}\,\,\overline{w_{\max}}\,\,\overline{w_{\min}}\,\,\overline{\alpha_{\max}}}
{\Phi\,\,\gamma\,\, (\Phi-\gamma)\,\,\underline{E}\,\,\underline{w_{\min}}\,\,\underline{\alpha_{\min}}\,\,\underline{w_{\max}} } \right)}$,
\\
\bf \hspace*{-1.5in}Theorem~\ref{log2}
\end{tabular}
&                                                   \\ 
\cmidrule{2-4}
\end{tabular}
}

{\normalsize \hspace*{0.5in}{$^{\pmb{\ddagger}}$See Theorem~\ref{log2} for definitions of some parameters in the approximation ratio.} }
\newline
{\normalsize \hspace*{0.5in}{$^{\pmb{\dagger}}$See page~\pageref{star-ref} for statement of assumption {\bf ($\mathbf{\star}$)}, which is weaker than the assumption $\mathsf{P}\neq\NP$.} }
\vspace*{-12pt}
\caption{\label{s1}A summary of our results; $\eps>0$ is any arbitrary constant and $0<\delta<1$ is some constant.}
}
\end{table}

\section{Overview of Our Results and Their Implications on Banking Networks}

Table~\ref{s1} summarizes our results,
where the notation $\mathrm{poly}\left(x_1,x_2,\dots,x_k\right)$ denotes a constant-degree polynomial in variables 
$x_1,x_2,\dots,x_k$.
Our results for heterogeneous networks show that the problem of computing stability indices for them is harder than that for homogeneous 
networks, as one would naturally expect. 

\subsection{Brief Overview of Proof Techniques}

\subsubsection{Homogeneous Networks, $\pmb{\mbox{\nam}}$}

\paragraph{$\pmb{T=2}$, approximation hardness and approximation algorithm} 
The reduction for approximation hardness is from a corresponding inapproximability result for the dominating set problem for graphs.
The logarithmic approximation {\em almost} matches the lower bound.
Even though this algorithmic problem can be cast as a covering problem, 
one {\em cannot} explicitly enumerate {\em exponentially many} covering sets in polynomial time.
Instead, we reformulate the problem to that of computing an optimal solution of a polynomial-size 
integer linear programming ({\sf ILP}), and then use the greedy approach of~\cite{D82} for approximation. 
A careful calculation of the size of the coefficients of the {\sf ILP} ensures that we have the desired approximation bound.

\paragraph{Any $\pmb{T>1}$, approximation hardness and exact algorithm} 
The $\apx$-hardness result, which holds even if the degrees of all nodes are {\em small}
constants, is via a reduction from the node cover problem for $3$-regular graphs.
Technical complications in the reduction arise from making sure that the generated graph instance of \nam\ is {\em acyclic}, 
no new nodes fail for any $t>3$, but the network can be dead without each node being individually shocked.  
If the network is a rooted in-arborescence and every node can be individually shocked to fail, then 
we design an $O\left(n^2\right)$ time {\em exact} algorithm via dynamic programming; 
as a by product it also follows that the value of the stability index of this kind of network with {\em bounded} 
node degrees is {\em large}. 

\subsubsection{Homogeneous Networks, $\pmb{\mbox{\dnam}}$}

\paragraph{Any $\pmb{T}$, approximation hardness and exact algorithm} 
For hardness, we translate a lower bound for the {\em maximum coverage} problem~\cite{F98}.
The reduction relies on the fact that in dual stability measure every node of the network need {\em not} fail. 
If the given graph is a rooted in-arborescence and every node can be individually shocked to fail, we provide an $O\left(n^3\right)$ time exact 
algorithm via dynamic programming.

\subsubsection{Heterogeneous Networks, $\pmb{\mbox{\nam}}$}

\paragraph{Any $\pmb{T}$, approximation hardness} 
The reduction is from a corresponding inapproximability result for the minimum set covering problem.
Unlike homogeneous networks, unequal shares of the total external assets by various banks allows us to 
encode an instance of set cover by ``equalizing'' effects of nodes.

\paragraph{$\pmb{T=2}$}
The {\bf approximation algorithm} 
uses linear program in Theorem~\ref{log1} with more careful calculations.

\paragraph{Any $\pmb{T>2}$, approximation hardness} 
This stronger poly-logarithmic inapproximability result than that in Theorem~\ref{hetero-thm1} is obtained by a reduction from 
{\sc Minrep}, a graph-theoretic abstraction of two prover multi-round protocol 
for any problem in $\NP$. Many technical complications in the reduction, 
culminating to a set of $22$ symbolic linear equations between the parameters that we must satisfy. Intuitively,
the two provers in \minrep\ correspond to two nodes in the network that cooperate to fail to another specified set of nodes.

\subsubsection{Heterogeneous Networks, $\pmb{\mbox{{{{\sc Dual-Stab$_{2,\Phi,\kappa}$}}}}}$,}

\vspace*{-23pt}
\paragraph{\hspace*{260pt}approximation hardness} 
The reduction for this 
stronger inapproximability result is from the {\em densest hyper-graph} problem.

\subsection{Implications of Our Results on Banking Networks}

\subsubsection{Effects of Topological Connectivity}

Though researchers agree that the connectivity of banking networks affects
its stability~\cite{AG2000,GK08}, the conclusions drawn are mixed, namely some researchers 
conclude that lesser connectivity implies more susceptibility to contagion whereas other researchers 
conclude in the opposite. Based on our results and their proofs, we found that topological connectivity does play a significant role in stability of the network 
in the following complex manner.

\medskip
\hangindent=11pt
\hangafter=0
\noindent
{\bf Even acyclic networks display complex stability behavior}:
Sometimes a cause of the instability of a banking network is attributed to {\em cyclical} dependencies of borrowing and lending mechanisms among major banks, 
\EG, banks $v_1$, $v_2$ and $v_3$ borrowing from banks $v_2$, $v_3$ and $v_1$, respectively.
Our results show that computing the stability measures may be difficult even without the presence of such cycles. Indeed, larger inapproximability 
results, especially for heterogeneous networks, are possible because slight change in network parameters can cause a large change in the
stability measure. On the other hand, acyclic small-degree rooted in-arborescence networks exhibit higher values of the stability measure, \EG, 
if the maximum in-degree of any node in a rooted in-arborescence is $5$ and the shock parameter $\Phi$ is no more than twice the value of the percentage
of equity to assets $\gamma$, then by Theorem~\ref{poly1} $\vi^\ast(G,T)>0.1$. 

\medskip
\hangindent=11pt
\hangafter=0
\noindent
{\bf Intersection of borrowing chains may cause lower stability}:
By a {\em borrowing chain} we mean a directed path from a node $v_1$ to another node $v_2$, indicating that bank $v_2$ effectively borrowed from bank $v_1$
through a sequence of successive intermediaries. Now, assume that there is another directed path from $v_1$ to another node $v_3$. Then, 
failure of $v_2$ and $v_3$ propagates the resulting shocks to $v_1$ and, if the shocks arrive at the same step, then the total shock
received by bank $v_1$ is the addition of these two shocks, which in turn passes this ``amplified'' shock to other nodes in the network.

\medskip
\noindent
Based on these kinds of observations, it can be reasonably inferred that 
homogeneous networks with topologies more like a small-degree in-arborescence have higher stabilities, whereas 
networks of other types of topologies may have lower stabilities even if the topologies are acyclic.
For example, as we observe later, 
when $\mathrm{deg}_{\mathrm{in}}^{\max}=3$, $\gamma=0.1$ and $\Phi=0.15$, we get $\vi^\ast(G,T)>0.22$ and 
the network cannot be put to death without shocking more than $22\%$ of the nodes. 

\subsubsection{Effects of Ratio of External to Internal Assets ($E/I$) and percentage of equity to assets ($\gamma$) for Homogeneous Networks}

As our relevant results and their proofs show, lower values of $E/I$ and $\gamma$ may cause the network stability to be extremely sensitive with respect
to variations of other parameters of a homogeneous network. For example, in the proof of Theorem~\ref{th1} we have 
$\lim_{n\to\infty} \nicefrac{E}{I}=\lim_{n\to\infty} \gamma=0$, leading to variation of the stability index by a logarithmic factor; however, in the proof of Theorem~\ref{apx1}
we have $\nicefrac{E}{I}=0.25$ and $\gamma=0.23$ leading to much smaller variation of the stability index. 

\subsubsection{Homogeneous vs. Heterogeneous Networks}

Our results and proofs show that heterogeneous networks of banks with diverse equities tend to exhibit wider fluctuations of the
stability index with respect to parameters, \EG, Theorem~\ref{hetero-thm2} shows a polylogarithmic fluctuation even if
the ratio $E/I$ is large. 

\subsubsection{Further Empirical Study}

Subsequent to writing this paper, DasGupta and Kaligounder in a separate article~\cite{arxiv2} performed a thorough {\em empirical} analysis of the stability measure 
over more than $700000$ combinations of networks types and parameters, and uncovered many interesting insights into the relationships of the stability with other 
relevant parameters of the network, such as:
\begin{description}
\item[\bf Effect of uneven distribution of assets:]
Networks where all banks have roughly the {\em same} external assets are more stable
over similar networks in which fewer banks have a disproportionately higher external assets compared to the remaining banks, and 
failures of those banks with {\em higher} assets contribute more damage to the stability of the network.
Furthermore, networks in which fewer banks have a disproportionately higher external assets compared to the remaining banks
has a minimal instability even if their equity to asset ratio is large and comparable to loss of external assets.
This is not the case for networks where all banks have roughly the same external assets.
Thus, in summary, they concluded that banks with disproportionately large external assets (``banks that are too big'') affect the stability of
the entire banking network in an {\em adverse} manner.

\item[Effect of connectivity:]
For banking networks where all banks have roughly the same amount of external assets, higher connectivity leads to {\em lower} stability.
In contrast, for banking networks in which few banks have disproportionately higher external assets compared to the remaining banks, 
higher connectivity leads to {\em higher} global stability.

\item[Correlated versus random failures:]
{\em Correlated} initial failures of banks causes more damage to the entire banking network as opposed to just {\em random} initial failures of banks.

\item[Phase transition properties of global stability:]
The global stability exhibits several sharp {\em phase transitions} for various banking networks within certain parameter ranges. 
\end{description}

\section{Preliminary Observations on Shock Propagation}

\begin{proposition}\label{obs1}
Let $\langle G=(V,F),\gamma,\beta,E \rangle$ be the given (homogeneous or heterogeneous) banking network. Then, the following are true:
\begin{description}
\item[(a)]
If $\dout(v)=0$ for some $v\in V$, then node $v$ must be given a shock (and, must fail due to this shock) 
for the entire network to fail.

\item[(b)]
Let $\alpha$ be the number of edges in the longest directed simple path in $G$. Then, no new
node fails at any time $t>\alpha$.

\item[(c)]
We can assume without loss of generality that $G$ is weakly connected, \IE, the un-oriented version of $G$ is connected.
\end{description}
\end{proposition}

\begin{proof}~\\
\noindent
{\bf (a)}
Since $\dout(v)=0$, no part of any shock given to any other nodes in the network can reach $v$. Thus, the network of 
$v$, namely $c_v=\gamma\,a_v$ stays strictly positive (since $\gamma>0$) and node $v$ never fails.

\vspace*{0.1in}
\noindent
{\bf (b)}
Let $\mathrm{t}_{\,\mathrm{last}}$ be the latest time a node of $G$ failed, and let $V(t)$ be the set of nodes that failed at time $t=1,2,\dots,\mathrm{t}_{\,\mathrm{last}}$. 
Then, $V(1),V(2),\dots,V(\mathrm{t}_{\,\mathrm{last}})$ is a partition of $V$. For every $i=1,2,\dots,\mathrm{t}_{\,\mathrm{last}}-1$, 
add directed edges $(u,v)$ from a node $u\in V(i)$ to a node $v\in V(i+1)$ if $u$ was last
node that transmitted any part of the shock to $v$ before $v$ failed. Note that $(u,v)$ is also an edge of $G$ and for every node $v\in V(i+1)$ there must
be an edge $(u,v)$ for some node $u\in V(i)$. Thus, $G$ has a path of length at least $\mathrm{t}_{\,\mathrm{last}}$.

\vspace*{0.1in}
\noindent
{\bf (c)}
This holds since otherwise the stability measures can be computed separately on each weakly connected component.
\end{proof}

\section{Homogeneous Networks, $\pmb{\mbox{{{{\sc Stab$_{2,\Phi}$}}}}}$, Logarithmic Inapproximability} 

\begin{theorem}\label{th1}
$\vi^\ast(G,2)$ cannot be approximated in polynomial time within a factor of $(1-\eps)\ln n$, for any constant $\eps>0$,
unless $\NP\subseteq\mathsf{DTIME}\left(n^{\log\log n}\right)$.
\end{theorem}

\begin{proof}
The {\em dominating set} problem for an undirected graph (DOMIN-SET) is defined as follows: 
{\em given an undirected graph $G=(V,F)$ with $n=|V|$ nodes, 
find a minimum cardinality subset of nodes $V'\subset V$ such that every node in $V\setminus V'$ is 
incident on at least one edge whose other end-point is in $V'$.}
It is known that DOMIN-SAT is equivalent to the minimum set-cover problem under L-reduction~\cite{BM84}, and thus
cannot be approximated within a factor of $(1-\eps)\ln n$ unless 
$\NP\subseteq\mathsf{DTIME}\left(n^{\log\log n}\right)$~\cite{F98}.

Consider an instance $G=(V,F)$ of DOMIN-SET with $n$ nodes and $m$ edges, and let $\opt$ denote the size of 
an optimal solution for this instance.
Our (directed) banking network $\overrightarrow{G}=(\overrightarrow{V},\overrightarrow{F})$ is obtained from $G$ by replacing each undirected edge $\{u,v\}$
by two directed edges $(u,v)$ and $(v,u)$. 
Thus we have $0<\din(v)=\dout(v)<n$ for every node $v\in V$.
We set the global parameters as follows: 
$E=10\,n$, $\gamma=n^{-2}$ and $\Phi=1$.

For a node $v$, let $\nbr(v)=\left\{\,u\,|\,\{u,v\}\in E\,\right\}$ be the set of neighbors of $v$ in $G$.
We claim that if a node $v$ is shocked at time $t=1$, then all nodes in 
in $\{v\}\cup \nbr(v)$ fail at time $t=2$. Indeed, suppose that $v$ is shocked at $t=1$. Then, $v$ surely fails
because
\begin{gather*}
\Phi\,e_v
=
\din(v)-\dout(v)+\frac{E}{n}
=
10
>
\frac{2}{n}
>
\frac{\din(v)+\frac{E}{n}}{n^2}
=
\gamma\,a_v
\end{gather*}
Now, consider $t=2$ and consider a node $v$ such that $v$ has not failed but a node $u\in\nbr(v)$ failed at time $t=1$.
Then, node $v$ surely fails because
\begin{gather*}
s_{v,2}
\geq
\frac{\min\{s_{u,1}-c_u,b_u\}}{\din(u,2)}
=
\frac{\min\{\Phi\,e_u-\gamma\,a_u,\din(u)\}}{\din(u)}
>
\min\left\{\, \frac{10-\frac{2}{n}}{\din(u)} ,\,1 \, \right\}
\\
>
\frac{2}{n}
>
\frac{\din(v)+\frac{E}{n}}{n^2}
=
\gamma\,a_v
\end{gather*}
Thus, we have a $1$--$1$ correspondence between the solutions of DOMIN-SET
and death of $\overrightarrow{G}$, namely $V'\subset V$ is a solution of DOMIN-SET if and 
only if shocking the nodes in $V'$ makes $\overrightarrow{G}$ fail at time $t=2$.
\end{proof}

\section{Homogeneous Networks, $\pmb{\mbox{{{{\sc Stab$_{2,\Phi}$}}}}}$, Logarithmic Approximation} 

\begin{theorem}\label{log1}
${\mbox{{{{\sc Stab$_{2,\Phi}$}}}}}$ admits a polynomial-time algorithm with approximation ratio
\linebreak
$\mathrm{O}\left(\log \left( \dfrac{n\,\,\Phi\,\,E}{\gamma\,\,(\Phi-\gamma)\,\,\left| E-\Phi \right|} \right)\right)$.
\end{theorem}

\begin{proof}
Suppose that $\Phi e_u<0$ for some node $u\in V$. Then, there exists an optimal solution 
in which we do not shock the node $u$. Indeed, if $u$ was shocked, the equity of $u$
increases from $c_u$ to $c_u + |\,\Phi\,e_u\,|$ and $u$ does not propagate any shock to
other nodes. Thus, if $u$ still fails at $t=2$, then it also fails at $t=2$ if it was not shocked.

Let $\vs$ denote the set of nodes that we will select for shocking, and, 
for every node $v\in V$, let $\delta_{\,v,u}$ be defined as: 
$\delta_{\,v,u}
=\left\{
\begin{array}{ll}
\max\{\,0,\,\Phi\,e_v\},  & \mbox{if $u=v$} \\
\dfrac{\min\{\Phi\,e_v-c_v,\,b_v\}}{\din(v)}, & \mbox{if $\Phi\,e_v>c_v$ and $(u,v)\in F$} \\
0,                                           & \mbox{otherwise} \\
\end{array}
\right. \!\!\!\!\!\!\!.$
Then, our problem reduces to a covering problem of the following type: 

\begin{quote}
{\em 
\hspace*{-0.3in}
find a minimum cardinality subset $\vs\subseteq V$
such that, for every node $u$, $\sum_{v\in \vs}\!\!\delta_{\,v,u}>c_u$. 
}
\end{quote}

\noindent
Note that we cannot even explicitly enumerate, for a node $u\in V$, all subsets $V'\subseteq V\setminus\{u\}$
such that $\sum_{v\in V'}\!\!\delta_{\,v,u}>c_u$, since there are exponentially many such subsets.
Let the binary variable $x_v\in\{0,1\}$ be the indicator variable for a node $v\in V$ for inclusion in $\vs$.
However, we can reformulate our problem as the following integer linear programming problem:
\begin{gather}
\hspace*{-1.05in} \mbox{minimize } \sum_{v\in V} x_v \notag \\
\mbox{subject to } \,\forall u\in V\colon \sum_{v\in V}\!\!\delta_{\,v,u}x_v>c_u \label{ip1} \\
\hspace*{1in} x_u \in\{0,1\} \notag
\end{gather}
Let $\displaystyle\zeta=\min_{u\in V} \big\{ \, \min_{v\in V} \{\delta_{\,u,v}\},\, c_u\,  \big\}$.
We can rewrite each constraint $\displaystyle\sum_{v\in V}\!\!\delta_{\,v,u}x_v>c_u$ as 
$\displaystyle\sum_{v\in V}\!\!\frac{\delta_{\,v,u}}{\zeta}x_v>\frac{c_u}{\zeta}$ 
to ensure that every non-zero entry is at least $1$.
Since the coefficients of the constraints and the objective function are all positive real numbers,~\eqref{ip1} can be approximated by the 
greedy algorithm described in~\cite[Theorem 4.1]{D82} with an approximation ratio of 
$2 + \ln n + \ln \left( \max_{v\in V} \left\{ \sum_{u\in V} \frac{\delta_{\,v,u}}{\zeta} \right\} \right)$.
Now, observe that: 
\begin{gather*}
\min_{\substack{u\in V \\ \delta_{\,u,u}>0}} \left\{ \delta_{\,u,u} \right\}
=
\min_{\substack{u\in V \\ \delta_{\,u,u}>0}} \left\{  \Phi\, \left( \din(u)-\dout(u) + \frac{E}{n} \right) \right\}
=
\Omega \left( \, \frac{\big|\,E-\Phi\,\big|}{n} \, \right) 
\end{gather*}
\vspace*{-0.3in}
\begin{gather*}
\min_{u\in V} \min_{\substack{v\in V \\ \delta_{\,u,v}>0 } } \left\{ \delta_{\,u,v} \right\} 
=
\min_{u\in V} \min_{\substack{v\in V \\ \Phi\,e_v>c_v} } \left\{ (\Phi-\gamma) \left( 1+\frac{E}{\din(v)} \right) - \Phi \frac{\dout(v)}{\din(v)} \right\} 
=
\Omega \left( \frac{(\Phi-\gamma)\,E}{n} \right)
\end{gather*}
\vspace*{-0.3in}
\begin{gather*}
\min_{u\in V} \{c_u\}
=
\min_{u\in V} \left\{ \gamma\,\left( \din(u)+\frac{E}{n} \right) \right\}
=
\Omega \left( \frac{\gamma\,E}{n} \right)
\\
\zeta=\min \big\{ \, \min_{u\in V}\min_{v\in V} \{\delta_{\,u,v}\},\, \min_{u\in V} \{c_u\}\,  \big\}
=
\Omega \left( \min \left\{ \frac{\big|\,E-\Phi\,\big|}{n}, \, \frac{(\Phi-\gamma)\,E}{n}, \, \frac{\gamma\,E}{n} \right\} \right)
\\
\max_{v\in V} \sum_{u\in V} \delta_{\,v,u} 
\leq n\,
\max_{u\in V} \left\{ 
(\Phi-\gamma) \left( 1+\frac{E}{\din(u)} \right) - \Phi \frac{\dout(u)}{\din(u)}
\right\} 
=
\mathrm{O} \left( \,n\,(\Phi-\gamma)\,E \, \right)
\end{gather*}
and thus, 
$\max_{v\in V} \left\{ \sum_{u\in V} \frac{\delta_{\,v,u}}{\zeta} \right\} = \mathrm{O} \left( \mathrm{poly} \left(n , \frac{\Phi}{\gamma} , \frac{1}{\Phi-\gamma}, 
\frac{E}{\big| E - \Phi\big|} \right) \right)$,
giving the approximation bound.
\end{proof}

\section{Homogeneous Networks, $\pmb{\mbox{\nam}}$, any $T$, $\apx$-hardness} 

\begin{theorem}\label{apx1}
For any $T$, computing $\vi^\ast(G,T)$ is $\apx$-hard even if the banking network $G$ 
is a DAG with 
$\din(v)\leq 3$ and $\dout(v)\leq 2$ for every node $v$.
\end{theorem}

\begin{proof}
We reduce the 
$3$-{\sf MIN}-{\sf NODE}-{\sf COVER} 
problem to \nam. 
$3$-{\sf MIN}-{\sf NODE}-{\sf COVER} 
is defined as follows. 
We are given an undirected $3$-regular graph $G$, \IE, an undirected graph $G=(V,F)$ in which the degree of every node 
is exactly $3$ (and thus $|F|=1.5\,|V|$).
A valid solution (node cover) is a subset of nodes $V'\subseteq V$ such that every edge is incident to at least
one node in $V'$. The goal is then to find a node cover $V'\subseteq V$ such that $|V'|$ is {\em minimized}.
This problem is known to be $\apx$-hard~\cite{BK98}.

\begin{figure}[htbp]
\begin{pspicture}(-1,-3.5)(7,3)
\psset{xunit=1cm,yunit=1cm}
\pscircle[linewidth=1pt,origin={0,0},fillstyle=none,fillcolor=lightgray](0,0){0.15}
\rput(-0.3,0){$\pmb{v_1}$}
\pscircle[linewidth=1pt,origin={1,-1},fillstyle=none,fillcolor=lightgray](0,0){0.15}
\rput(1,-1.3){$\pmb{v_2}$}
\pscircle[linewidth=1pt,origin={2,-1},fillstyle=none,fillcolor=lightgray](0,0){0.15}
\rput(2,-1.3){$\pmb{v_3}$}
\pscircle[linewidth=1pt,origin={3,0},fillstyle=none,fillcolor=lightgray](0,0){0.15}
\rput(3.4,0){$\pmb{v_4}$}
\pscircle[linewidth=1pt,origin={2,1},fillstyle=none,fillcolor=lightgray](0,0){0.15}
\rput(2,1.3){$\pmb{v_5}$}
\pscircle[linewidth=1pt,origin={1,1},fillstyle=none,fillcolor=lightgray](0,0){0.15}
\rput(1,1.3){$\pmb{v_6}$}
\psline[linewidth=1pt,arrowsize=1.5pt 4,linecolor=black](0.1,-0.1)(0.9,-0.9)
\psline[linewidth=1pt,arrowsize=1.5pt 4,linecolor=black](1.1,-1)(1.9,-1)
\rput(1.5,-0.8){$\scriptscriptstyle\pmb{e_{2,3}}$}
\psline[linewidth=1pt,arrowsize=1.5pt 4,linecolor=black](2.1,-0.9)(2.9,-0.1)
\psline[linewidth=1pt,arrowsize=1.5pt 4,linecolor=black](2.1,0.9)(2.9,0.1)
\psline[linewidth=1pt,arrowsize=1.5pt 4,linecolor=black](1.1,1)(1.9,1)
\psline[linewidth=1pt,arrowsize=1.5pt 4,linecolor=black](0.1,0.1)(0.9,0.9)
\psline[linewidth=1pt,arrowsize=1.5pt 4,linecolor=black](0.1,0)(2.9,0)
\psline[linewidth=1pt,arrowsize=1.5pt 4,linecolor=black](1,-0.9)(1,0.9)
\psline[linewidth=1pt,arrowsize=1.5pt 4,linecolor=black](2,-0.9)(2,0.9)
\rput(1.5,-1.8){$\pmb{G=(V,F)}$}
\pscircle[linewidth=1pt,origin={6,-1},fillstyle=none,fillcolor=lightgray](0,0){0.15}
\rput(5.6,-1){$\pmb{u_1}$}
\pscircle[linewidth=1pt,origin={6,-3},fillstyle=none,fillcolor=lightgray](0,0){0.15}
\rput(6,-3.4){$\pmb{u_1'}$}
\psline[linewidth=1pt,arrowsize=1.5pt 4,linecolor=black]{->}(6,-1.15)(6,-2.85)
\pscircle[linewidth=1pt,origin={7,-1},fillstyle=none,fillcolor=lightgray](0,0){0.15}
\rput(6.6,-1){$\pmb{u_2}$}
\pscircle[linewidth=1pt,origin={7,-3},fillstyle=none,fillcolor=lightgray](0,0){0.15}
\rput(7,-3.4){$\pmb{u_2}'$}
\psline[linewidth=1pt,arrowsize=1.5pt 4,linecolor=black]{->}(7,-1.15)(7,-2.85)
\pscircle[linewidth=1pt,origin={8,-1},fillstyle=none,fillcolor=lightgray](0,0){0.15}
\pscircle[linewidth=1pt,origin={8,-3},fillstyle=none,fillcolor=lightgray](0,0){0.15}
\rput(7.6,-1){$\pmb{u_3}$}
\rput(8,-3.4){$\pmb{u_3'}$}
\psline[linewidth=1pt,arrowsize=1.5pt 4,linecolor=black]{->}(8,-1.15)(8,-2.85)
\pscircle[linewidth=1pt,origin={9,-1},fillstyle=none,fillcolor=lightgray](0,0){0.15}
\pscircle[linewidth=1pt,origin={9,-3},fillstyle=none,fillcolor=lightgray](0,0){0.15}
\rput(8.6,-1){$\pmb{u_4}$}
\rput(9,-3.4){$\pmb{u_4'}$}
\psline[linewidth=1pt,arrowsize=1.5pt 4,linecolor=black]{->}(9,-1.15)(9,-2.85)
\pscircle[linewidth=1pt,origin={10,-1},fillstyle=none,fillcolor=lightgray](0,0){0.15}
\pscircle[linewidth=1pt,origin={10,-3},fillstyle=none,fillcolor=lightgray](0,0){0.15}
\rput(9.6,-1){$\pmb{u_5}$}
\rput(10,-3.4){$\pmb{u_5'}$}
\psline[linewidth=1pt,arrowsize=1.5pt 4,linecolor=black]{->}(10,-1.15)(10,-2.85)
\pscircle[linewidth=1pt,origin={11,-1},fillstyle=none,fillcolor=lightgray](0,0){0.15}
\pscircle[linewidth=1pt,origin={11,-3},fillstyle=none,fillcolor=lightgray](0,0){0.15}
\rput(10.6,-1){$\pmb{u_6}$}
\rput(11,-3.4){$\pmb{u_6'}$}
\psline[linewidth=1pt,arrowsize=1.5pt 4,linecolor=black]{->}(11,-1.15)(11,-2.85)
\pscircle[linewidth=1pt,origin={5,1},fillstyle=none,fillcolor=lightgray](0,0){0.15}
\rput(5,1.3){$\pmb{e_{1,2}}$}
\psline[linewidth=1pt,arrowsize=1.5pt 4,linecolor=black]{->}(5,0.85)(5.9,-0.9)
\psline[linewidth=1pt,arrowsize=1.5pt 4,linecolor=black]{->}(5.05,0.9)(6.9,-0.9)
\pscircle[linewidth=1pt,origin={6,1},fillstyle=none,fillcolor=lightgray](0,0){0.15}
\rput(6,1.3){$\pmb{e_{1,4}}$}
\psline[linewidth=1pt,arrowsize=1.5pt 4,linecolor=black]{->}(6,0.85)(5.95,-0.9)
\psline[linewidth=1pt,arrowsize=1.5pt 4,linecolor=black]{->}(6.05,0.9)(8.9,-0.9)
\pscircle[linewidth=1pt,origin={7,1},fillstyle=none,fillcolor=lightgray](0,0){0.15}
\rput(7,1.3){$\pmb{e_{1,6}}$}
\psline[linewidth=1pt,arrowsize=1.5pt 4,linecolor=black]{->}(7,0.85)(6,-0.95)
\psline[linewidth=1pt,arrowsize=1.5pt 4,linecolor=black]{->}(7.05,0.9)(10.9,-0.9)
\pscircle[linewidth=1pt,origin={8,1},fillstyle=none,fillcolor=lightgray](0,0){0.15}
\rput(8,1.3){$\pmb{e_{2,3}}$}
\rput(8,1.7){$\scriptstyle{\{v_2,v_3\}}$}
\psline[linewidth=1pt,arrowsize=1.5pt 4,linecolor=black]{->}(8,0.85)(6.9,-0.9)
\psline[linewidth=1pt,arrowsize=1.5pt 4,linecolor=black]{->}(8.05,0.9)(7.9,-0.9)
\pscircle[linewidth=1pt,origin={9,1},fillstyle=none,fillcolor=lightgray](0,0){0.15}
\rput(9,1.3){$\pmb{e_{2,5}}$}
\psline[linewidth=1pt,arrowsize=1.5pt 4,linecolor=black]{->}(8.95,0.9)(6.95,-0.95)
\psline[linewidth=1pt,arrowsize=1.5pt 4,linecolor=black]{->}(9,0.95)(9.95,-0.95)
\pscircle[linewidth=1pt,origin={10,1},fillstyle=none,fillcolor=lightgray](0,0){0.15}
\rput(10,1.3){$\pmb{e_{3,4}}$}
\psline[linewidth=1pt,arrowsize=1.5pt 4,linecolor=black]{->}(9.95,0.9)(7.95,-0.95)
\psline[linewidth=1pt,arrowsize=1.5pt 4,linecolor=black]{->}(10,0.9)(8.95,-0.95)
\pscircle[linewidth=1pt,origin={11,1},fillstyle=none,fillcolor=lightgray](0,0){0.15}
\rput(11,1.3){$\pmb{e_{3,5}}$}
\psline[linewidth=1pt,arrowsize=1.5pt 4,linecolor=black]{->}(10.95,0.9)(8.1,-1)
\psline[linewidth=1pt,arrowsize=1.5pt 4,linecolor=black]{->}(11.05,0.9)(9.95,-0.95)
\pscircle[linewidth=1pt,origin={12,1},fillstyle=none,fillcolor=lightgray](0,0){0.15}
\rput(12,1.3){$\pmb{e_{4,5}}$}
\psline[linewidth=1pt,arrowsize=1.5pt 4,linecolor=black]{->}(11.95,0.9)(9.1,-1)
\psline[linewidth=1pt,arrowsize=1.5pt 4,linecolor=black]{->}(12.05,0.9)(10,-0.95)
\pscircle[linewidth=1pt,origin={13,1},fillstyle=none,fillcolor=lightgray](0,0){0.15}
\rput(13,1.3){$\pmb{e_{5,6}}$}
\psline[linewidth=1pt,arrowsize=1.5pt 4,linecolor=black]{->}(12.95,0.9)(10.1,-1)
\psline[linewidth=1pt,arrowsize=1.5pt 4,linecolor=black]{->}(13,0.95)(11.1,-1)
\rput(13.8,1.1){\bf\small sink}
\rput(13.8,0.75){\bf\small nodes}
\rput(13,-3.3){\bf\small super-source nodes}
\rput(4.5,-1.8){$\pmb{\overrightarrow{G}=(\overrightarrow{V},\overrightarrow{F})}$}
\end{pspicture}
\caption{\label{fig1}A $3$-regular graph $G=(V,F)$ and its corresponding banking network $\stackrel{\to}{G}=(\stackrel{\to}{V},\stackrel{\to}{F})$.}
\end{figure}
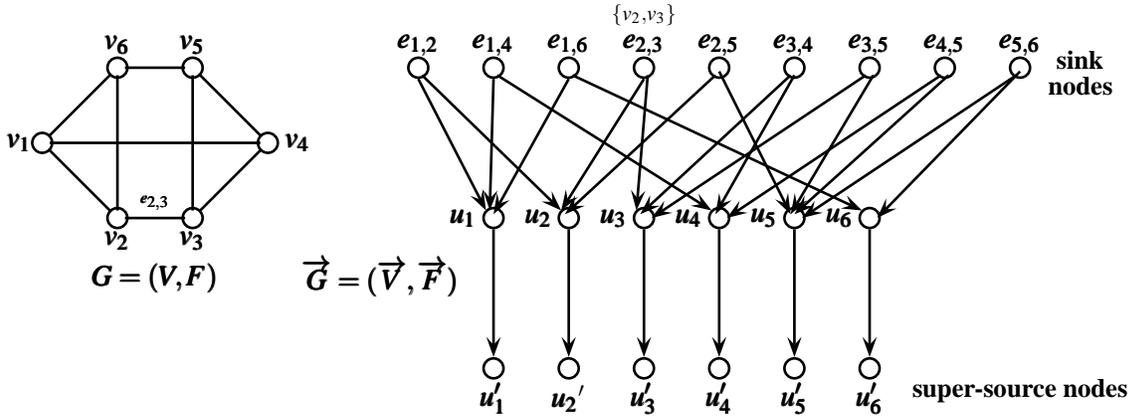

Given such an instance $G=(V,F)$ of 
$3$-{\sf MIN}-{\sf NODE}-{\sf COVER}, we construct an instance of 
the banking network $\overrightarrow{G}=(\overrightarrow{V},\overrightarrow{F})$ as follows:
\begin{itemize}
\item
For every node $v_i\in V$, we have two nodes $u_i,u_i'$ in $\overrightarrow{V}$, and a directed edge $(u_i,u_i')$.  
We refer to $u_i'$ as a ``super-source'' node.

\item
For every edge 
$\{v_i,v_j\}\in F$ with $i<j$, we have a (``sink'') node $e_{i,j}$ in $\overrightarrow{V}$ and two directed edges
$(e_{i,j},u_i)$ and $(e_{i,j},u_j)$ in $\overrightarrow{F}$.
For notational convenience, the node $e_{i,j}$ is also sometimes referred to as the node $e_{j,i}$.
\end{itemize}
Thus, $|\overrightarrow{V}|=3.5\,|V|$, and $|\overrightarrow{F}|=4\,|V|$.
See~\FI{fig1} for an illustration.
Observe that:
\begin{itemize}
\item
$\din\left(u_i\right)=3$ and $\dout\left(u_i\right)=1$ for all $i=1,2,\dots,|V|$. 

\item
$\din\left(u_i'\right)=1$ and $\dout\left(u_i'\right)=0$ for all $i=1,2,\dots,|V|$. 
Thus, by Proposition~\ref{obs1}(a), 
every node $u_i'$ must be shocked to make the network fail.

\pichskip{0.9cm}
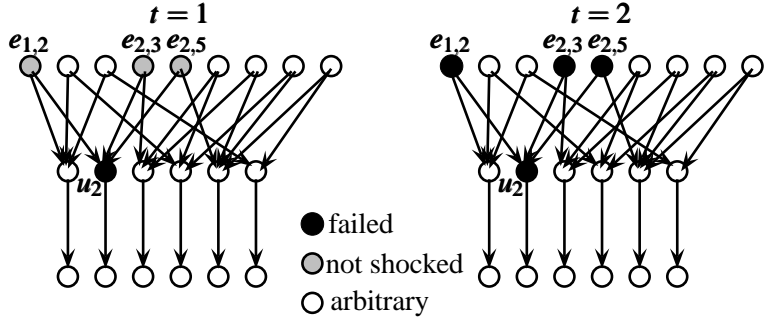
\piccaption{\label{fig10}Case~{\bf (III)}: if node $u_2$ is shocked then the nodes $e_{1,2},e_{2,3}$ and $e_{2,5}$ must fail at $t=2$.}
\parpic(0pt,0pt)[r][c]
{
\begin{pspicture}(1,-2.6)(6.5,1.45)
\psset{xunit=0.5cm,yunit=0.7cm}
\hspace*{-0.6in}
\rput(5,1.4){$\pmb{e_{1,2}}$}
\rput(9.2,1.4){$\pmb{e_{2,5}}$}
\pscircle[linewidth=1pt,origin={6,-1},fillstyle=none,fillcolor=lightgray](0,0){0.15}
\pscircle[linewidth=1pt,origin={6,-3},fillstyle=none,fillcolor=lightgray](0,0){0.15}
\psline[linewidth=1pt,arrowsize=1.5pt 4,linecolor=black]{->}(6,-1.15)(6,-2.85)
\pscircle[linewidth=1pt,origin={7,-1},fillstyle=solid,fillcolor=black](0,0){0.15}
\rput(6.6,-1.3){$\pmb{u_2}$}
\pscircle[linewidth=1pt,origin={7,-3},fillstyle=none,fillcolor=lightgray](0,0){0.15}
\psline[linewidth=1pt,arrowsize=1.5pt 4,linecolor=black]{->}(7,-1.15)(7,-2.85)
\pscircle[linewidth=1pt,origin={8,-1},fillstyle=none,fillcolor=lightgray](0,0){0.15}
\pscircle[linewidth=1pt,origin={8,-3},fillstyle=none,fillcolor=lightgray](0,0){0.15}
\psline[linewidth=1pt,arrowsize=1.5pt 4,linecolor=black]{->}(8,-1.15)(8,-2.85)
\pscircle[linewidth=1pt,origin={9,-3},fillstyle=none,fillcolor=lightgray](0,0){0.15}
\psline[linewidth=1pt,arrowsize=1.5pt 4,linecolor=black]{->}(9,-1.15)(9,-2.85)
\pscircle[linewidth=1pt,origin={9,-1},fillstyle=none,fillcolor=lightgray](0,0){0.15}
\pscircle[linewidth=1pt,origin={10,-3},fillstyle=none,fillcolor=lightgray](0,0){0.15}
\psline[linewidth=1pt,arrowsize=1.5pt 4,linecolor=black]{->}(10,-1.15)(10,-2.85)
\pscircle[linewidth=1pt,origin={10,-1},fillstyle=none,fillcolor=lightgray](0,0){0.15}
\pscircle[linewidth=1pt,origin={11,-3},fillstyle=none,fillcolor=lightgray](0,0){0.15}
\psline[linewidth=1pt,arrowsize=1.5pt 4,linecolor=black]{->}(11,-1.15)(11,-2.85)
\pscircle[linewidth=1pt,origin={11,-1},fillstyle=none,fillcolor=lightgray](0,0){0.15}
\pscircle[linewidth=1pt,origin={5,1},fillstyle=solid,fillcolor=lightgray](0,0){0.15}
\psline[linewidth=1pt,arrowsize=1.5pt 4,linecolor=black]{->}(5,0.85)(5.9,-0.9)
\psline[linewidth=1pt,arrowsize=1.5pt 4,linecolor=black]{->}(5.05,0.9)(6.9,-0.9)
\pscircle[linewidth=1pt,origin={6,1},fillstyle=none,fillcolor=lightgray](0,0){0.15}
\psline[linewidth=1pt,arrowsize=1.5pt 4,linecolor=black]{->}(6,0.85)(5.95,-0.9)
\psline[linewidth=1pt,arrowsize=1.5pt 4,linecolor=black]{->}(6.05,0.9)(8.9,-0.9)
\pscircle[linewidth=1pt,origin={7,1},fillstyle=none,fillcolor=lightgray](0,0){0.15}
\psline[linewidth=1pt,arrowsize=1.5pt 4,linecolor=black]{->}(7,0.85)(6,-0.95)
\psline[linewidth=1pt,arrowsize=1.5pt 4,linecolor=black]{->}(7.05,0.9)(10.9,-0.9)
\pscircle[linewidth=1pt,origin={8,1},fillstyle=solid,fillcolor=lightgray](0,0){0.15}
\rput(8,1.4){$\pmb{e_{2,3}}$}
\psline[linewidth=1pt,arrowsize=1.5pt 4,linecolor=black]{->}(8,0.85)(6.9,-0.9)
\psline[linewidth=1pt,arrowsize=1.5pt 4,linecolor=black]{->}(8.05,0.9)(7.9,-0.9)
\pscircle[linewidth=1pt,origin={9,1},fillstyle=solid,fillcolor=lightgray](0,0){0.15}
\psline[linewidth=1pt,arrowsize=1.5pt 4,linecolor=black]{->}(8.95,0.9)(6.95,-0.95)
\psline[linewidth=1pt,arrowsize=1.5pt 4,linecolor=black]{->}(9,0.95)(9.95,-0.95)
\pscircle[linewidth=1pt,origin={10,1},fillstyle=none,fillcolor=lightgray](0,0){0.15}
\psline[linewidth=1pt,arrowsize=1.5pt 4,linecolor=black]{->}(9.95,0.9)(7.95,-0.95)
\psline[linewidth=1pt,arrowsize=1.5pt 4,linecolor=black]{->}(10,0.9)(8.95,-0.95)
\pscircle[linewidth=1pt,origin={11,1},fillstyle=none,fillcolor=lightgray](0,0){0.15}
\psline[linewidth=1pt,arrowsize=1.5pt 4,linecolor=black]{->}(10.95,0.9)(8.1,-1)
\psline[linewidth=1pt,arrowsize=1.5pt 4,linecolor=black]{->}(11.05,0.9)(9.95,-0.95)
\pscircle[linewidth=1pt,origin={12,1},fillstyle=none,fillcolor=lightgray](0,0){0.15}
\psline[linewidth=1pt,arrowsize=1.5pt 4,linecolor=black]{->}(11.95,0.9)(9.1,-1)
\psline[linewidth=1pt,arrowsize=1.5pt 4,linecolor=black]{->}(12.05,0.9)(10,-0.95)
\pscircle[linewidth=1pt,origin={13,1},fillstyle=none,fillcolor=lightgray](0,0){0.15}
\psline[linewidth=1pt,arrowsize=1.5pt 4,linecolor=black]{->}(12.95,0.9)(10.1,-1)
\psline[linewidth=1pt,arrowsize=1.5pt 4,linecolor=black]{->}(13,0.95)(11.1,-1)
\rput(9,2){$\pmb{t=1}$}
\pscircle[linewidth=0pt,origin={12.5,-2},fillstyle=solid,fillcolor=black](0,0){0.15}
\rput(13.8,-2){failed}
\pscircle[linewidth=1pt,origin={12.5,-2.75},fillstyle=solid,fillcolor=lightgray](0,0){0.15}
\rput(14.7,-2.75){not shocked}
\pscircle[linewidth=1pt,origin={12.5,-3.5},fillstyle=solid,fillcolor=white](0,0){0.15}
\rput(14.3,-3.5){arbitrary}
\hspace*{5.6cm}
\rput(5,1.4){$\pmb{e_{1,2}}$}
\rput(9.2,1.4){$\pmb{e_{2,5}}$}
\pscircle[linewidth=1pt,origin={6,-1},fillstyle=none,fillcolor=lightgray](0,0){0.15}
\pscircle[linewidth=1pt,origin={6,-3},fillstyle=none,fillcolor=lightgray](0,0){0.15}
\psline[linewidth=1pt,arrowsize=1.5pt 4,linecolor=black]{->}(6,-1.15)(6,-2.85)
\pscircle[linewidth=1pt,origin={7,-1},fillstyle=solid,fillcolor=black](0,0){0.15}
\rput(6.6,-1.3){$\pmb{u_2}$}
\pscircle[linewidth=1pt,origin={7,-3},fillstyle=none,fillcolor=lightgray](0,0){0.15}
\psline[linewidth=1pt,arrowsize=1.5pt 4,linecolor=black]{->}(7,-1.15)(7,-2.85)
\pscircle[linewidth=1pt,origin={8,-1},fillstyle=none,fillcolor=black](0,0){0.15}
\pscircle[linewidth=1pt,origin={8,-3},fillstyle=none,fillcolor=lightgray](0,0){0.15}
\psline[linewidth=1pt,arrowsize=1.5pt 4,linecolor=black]{->}(8,-1.15)(8,-2.85)
\pscircle[linewidth=1pt,origin={9,-3},fillstyle=none,fillcolor=lightgray](0,0){0.15}
\psline[linewidth=1pt,arrowsize=1.5pt 4,linecolor=black]{->}(9,-1.15)(9,-2.85)
\pscircle[linewidth=1pt,origin={9,-1},fillstyle=none,fillcolor=lightgray](0,0){0.15}
\pscircle[linewidth=1pt,origin={10,-3},fillstyle=none,fillcolor=lightgray](0,0){0.15}
\psline[linewidth=1pt,arrowsize=1.5pt 4,linecolor=black]{->}(10,-1.15)(10,-2.85)
\pscircle[linewidth=1pt,origin={10,-1},fillstyle=none,fillcolor=lightgray](0,0){0.15}
\pscircle[linewidth=1pt,origin={11,-3},fillstyle=none,fillcolor=lightgray](0,0){0.15}
\psline[linewidth=1pt,arrowsize=1.5pt 4,linecolor=black]{->}(11,-1.15)(11,-2.85)
\pscircle[linewidth=1pt,origin={11,-1},fillstyle=none,fillcolor=lightgray](0,0){0.15}
\pscircle[linewidth=0pt,origin={5,1},fillstyle=solid,fillcolor=black](0,0){0.15}
\psline[linewidth=1pt,arrowsize=1.5pt 4,linecolor=black]{->}(5,0.85)(5.9,-0.9)
\psline[linewidth=1pt,arrowsize=1.5pt 4,linecolor=black]{->}(5.05,0.9)(6.9,-0.9)
\pscircle[linewidth=1pt,origin={6,1},fillstyle=none,fillcolor=lightgray](0,0){0.15}
\psline[linewidth=1pt,arrowsize=1.5pt 4,linecolor=black]{->}(6,0.85)(5.95,-0.9)
\psline[linewidth=1pt,arrowsize=1.5pt 4,linecolor=black]{->}(6.05,0.9)(8.9,-0.9)
\pscircle[linewidth=1pt,origin={7,1},fillstyle=none,fillcolor=lightgray](0,0){0.15}
\psline[linewidth=1pt,arrowsize=1.5pt 4,linecolor=black]{->}(7,0.85)(6,-0.95)
\psline[linewidth=1pt,arrowsize=1.5pt 4,linecolor=black]{->}(7.05,0.9)(10.9,-0.9)
\pscircle[linewidth=1pt,origin={8,1},fillstyle=solid,fillcolor=black](0,0){0.15}
\rput(8,1.4){$\pmb{e_{2,3}}$}
\psline[linewidth=1pt,arrowsize=1.5pt 4,linecolor=black]{->}(8,0.85)(6.9,-0.9)
\psline[linewidth=1pt,arrowsize=1.5pt 4,linecolor=black]{->}(8.05,0.9)(7.9,-0.9)
\pscircle[linewidth=1pt,origin={9,1},fillstyle=solid,fillcolor=black](0,0){0.15}
\psline[linewidth=1pt,arrowsize=1.5pt 4,linecolor=black]{->}(8.95,0.9)(6.95,-0.95)
\psline[linewidth=1pt,arrowsize=1.5pt 4,linecolor=black]{->}(9,0.95)(9.95,-0.95)
\pscircle[linewidth=1pt,origin={10,1},fillstyle=none,fillcolor=lightgray](0,0){0.15}
\psline[linewidth=1pt,arrowsize=1.5pt 4,linecolor=black]{->}(9.95,0.9)(7.95,-0.95)
\psline[linewidth=1pt,arrowsize=1.5pt 4,linecolor=black]{->}(10,0.9)(8.95,-0.95)
\pscircle[linewidth=1pt,origin={11,1},fillstyle=none,fillcolor=lightgray](0,0){0.15}
\psline[linewidth=1pt,arrowsize=1.5pt 4,linecolor=black]{->}(10.95,0.9)(8.1,-1)
\psline[linewidth=1pt,arrowsize=1.5pt 4,linecolor=black]{->}(11.05,0.9)(9.95,-0.95)
\pscircle[linewidth=1pt,origin={12,1},fillstyle=none,fillcolor=lightgray](0,0){0.15}
\psline[linewidth=1pt,arrowsize=1.5pt 4,linecolor=black]{->}(11.95,0.9)(9.1,-1)
\psline[linewidth=1pt,arrowsize=1.5pt 4,linecolor=black]{->}(12.05,0.9)(10,-0.95)
\pscircle[linewidth=1pt,origin={13,1},fillstyle=none,fillcolor=lightgray](0,0){0.15}
\psline[linewidth=1pt,arrowsize=1.5pt 4,linecolor=black]{->}(12.95,0.9)(10.1,-1)
\psline[linewidth=1pt,arrowsize=1.5pt 4,linecolor=black]{->}(13,0.95)(11.1,-1)
\rput(9,2){$\pmb{t=2}$}
\end{pspicture}
}

\item
$\din(e_{i,j})=0$ and 
\\ 
$\dout(e_{i,j})=2$ for all $i$ and $j$.
Since $\din(e_{i,j})=0$, if a node $e_{i,j}$ is shocked, no part of the shock is propagated to any other node in the network.

\item
Since the longest path in $\overrightarrow{G}$ has $2$ edges, by Proposition \ref{obs1}(b) no new node fails at any $t>3$.
\end{itemize}
For notational convenience, let $n=|V|$, $\E=E/n$, 
and $e_{i,j_1},e_{i,j_2}$ and $e_{i,j_3}$ be the three edges $\{v_i,v_{j_1}\}$, $\{v_i,v_{j_2}\}$ and $\{v_i,v_{j_3}\}$ in $G$ that are incident on the node $v_i$.
We will select the remaining network parameters, namely $\gamma$, $\Phi$ and $\E$, based on the following desirable properties.

\vspace*{0.1in}
\noindent
{\bf (I)}
If a node $u_i'$ is shocked at $t=1$, it fails: 
\begin{gather}
\Phi\,\left( \din(u_i')-\dout(u_i')+\E\right) > \gamma\,\left( \din(u_i')+\E \right)
\,\,\,\equiv\,\,\,
\Phi\,\left(1+\E\right) > \gamma\,\left(1+\E\right) 
\,\,\,\equiv\,\,\,
\Phi > \gamma
\label{eq1} 
\end{gather}
{\bf (II)}
If a node $e_{i,j}$ is shocked, it does not fail:
\begin{gather}
\din\left(e_{i,j}\right)-\dout\left(e_{i,j}\right)+\E<0
\,\,\,\equiv\,\,\,
\E < 2
\label{eq1point5} 
\end{gather}

\vspace*{0.1in}
\noindent
{\bf (III)}
If a node $u_i$ is shocked at $t=1$, then $u_i$ fails at $t=1$, and 
the nodes $e_{i,j_1},e_{i,j_2}$ and $e_{i,j_3}$ fail at time $t=2$ if they were not shocked (see \FI{fig10} for an illustration):
\begin{gather*}
\frac{\min\big\{\,\Phi\,\left( \din(u_i)-\dout(u_i)+\E\right) - \gamma\,\left( \din(u_i)+\E \right),\,\din(u_i)\,\big\}}{\din(u_i)} 
> 
\gamma\,\left( \din(e_{i,j_1})+\E \right) \\
\,\,\,\equiv\,\,\,
\frac{\min\big\{\,\Phi\,(2+\E) - \gamma\,(3+\E),\,3\,\big\}}{3} > \gamma\,\E
\end{gather*}
The above inequality is satisfied provided:
\begin{gather}
\Phi\,(2+\E) > \gamma\,(3+4\E)
\label{eq2-1} 
\\
1 > \gamma\,\E
\,\,\,\equiv\,\,\,
\gamma < \frac{1}{\E} 
\label{eq2-2} 
\end{gather}
{\bf (IV)}
Consider a sink node $e_{i,j}$. Then, we require that 
if one or both of the super-source node $u_i'$ and $u_j'$ are shocked at $t=1$ but the none of the nodes $u_i$, $u_j$ and $e_{i,j}$ were shocked, then we require that 
one or both of the corresponding nodes $u_i$ and $u_j$ fail at $t=2$, but the node $e_{i,j}$ {\em never} fails. Pictorially, we want a situation as depicted in~\FI{fig2}. This is satisfied 
provided the following inequalities hold:
\begin{description}
\item[(IV-1)]
$u_i$ fails at $t=2$ if $u_i'$ was shocked (the case of $u_j$ and $u_j'$ is similar):
\begin{gather*}
\frac{\min\big\{\,\Phi\,\left( \din(u_i')-\dout(u_i')+\E\right) - \gamma\,\left( \din(u_i')+\E\,\right),\,\din(u_i')\,\big\}}{\din(u_i')} > 
\gamma\,\left( \din(u_i)+\E \right)  \\
\,\,\,\equiv\,\,\,
\frac{\min\big\{\,(\Phi-\gamma)(1+\E),\,1\,\big\}}{1} > \gamma\,(3+\E)
\end{gather*}
The above inequality is satisfied provided:
\begin{gather}
(\Phi-\gamma)(1+\E) > \gamma\,(3+\E)
\,\,\,\equiv\,\,\,
\Phi(1+\E) > \gamma\,(4+2\,\E)
\label{eq3} 
\\
1 > \gamma\,(3+\E) 
\,\,\,\equiv\,\,\,
\gamma < \frac{1}{3+\E} 
\label{eq4} 
\end{gather}

\item[(IV-2)]
$e_{i,j}$ never fails even if both $u_i$ and $u_j$ have failed:
\begin{gather*}
\frac{\min\big\{\,(\Phi-\gamma)(1+\E),\,1\,\big\}}{1}-\gamma\,(3+\E)
\leq
\frac{\gamma\,\E}{2} 
\,\,\,\equiv\,\,\,
\min\big\{\,(\Phi-\gamma)(1+\E),\,1\,\big\}
\leq
3\,\gamma\,\left(1+\frac{\E}{2}\right) 
\end{gather*}
The above inequality is satisfied provided:
\begin{gather}
(\Phi-\gamma)(1+\E)
\leq
3\,\gamma\,\left(1+\frac{\E}{2}\right) 
\,\,\,\equiv\,\,\,
\Phi\,(1+\E)
\leq
\gamma\,\left(4+\frac{5\E}{2}\right) 
\label{eq5} 
\\
1
\leq
3\,\gamma\,\left(1+\frac{\E}{2}\right) 
\,\,\,\equiv\,\,\,
\gamma
\geq
\frac{2}{6+3\,E}
\label{eq6} 
\end{gather}
\end{description}

\pichskip{0.9cm}
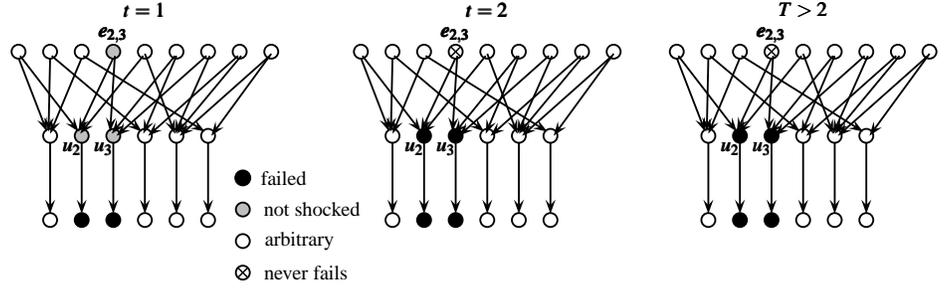
\piccaption{\label{fig2}Case~{\bf (IV)}: to make $e_{2,3}$ fail, at least one of $u_2$ or $u_3$ must be shocked.}
\parpic[r][c]
{
\scalebox{0.7}{
\begin{pspicture}(3,-3.5)(7,2.2)
\psset{xunit=0.6cm,yunit=0.8cm}
\pscircle[linewidth=1pt,origin={6,-1},fillstyle=none,fillcolor=lightgray](0,0){0.15}
\pscircle[linewidth=1pt,origin={6,-3},fillstyle=none,fillcolor=lightgray](0,0){0.15}
\psline[linewidth=1pt,arrowsize=1.5pt 4,linecolor=black]{->}(6,-1.15)(6,-2.85)
\pscircle[linewidth=1pt,origin={7,-1},fillstyle=solid,fillcolor=lightgray](0,0){0.15}
\pscircle[linewidth=1pt,origin={7,-3},fillstyle=solid,fillcolor=black](0,0){0.15}
\psline[linewidth=1pt,arrowsize=1.5pt 4,linecolor=black]{->}(7,-1.15)(7,-2.85)
\pscircle[linewidth=1pt,origin={8,-1},fillstyle=solid,fillcolor=lightgray](0,0){0.15}
\pscircle[linewidth=1pt,origin={8,-3},fillstyle=solid,fillcolor=black](0,0){0.15}
\psline[linewidth=1pt,arrowsize=1.5pt 4,linecolor=black]{->}(8,-1.15)(8,-2.85)
\pscircle[linewidth=1pt,origin={9,-3},fillstyle=none,fillcolor=lightgray](0,0){0.15}
\psline[linewidth=1pt,arrowsize=1.5pt 4,linecolor=black]{->}(9,-1.15)(9,-2.85)
\pscircle[linewidth=1pt,origin={9,-1},fillstyle=none,fillcolor=lightgray](0,0){0.15}
\pscircle[linewidth=1pt,origin={10,-3},fillstyle=none,fillcolor=lightgray](0,0){0.15}
\psline[linewidth=1pt,arrowsize=1.5pt 4,linecolor=black]{->}(10,-1.15)(10,-2.85)
\pscircle[linewidth=1pt,origin={10,-1},fillstyle=none,fillcolor=lightgray](0,0){0.15}
\pscircle[linewidth=1pt,origin={11,-3},fillstyle=none,fillcolor=lightgray](0,0){0.15}
\psline[linewidth=1pt,arrowsize=1.5pt 4,linecolor=black]{->}(11,-1.15)(11,-2.85)
\pscircle[linewidth=1pt,origin={11,-1},fillstyle=none,fillcolor=lightgray](0,0){0.15}
\pscircle[linewidth=1pt,origin={5,1},fillstyle=none,fillcolor=lightgray](0,0){0.15}
\psline[linewidth=1pt,arrowsize=1.5pt 4,linecolor=black]{->}(5,0.85)(5.9,-0.9)
\psline[linewidth=1pt,arrowsize=1.5pt 4,linecolor=black]{->}(5.05,0.9)(6.9,-0.9)
\pscircle[linewidth=1pt,origin={6,1},fillstyle=none,fillcolor=lightgray](0,0){0.15}
\psline[linewidth=1pt,arrowsize=1.5pt 4,linecolor=black]{->}(6,0.85)(5.95,-0.9)
\psline[linewidth=1pt,arrowsize=1.5pt 4,linecolor=black]{->}(6.05,0.9)(8.9,-0.9)
\pscircle[linewidth=1pt,origin={7,1},fillstyle=none,fillcolor=lightgray](0,0){0.15}
\psline[linewidth=1pt,arrowsize=1.5pt 4,linecolor=black]{->}(7,0.85)(6,-0.95)
\psline[linewidth=1pt,arrowsize=1.5pt 4,linecolor=black]{->}(7.05,0.9)(10.9,-0.9)
\pscircle[linewidth=1pt,origin={8,1},fillstyle=solid,fillcolor=lightgray](0,0){0.15}
\rput(8,1.4){$\pmb{e_{2,3}}$}
\psline[linewidth=1pt,arrowsize=1.5pt 4,linecolor=black]{->}(8,0.85)(6.9,-0.9)
\psline[linewidth=1pt,arrowsize=1.5pt 4,linecolor=black]{->}(8.05,0.9)(7.9,-0.9)
\pscircle[linewidth=1pt,origin={9,1},fillstyle=none,fillcolor=lightgray](0,0){0.15}
\psline[linewidth=1pt,arrowsize=1.5pt 4,linecolor=black]{->}(8.95,0.9)(6.95,-0.95)
\psline[linewidth=1pt,arrowsize=1.5pt 4,linecolor=black]{->}(9,0.95)(9.95,-0.95)
\pscircle[linewidth=1pt,origin={10,1},fillstyle=none,fillcolor=lightgray](0,0){0.15}
\psline[linewidth=1pt,arrowsize=1.5pt 4,linecolor=black]{->}(9.95,0.9)(7.95,-0.95)
\psline[linewidth=1pt,arrowsize=1.5pt 4,linecolor=black]{->}(10,0.9)(8.95,-0.95)
\pscircle[linewidth=1pt,origin={11,1},fillstyle=none,fillcolor=lightgray](0,0){0.15}
\psline[linewidth=1pt,arrowsize=1.5pt 4,linecolor=black]{->}(10.95,0.9)(8.1,-1)
\psline[linewidth=1pt,arrowsize=1.5pt 4,linecolor=black]{->}(11.05,0.9)(9.95,-0.95)
\pscircle[linewidth=1pt,origin={12,1},fillstyle=none,fillcolor=lightgray](0,0){0.15}
\psline[linewidth=1pt,arrowsize=1.5pt 4,linecolor=black]{->}(11.95,0.9)(9.1,-1)
\psline[linewidth=1pt,arrowsize=1.5pt 4,linecolor=black]{->}(12.05,0.9)(10,-0.95)
\pscircle[linewidth=1pt,origin={13,1},fillstyle=none,fillcolor=lightgray](0,0){0.15}
\psline[linewidth=1pt,arrowsize=1.5pt 4,linecolor=black]{->}(12.95,0.9)(10.1,-1)
\psline[linewidth=1pt,arrowsize=1.5pt 4,linecolor=black]{->}(13,0.95)(11.1,-1)
\rput(9,2){$\pmb{t=1}$}
\pscircle[linewidth=0pt,origin={12.1,-2},fillstyle=solid,fillcolor=black](0,0){0.15}
\rput(13.4,-2){failed}
\pscircle[linewidth=1pt,origin={12.1,-2.75},fillstyle=solid,fillcolor=lightgray](0,0){0.15}
\rput(14.3,-2.75){not shocked}
\pscircle[linewidth=1pt,origin={12.1,-3.5},fillstyle=solid,fillcolor=white](0,0){0.15}
\rput(13.9,-3.5){arbitrary}
\pscircle[linewidth=1pt,origin={12.1,-4.25},fillstyle=crosshatch](0,0){0.15}
\rput(14.1,-4.25){never fails}
\rput(6.7,-1.3){$\pmb{u_2}$}
\rput(7.7,-1.3){$\pmb{u_3}$}
\rput(8,1.4){$\pmb{e_{2,3}}$}
\hspace*{6.5cm}
\pscircle[linewidth=1pt,origin={6,-1},fillstyle=none,fillcolor=lightgray](0,0){0.15}
\pscircle[linewidth=1pt,origin={6,-3},fillstyle=none,fillcolor=lightgray](0,0){0.15}
\psline[linewidth=1pt,arrowsize=1.5pt 4,linecolor=black]{->}(6,-1.15)(6,-2.85)
\pscircle[linewidth=1pt,origin={7,-1},fillstyle=solid,fillcolor=black](0,0){0.15}
\pscircle[linewidth=1pt,origin={7,-3},fillstyle=solid,fillcolor=black](0,0){0.15}
\psline[linewidth=1pt,arrowsize=1.5pt 4,linecolor=black]{->}(7,-1.15)(7,-2.85)
\pscircle[linewidth=1pt,origin={8,-1},fillstyle=solid,fillcolor=black](0,0){0.15}
\pscircle[linewidth=1pt,origin={8,-3},fillstyle=solid,fillcolor=black](0,0){0.15}
\psline[linewidth=1pt,arrowsize=1.5pt 4,linecolor=black]{->}(8,-1.15)(8,-2.85)
\pscircle[linewidth=1pt,origin={9,-3},fillstyle=none,fillcolor=lightgray](0,0){0.15}
\psline[linewidth=1pt,arrowsize=1.5pt 4,linecolor=black]{->}(9,-1.15)(9,-2.85)
\pscircle[linewidth=1pt,origin={9,-1},fillstyle=none,fillcolor=lightgray](0,0){0.15}
\pscircle[linewidth=1pt,origin={10,-3},fillstyle=none,fillcolor=lightgray](0,0){0.15}
\psline[linewidth=1pt,arrowsize=1.5pt 4,linecolor=black]{->}(10,-1.15)(10,-2.85)
\pscircle[linewidth=1pt,origin={10,-1},fillstyle=none,fillcolor=lightgray](0,0){0.15}
\pscircle[linewidth=1pt,origin={11,-3},fillstyle=none,fillcolor=lightgray](0,0){0.15}
\psline[linewidth=1pt,arrowsize=1.5pt 4,linecolor=black]{->}(11,-1.15)(11,-2.85)
\pscircle[linewidth=1pt,origin={11,-1},fillstyle=none,fillcolor=lightgray](0,0){0.15}
\pscircle[linewidth=1pt,origin={5,1},fillstyle=none,fillcolor=lightgray](0,0){0.15}
\psline[linewidth=1pt,arrowsize=1.5pt 4,linecolor=black]{->}(5,0.85)(5.9,-0.9)
\psline[linewidth=1pt,arrowsize=1.5pt 4,linecolor=black]{->}(5.05,0.9)(6.9,-0.9)
\pscircle[linewidth=1pt,origin={6,1},fillstyle=none,fillcolor=lightgray](0,0){0.15}
\psline[linewidth=1pt,arrowsize=1.5pt 4,linecolor=black]{->}(6,0.85)(5.95,-0.9)
\psline[linewidth=1pt,arrowsize=1.5pt 4,linecolor=black]{->}(6.05,0.9)(8.9,-0.9)
\pscircle[linewidth=1pt,origin={7,1},fillstyle=none,fillcolor=lightgray](0,0){0.15}
\psline[linewidth=1pt,arrowsize=1.5pt 4,linecolor=black]{->}(7,0.85)(6,-0.95)
\psline[linewidth=1pt,arrowsize=1.5pt 4,linecolor=black]{->}(7.05,0.9)(10.9,-0.9)
\pscircle[linewidth=1pt,origin={8,1},fillstyle=crosshatch](0,0){0.15}
\rput(8,1.4){$\pmb{e_{2,3}}$}
\psline[linewidth=1pt,arrowsize=1.5pt 4,linecolor=black]{->}(8,0.85)(6.9,-0.9)
\psline[linewidth=1pt,arrowsize=1.5pt 4,linecolor=black]{->}(8.05,0.9)(7.9,-0.9)
\pscircle[linewidth=1pt,origin={9,1},fillstyle=none,fillcolor=lightgray](0,0){0.15}
\psline[linewidth=1pt,arrowsize=1.5pt 4,linecolor=black]{->}(8.95,0.9)(6.95,-0.95)
\psline[linewidth=1pt,arrowsize=1.5pt 4,linecolor=black]{->}(9,0.95)(9.95,-0.95)
\pscircle[linewidth=1pt,origin={10,1},fillstyle=none,fillcolor=lightgray](0,0){0.15}
\psline[linewidth=1pt,arrowsize=1.5pt 4,linecolor=black]{->}(9.95,0.9)(7.95,-0.95)
\psline[linewidth=1pt,arrowsize=1.5pt 4,linecolor=black]{->}(10,0.9)(8.95,-0.95)
\pscircle[linewidth=1pt,origin={11,1},fillstyle=none,fillcolor=lightgray](0,0){0.15}
\psline[linewidth=1pt,arrowsize=1.5pt 4,linecolor=black]{->}(10.95,0.9)(8.1,-1)
\psline[linewidth=1pt,arrowsize=1.5pt 4,linecolor=black]{->}(11.05,0.9)(9.95,-0.95)
\pscircle[linewidth=1pt,origin={12,1},fillstyle=none,fillcolor=lightgray](0,0){0.15}
\psline[linewidth=1pt,arrowsize=1.5pt 4,linecolor=black]{->}(11.95,0.9)(9.1,-1)
\psline[linewidth=1pt,arrowsize=1.5pt 4,linecolor=black]{->}(12.05,0.9)(10,-0.95)
\pscircle[linewidth=1pt,origin={13,1},fillstyle=none,fillcolor=lightgray](0,0){0.15}
\psline[linewidth=1pt,arrowsize=1.5pt 4,linecolor=black]{->}(12.95,0.9)(10.1,-1)
\psline[linewidth=1pt,arrowsize=1.5pt 4,linecolor=black]{->}(13,0.95)(11.1,-1)
\rput(9,2){$\pmb{t=2}$}
\rput(6.7,-1.3){$\pmb{u_2}$}
\rput(7.7,-1.3){$\pmb{u_3}$}
\rput(8,1.4){$\pmb{e_{2,3}}$}
\hspace*{6cm}
\pscircle[linewidth=1pt,origin={6,-1},fillstyle=none,fillcolor=lightgray](0,0){0.15}
\pscircle[linewidth=1pt,origin={6,-3},fillstyle=none,fillcolor=lightgray](0,0){0.15}
\psline[linewidth=1pt,arrowsize=1.5pt 4,linecolor=black]{->}(6,-1.15)(6,-2.85)
\pscircle[linewidth=1pt,origin={7,-1},fillstyle=solid,fillcolor=black](0,0){0.15}
\pscircle[linewidth=1pt,origin={7,-3},fillstyle=solid,fillcolor=black](0,0){0.15}
\psline[linewidth=1pt,arrowsize=1.5pt 4,linecolor=black]{->}(7,-1.15)(7,-2.85)
\pscircle[linewidth=1pt,origin={8,-1},fillstyle=solid,fillcolor=black](0,0){0.15}
\pscircle[linewidth=1pt,origin={8,-3},fillstyle=solid,fillcolor=black](0,0){0.15}
\psline[linewidth=1pt,arrowsize=1.5pt 4,linecolor=black]{->}(8,-1.15)(8,-2.85)
\pscircle[linewidth=1pt,origin={9,-3},fillstyle=none,fillcolor=lightgray](0,0){0.15}
\psline[linewidth=1pt,arrowsize=1.5pt 4,linecolor=black]{->}(9,-1.15)(9,-2.85)
\pscircle[linewidth=1pt,origin={9,-1},fillstyle=none,fillcolor=lightgray](0,0){0.15}
\pscircle[linewidth=1pt,origin={10,-3},fillstyle=none,fillcolor=lightgray](0,0){0.15}
\psline[linewidth=1pt,arrowsize=1.5pt 4,linecolor=black]{->}(10,-1.15)(10,-2.85)
\pscircle[linewidth=1pt,origin={10,-1},fillstyle=none,fillcolor=lightgray](0,0){0.15}
\pscircle[linewidth=1pt,origin={11,-3},fillstyle=none,fillcolor=lightgray](0,0){0.15}
\psline[linewidth=1pt,arrowsize=1.5pt 4,linecolor=black]{->}(11,-1.15)(11,-2.85)
\pscircle[linewidth=1pt,origin={11,-1},fillstyle=none,fillcolor=lightgray](0,0){0.15}
\pscircle[linewidth=1pt,origin={5,1},fillstyle=none,fillcolor=lightgray](0,0){0.15}
\psline[linewidth=1pt,arrowsize=1.5pt 4,linecolor=black]{->}(5,0.85)(5.9,-0.9)
\psline[linewidth=1pt,arrowsize=1.5pt 4,linecolor=black]{->}(5.05,0.9)(6.9,-0.9)
\pscircle[linewidth=1pt,origin={6,1},fillstyle=none,fillcolor=lightgray](0,0){0.15}
\psline[linewidth=1pt,arrowsize=1.5pt 4,linecolor=black]{->}(6,0.85)(5.95,-0.9)
\psline[linewidth=1pt,arrowsize=1.5pt 4,linecolor=black]{->}(6.05,0.9)(8.9,-0.9)
\pscircle[linewidth=1pt,origin={7,1},fillstyle=none,fillcolor=lightgray](0,0){0.15}
\psline[linewidth=1pt,arrowsize=1.5pt 4,linecolor=black]{->}(7,0.85)(6,-0.95)
\psline[linewidth=1pt,arrowsize=1.5pt 4,linecolor=black]{->}(7.05,0.9)(10.9,-0.9)
\pscircle[linewidth=1pt,origin={8,1},fillstyle=crosshatch](0,0){0.15}
\rput(8,1.4){$\pmb{e_{2,3}}$}
\psline[linewidth=1pt,arrowsize=1.5pt 4,linecolor=black]{->}(8,0.85)(6.9,-0.9)
\psline[linewidth=1pt,arrowsize=1.5pt 4,linecolor=black]{->}(8.05,0.9)(7.9,-0.9)
\pscircle[linewidth=1pt,origin={9,1},fillstyle=none,fillcolor=lightgray](0,0){0.15}
\psline[linewidth=1pt,arrowsize=1.5pt 4,linecolor=black]{->}(8.95,0.9)(6.95,-0.95)
\psline[linewidth=1pt,arrowsize=1.5pt 4,linecolor=black]{->}(9,0.95)(9.95,-0.95)
\pscircle[linewidth=1pt,origin={10,1},fillstyle=none,fillcolor=lightgray](0,0){0.15}
\psline[linewidth=1pt,arrowsize=1.5pt 4,linecolor=black]{->}(9.95,0.9)(7.95,-0.95)
\psline[linewidth=1pt,arrowsize=1.5pt 4,linecolor=black]{->}(10,0.9)(8.95,-0.95)
\pscircle[linewidth=1pt,origin={11,1},fillstyle=none,fillcolor=lightgray](0,0){0.15}
\psline[linewidth=1pt,arrowsize=1.5pt 4,linecolor=black]{->}(10.95,0.9)(8.1,-1)
\psline[linewidth=1pt,arrowsize=1.5pt 4,linecolor=black]{->}(11.05,0.9)(9.95,-0.95)
\pscircle[linewidth=1pt,origin={12,1},fillstyle=none,fillcolor=lightgray](0,0){0.15}
\psline[linewidth=1pt,arrowsize=1.5pt 4,linecolor=black]{->}(11.95,0.9)(9.1,-1)
\psline[linewidth=1pt,arrowsize=1.5pt 4,linecolor=black]{->}(12.05,0.9)(10,-0.95)
\pscircle[linewidth=1pt,origin={13,1},fillstyle=none,fillcolor=lightgray](0,0){0.15}
\psline[linewidth=1pt,arrowsize=1.5pt 4,linecolor=black]{->}(12.95,0.9)(10.1,-1)
\psline[linewidth=1pt,arrowsize=1.5pt 4,linecolor=black]{->}(13,0.95)(11.1,-1)
\rput(9,2){$\pmb{T>2}$}
\rput(6.7,-1.3){$\pmb{u_2}$}
\rput(7.7,-1.3){$\pmb{u_3}$}
\rput(8,1.4){$\pmb{e_{2,3}}$}
\end{pspicture}
} 
}

\noindent
There are obviously many choices of parameters $\gamma$, $\Phi$ and $\E$ that satisfy 
Equations~\eqref{eq1}--\eqref{eq6}; here we exhibit just one. 
Let $\E=1$ which satisfied Equation~\eqref{eq1point5}. Choosing $\gamma=0.23$ satisfies Equations~\eqref{eq2-2}, \eqref{eq4} and \eqref{eq6}.
Letting $\Phi=0.7$ satisfies Equations~\eqref{eq1}, \eqref{eq2-1}, \eqref{eq3} and \eqref{eq5}.

Suppose that $V'\subset V$ is a solution of $3$-{\sf MIN}-{\sf NODE}-{\sf COVER}. 
Then, we shock all the super-nodes, and the nodes in $V'$. By {\bf (I)} and {\bf (III)} all the super-nodes and the nodes in $\left(\cup_{v_i\in V\setminus V'} \{v_i\}\right)$ 
fails at $t=1$, and by {\bf (III)} the nodes in $\cup_{\substack{\{v_i,v_j\}\in E \\ i < j}} \{e_{i,j}\}$ fails $t=2$.
Thus, we obtain a solution of $\overrightarrow{G}$ by shocking $|V'|+n$ nodes.

Conversely, consider a solution of the \nam\ problem on $\overrightarrow{G}$.
Remember that all the super-nodes must be shocked, which ensures 
that we need to shock $n+a$ nodes for some integer $a\geq 0$, and 
that any node $v_i$ that is not shocked will fail at $t=2$. By {\bf (II)} it is of no use to shock the sink nodes.
Thus, the shocked nodes consist of all super-nodes and a subset $V'$ of cardinality $a$ of the nodes $u_1,u_2,\dots,u_n$.
By {\bf (IV)} for every node $e_{i,j}$ at least one of the nodes $u_i$ or $u_j$ must be in $U$.
Thus, the set of nodes $\{v_i\,|\,u_i\in U\}$ form a node cover of $G$ of size $a$.

That the reduction is an L-reduction follows from the observation that any locally improvable solution of 
$3$-{\sf MIN}-{\sf NODE}-{\sf COVER} has between $n/3$ and $n$ nodes.
\end{proof}

\section{Restricted Homogeneous Networks, $\pmb{\mbox{\nam}}$, Any $T$, Exact Solution} 

The $\apx$-hardness result of Theorem~\ref{apx1} has constant values for both $\Phi$ and $\gamma$, 
and requires $\dout(v)=2$ for some nodes $v$. We show that if $\dout(v)\leq 1$ for every node $v$ then 
under mild technical assumptions $\vi^\ast(G,T)$ can be computed in polynomial time for any $T$ and, in addition, if $\din(v)$ is bounded by a constant 
for every node $v$ then the network is highly stable (\IE, $\vi^\ast(G,T)$ is large).
Recall that an in-arborescence is a directed rooted tree where all edges are oriented towards the root.

\begin{theorem}\label{poly1}
If the banking network $G$ is a rooted in-arborescence then $\vi^\ast(G,T)>\dfrac{1}{1+\mathrm{deg}_{\mathrm{in}}^{\max}\left( \frac{\Phi}{\gamma}-1 \right)}$,
where $\mathrm{deg}_{\mathrm{in}}^{\max}=\max_{\,v\in V}\big\{\din(v)\big\}$.
Moreover, under the assumption that every node of $G$ can be individually failed by shocking, 
$\vi^\ast(G,T)$ can be computed exactly in $O\left(n^2\right)$ time.
\end{theorem}

\begin{remark}
Thus, for example, when $\mathrm{deg}_{\mathrm{in}}^{\max}=3$, $\gamma=0.1$ and $\Phi=0.15$, we get $\vi^\ast(G,T)>0.22$ and 
the network cannot be put to death without shocking more than $22\%$ of the nodes. The proof gives an example for which the lower bound is tight.
\end{remark}

In the rest of this section, we prove the above theorem. Let $G=(V,F)$ be the given in-arborescence rooted at node $r$. We will 
use the following notations and terminologies:
\begin{itemize}
\item
$u\rightarrow v$ and $u\leadsto v$ denote a directed edge and a directed path of one of more edges, respectively, from node $u$ to node $v$. 

\item
If $(u,v)\in F$ then $v$ is the {\em parent} of $u$ and $u$ is a {\em child} of $v$. Similarly, 
if $u\leadsto v$ exists in $G$ then $v$ an {\em ancestor} of $u$ and $u$ a {\em descendent} of $v$. 

\item
Let $\nabla(u)=\{\,v\,|\,u\leadsto v \mbox{ exists in } G\,\}$ denote the set of all proper ancestors of $u$, and 
$\Delta(u)=\{\,v\,|\,v\leadsto u \mbox{ exists in } G\,\}\cup\{u\}$ denote the set of all descendants of $u$ (including the node $u$ itself).
Note that for the network $G$ to fail, at least one node in $\nabla(u)\cup\{u\}$ must be shocked for every node $u$. 
\end{itemize}
Suppose that we shock a node $u$ of $G$ (and shock no other nodes in $\Delta(u)$). If $u$ fails, then the shock splits and propagates to a subset of nodes in $\Delta(u)$ 
until each split part of the shock terminates because of one of the following reasons:
\begin{itemize}
\item
the component of the shock reaches a ``leaf'' node $v$ with $\din(v)=0$, or  

\item
the component of the shock reaches a node $v$ with a sufficiently high $c_v$ such that $v$ does not fail.
\end{itemize}
Based on the above observations, we define the following quantities.

\begin{definition}[see~\FI{fr} for illustrations]
The {\em influence zone} of a shock on $u$, denoted by $\iz(u)$, is the set of all failed nodes $v\in\Delta(u)$ within time $T$ when $u$ is shocked (and, no
other node in $\Delta(u)$ is shocked). Note that $u\in\iz(u)$. 
\end{definition}

Note that, for any node $u$, $\iz(u)$ can be computed in $O(n)$ time.

\begin{figure}[htbp]
\begin{minipage}[c]{3.8in}
\begin{pspicture}(-3.5,-3.2)(5.5,3)
\psset{xunit=0.5cm,yunit=0.5cm}
\pscurve[linewidth=2pt,linecolor=gray,fillstyle=none,fillcolor=lightgray]
      (-3.8,-6.5)(-3.2,-5.5)(-1,-2.6)(0.2,-2.6)(0.5,-4.9)(4,-4.9)(1.3,0.5)(-2,0.5)(-6,-6.5)(-3.8,-6.5)
\rput(-5.5,-1.2){\color{black}\bf\Large $\pmb{\iz(u)}$}
\pscircle[linewidth=0pt,origin={0,0},fillstyle=solid,fillcolor=gray](0,0){0.15}
\rput(0.7,0){$u$}
\rput(-1.3,0){\footnotesize\bf parent}
\pscircle[linewidth=1pt,origin={-2,-2},fillstyle=solid,fillcolor=black](0,0){0.15}
\rput(-3,-2){\footnotesize\bf child}
\pscircle[linewidth=1pt,origin={0,-2},fillstyle=solid,fillcolor=black](0,0){0.15}
\pscircle[linewidth=1pt,origin={2,-2},fillstyle=solid,fillcolor=black](0,0){0.15}
\pscircle[linewidth=1pt,origin={-3,-4},fillstyle=solid,fillcolor=black](0,0){0.15}
\pscircle[linewidth=1pt,origin={-1,-4},fillstyle=none,fillcolor=lightgray](0,0){0.15}
\pscircle[linewidth=1pt,origin={1,-4},fillstyle=solid,fillcolor=black](0,0){0.15}
\pscircle[linewidth=1pt,origin={2,-4},fillstyle=solid,fillcolor=black](0,0){0.15}
\pscircle[linewidth=1pt,origin={3,-4},fillstyle=solid,fillcolor=black](0,0){0.15}
\pscircle[linewidth=1pt,origin={-5,-6},fillstyle=solid,fillcolor=black](0,0){0.15}
\pscircle[linewidth=1pt,origin={-4,-6},fillstyle=solid,fillcolor=black](0,0){0.15}
\pscircle[linewidth=1pt,origin={-3,-6},fillstyle=none,fillcolor=lightgray](0,0){0.15}
\pscircle[linewidth=1pt,origin={-2,-6},fillstyle=none,fillcolor=lightgray](0,0){0.15}
\pscircle[linewidth=1pt,origin={-1,-6},fillstyle=none,fillcolor=lightgray](0,0){0.15}
\pscircle[linewidth=1pt,origin={3,-6},fillstyle=none,fillcolor=lightgray](0,0){0.15}
\pscircle[linewidth=1pt,origin={4,-6},fillstyle=none,fillcolor=lightgray](0,0){0.15}
\psline[linewidth=1pt,arrowsize=1.5pt 4,linecolor=black]{<-}(2,2)(0.15,0.15)
\psline[linewidth=1pt,linestyle=dashed,linecolor=black](2,2)(3,3)
\psline[linewidth=1pt,arrowsize=1.5pt 4,linecolor=black]{<-}(-0.15,-0.15)(-2,-2)
\psline[linewidth=1pt,arrowsize=1.5pt 4,linecolor=black]{<-}(0,-0.15)(0,-2)
\psline[linewidth=1pt,arrowsize=1.5pt 4,linecolor=black]{<-}(0.15,-0.15)(2,-2)
\psline[linewidth=1pt,arrowsize=1.5pt 4,linecolor=black]{<-}(-2.15,-2.15)(-3,-4)
\psline[linewidth=1pt,arrowsize=1.5pt 4,linecolor=black]{<-}(-1.85,-2.15)(-0.85,-3.85)
\psline[linewidth=1pt,arrowsize=1.5pt 4,linecolor=black]{<-}(1.85,-2.15)(1,-4)
\psline[linewidth=1pt,arrowsize=1.5pt 4,linecolor=black]{<-}(2,-2.15)(2,-4)
\psline[linewidth=1pt,arrowsize=1.5pt 4,linecolor=black]{<-}(2.15,-2.15)(3,-4)
\psline[linewidth=1pt,arrowsize=1.5pt 4,linecolor=black]{<-}(-3.15,-4.15)(-5,-5.85)
\psline[linewidth=1pt,arrowsize=1.5pt 4,linecolor=black]{<-}(-3.05,-4.15)(-4,-5.85)
\psline[linewidth=1pt,arrowsize=1.5pt 4,linecolor=black]{<-}(-3,-4.15)(-3,-5.85)
\psline[linewidth=1pt,arrowsize=1.5pt 4,linecolor=black]{<-}(-2.95,-4.15)(-2,-5.85)
\psline[linewidth=1pt,arrowsize=1.5pt 4,linecolor=black]{<-}(-2.85,-4.15)(-1,-5.85)
\psline[linewidth=1pt,arrowsize=1.5pt 4,linecolor=black]{<-}(3,-4.15)(3,-5.85)
\psline[linewidth=1pt,arrowsize=1.5pt 4,linecolor=black]{<-}(3.15,-4.15)(4,-5.85)
\pscircle[linewidth=0pt,origin={4.7,5.6},fillstyle=solid,fillcolor=gray](0,0){0.15}
\rput(2.8,5.6){shocked}
\pscircle[linewidth=0pt,origin={4.7,4.6},fillstyle=solid,fillcolor=black](0,0){0.15}
\rput(0.9,4.6){failed (due to shock)}
\pscircle[linewidth=1pt,origin={4.7,3.6},fillstyle=none,fillcolor=black](0,0){0.15}
\rput(0.1,3.6){not shocked and not failed}
\end{pspicture}
\caption{\label{fr}Influence zone of a shock on $u$. 
}
\end{minipage}
\hspace*{0.1in}
\begin{minipage}[c]{2.4in}
\begin{pspicture}(-0,0)(5.8,-4.2)
\psset{xunit=1cm,yunit=1cm}
\pscircle[linewidth=1pt,origin={2,0},fillstyle=none,fillcolor=lightgray](0,0){0.2}
\rput(2,0){$\displaystyle\pmb{u}$}
\psline[linewidth=1pt,arrowsize=1.5pt 4,linecolor=black]{<-}(1.8,0)(0,-1)
\psline[linewidth=1pt,arrowsize=1.5pt 4,linecolor=black]{<-}(1.9,-0.1)(1,-1)
\psline[linewidth=1pt,arrowsize=1.5pt 4,linecolor=black,linestyle=dotted]{<-}(2,-0.2)(2,-1)
\psline[linewidth=1pt,arrowsize=1.5pt 4,linecolor=black]{<-}(2.1,-0.1)(3,-1)
\psline[linewidth=1pt,arrowsize=1.5pt 4,linecolor=black]{<-}(2.2,0)(4,-1)
\psline[linewidth=1pt,arrowsize=1.5pt 4,linecolor=black]{<-}(0,-1.2)(0,-1.8)
\psline[linewidth=1pt,arrowsize=1.5pt 4,linecolor=black]{<-}(1,-1.2)(1,-1.8)
\psline[linewidth=1pt,arrowsize=1.5pt 4,linecolor=black]{<-}(3,-1.2)(3,-1.8)
\psline[linewidth=1pt,arrowsize=1.5pt 4,linecolor=black]{<-}(4,-1.2)(4,-1.8)
\psline[linewidth=1pt,arrowsize=1.5pt 2,linecolor=black,linestyle=dotted]{<-}(0,-2.2)(0,-2.8)
\psline[linewidth=1pt,arrowsize=1.5pt 2,linecolor=black,linestyle=dotted]{<-}(1,-2.2)(1,-2.8)
\psline[linewidth=1pt,arrowsize=1.5pt 2,linecolor=black,linestyle=dotted]{<-}(3,-2.2)(3,-2.8)
\psline[linewidth=1pt,arrowsize=1.5pt 2,linecolor=black,linestyle=dotted]{<-}(4,-2.2)(4,-2.8)
\psline[linewidth=1pt,arrowsize=1.5pt 2,linecolor=black,linestyle=dotted]{<-}(0,-3.2)(0,-3.8)
\psline[linewidth=1pt,arrowsize=1.5pt 2,linecolor=black,linestyle=dotted]{<-}(1,-3.2)(1,-3.8)
\psline[linewidth=1pt,arrowsize=1.5pt 2,linecolor=black,linestyle=dotted]{<-}(3,-3.2)(3,-3.8)
\psline[linewidth=1pt,arrowsize=1.5pt 2,linecolor=black,linestyle=dotted]{<-}(4,-3.2)(4,-3.8)
\rput[mc](4.6,-2.5){\scalebox{2}[8.5]{\}}}
\rput[mc](5.4,-2.5){$\left\lfloor\frac{\Phi}{\gamma}-1\right\rfloor$}
\pscircle[linewidth=1pt,origin={0,-1},fillstyle=solid,fillcolor=white](0,0){0.2}
\pscircle[linewidth=1pt,origin={1,-1},fillstyle=solid,fillcolor=white](0,0){0.2}
\rput[mc](2,-1){\Large $\displaystyle\pmb{\ldots\ldots}$}
\pscircle[linewidth=1pt,origin={3,-1},fillstyle=solid,fillcolor=white](0,0){0.2}
\pscircle[linewidth=1pt,origin={4,-1},fillstyle=solid,fillcolor=white](0,0){0.2}
\pscircle[linewidth=1pt,origin={0,-2},fillstyle=none,fillcolor=lightgray](0,0){0.2}
\pscircle[linewidth=1pt,origin={1,-2},fillstyle=none,fillcolor=lightgray](0,0){0.2}
\rput[mc](2,-2){\Large $\displaystyle\pmb{\ldots\ldots}$}
\pscircle[linewidth=1pt,origin={3,-2},fillstyle=none,fillcolor=lightgray](0,0){0.2}
\pscircle[linewidth=1pt,origin={4,-2},fillstyle=none,fillcolor=lightgray](0,0){0.2}
\rput[mc](0,-2.9){\Large $\displaystyle\pmb{\vdots}$}
\rput[mc](1,-2.9){\Large $\displaystyle\pmb{\vdots}$}
\rput[mc](2,-3){\Large $\displaystyle\pmb{\ldots\ldots}$}
\rput[mc](3,-2.9){\Large $\displaystyle\pmb{\vdots}$}
\rput[mc](4,-2.9){\Large $\displaystyle\pmb{\vdots}$}
\pscircle[linewidth=1pt,origin={0,-4},fillstyle=none,fillcolor=lightgray](0,0){0.2}
\pscircle[linewidth=1pt,origin={1,-4},fillstyle=none,fillcolor=lightgray](0,0){0.2}
\rput[mc](2,-4){\Large $\displaystyle\pmb{\ldots\ldots}$}
\pscircle[linewidth=1pt,origin={3,-4},fillstyle=none,fillcolor=lightgray](0,0){0.2}
\pscircle[linewidth=1pt,origin={4,-4},fillstyle=none,fillcolor=lightgray](0,0){0.2}
\end{pspicture}
\caption{\label{tight1}A tight example for the bound in Lemma~\ref{bound1} ($\E=0$).}
\end{minipage}
\end{figure}

\begin{lemma}\label{bound1}
For any node $u$, $\big|\,\iz(u)\,\big|<1+\din(u)\left(\frac{\Phi}{\gamma}-1\right)$. 
\end{lemma}

\begin{proof}
For notational simplicity, let $\E=E/n$. 
If the node $u$ does not fail when shocked, or $u$ fails but it has no child, 
then $\big|\,\iz(u)\,\big|\leq 1$ and our claim holds since $\Phi>\gamma$. 
Otherwise, $u$ fails and each of its $\din(u)$ children at level~$2$ receives 
a part of the shock given by 
\begin{gather*}
\Game 
=
\min\left\{\,\frac{\Phi\left(\din(u)-1+\E\right) - \gamma\left(\din(u)+\E\right)}{\din(u)},\,1\right\}
\\
<
\Phi\left(1+\frac{\E}{\din(u)}\right)
-\gamma\left(1+\frac{\E}{\din(u)}\right)
\leq
\Phi\left(1+\E\right) - \gamma\left(1+\E\right)
\end{gather*}
Consider a child $v$ of $u$.  Each node $v'\in\Delta(v)$ that fails due to the shock subtracts an amount of $\gamma\left(\din(v')+\E\right)\geq\gamma\left(1+\E\right)$
from $\Game$ provided this subtraction does not result in a negative value. Thus, the total 
number of failed nodes is strictly less than 
$1+\din(u)\,\frac{\Phi\left(1+\E\right) - \gamma\left(1+\E\right)}{\gamma\left(1+\E\right)}=1+ \din(u) \left(\frac{\Phi}{\gamma}-1\right)$.
\end{proof}

\begin{remark}
The bound in Lemma~\ref{bound1} is tight as shown in \FI{tight1}. 
\end{remark}

Lemma~\ref{bound1} immediately implies that 
\[
\vi^\ast(G,T)
>
\frac{\displaystyle\dfrac{n}{ \displaystyle \max_{u\in V} \Big\{\, \iz(u)\, \Big\}  }}{n}
>
\dfrac{n\big/\left(1+\mathrm{deg}_{\mathrm{in}}^{\max} \left( \frac{\Phi}{\gamma}-1 \right) \right) }{n}
=
\dfrac{1}{1+\mathrm{deg}_{\mathrm{in}}^{\max} \left( \frac{\Phi}{\gamma}-1 \right)}
\]
We now provide a polynomial time algorithm to compute $\vi^\ast(G,T)$ exactly {\em assuming each node can be shocked to fail
individually}. For a node $u$, define the following:
\begin{itemize}
\item
For every node $u'\in\nabla(u)$, $\snsvi^\ast(G,T,u,u')$ is the number of nodes in an optimal solution of \nam\ 
for the subgraph induced by the nodes in $\Delta(u)$ (or $\infty$, if there is no feasible solution of \nam\ for this subgraph
under the stated conditions) assuming the following:
\begin{itemize}
\item
$u'$ was shocked,

\item
$u$ was {\em not} shocked, and 

\item
no node in the path $u'\leadsto u$ excluding $u'$ was shocked.
\end{itemize}

\item
$\ssvi^\ast(G,T,u)$ is the number of nodes in an optimal solution of \nam\ 
for the subgraph induced by the nodes in $\Delta(u)$ 
(or $\infty$, if there is no feasible solution of \nam\ under the stated conditions)\footnote{Intuitively, a value of $\infty$
signifies that the corresponding quantity is undefined.}
assuming that the node $u$ was shocked (and therefore failed).
\end{itemize}
We consider the usual partition of the nodes of $G$ into {\em levels}: $\level(r)=1$ and 
$\level(u)=\level(v)+1$ if $u$ is a child of $v$. 
We will compute $\ssvi^\ast(G,T,u)$ and $\snsvi^\ast(G,T,u,v)$
for the nodes $u$ level by level, starting with the 
highest level and proceeding to successive lower levels. By Observation~\ref{obs1}(a), the root $r$ must be shocked to fail 
for the entire network to fail, and thus $\ssvi^\ast(G,T,r)$ will provide us with our required optimal solution.

Every node $u$ at the highest level has $\din(u)=0$. In general, $\ssvi^\ast(G,T,u)$ and $\snsvi^\ast(G,T,u,u')$ 
can be computed for any node $u$ with $\din(u)=0$ as follows: 
\begin{description}
\item[Computing $\ssvi^\ast(G,T,u)$ when $\din(u)=0$:]
$\ssvi^\ast(G,T,u)=1$ by our assumption that every node can be shocked to fail.

\item[Computing $\snsvi^\ast(G,T,u,u')$ when $\din(u)=0$:]
$\,$

\begin{itemize}
\item
If $u\in\iz(u')$ then shocking node $v$ makes node $u$ fail. Since node $u$ fails without being shocked, we have $\snsvi^\ast(G,T,u,u')=0$.

\item
Otherwise, node $u$ does not fail. Thus, there is no feasible solution and $\snsvi^\ast(G,T,u,u')=\infty$.
\end{itemize}
\end{description}
Note that we only count the number of nodes in $\Delta(u)$ in the calculations of $\snsvi^\ast(G,T,u,u')$ and $\ssvi^\ast(G,T,u)$.

Now, consider a node $u$ at some level $\ell$ with $\din(u)>0$. Let $v_1,v_2,\dots,v_{\din(u)}$ be the children of $u$ at level $\ell+1$.
Note that $\nabla(v_1)=\nabla(v_2)=\dots=\nabla(v_{\din(u)})$.
\begin{description}
\item[Computing $\ssvi^\ast(G,T,u)$ when $\din(u)>0$:]
By our assumption, $u$ fails when shocked. Note that no node in $\Delta(u)\setminus\{u\}$ can receive any component of a shock
given to a node in $V\setminus\Delta(u)$ since $u$ failed.
For each child $v_i$ of $u$ we have two choices: 
$v_i$ is shocked and (and, therefore, fails), or 
$v_i$ is not shocked.
Thus, in this case we have 
$\ssvi^\ast(G,T,u) = 1 + \sum_{i=1}^{\din(u)} \min \Big\{\, \ssvi^\ast(G,T,v_i),\, \snsvi^\ast(G,T,v_i,u)\, \Big\}$.

\item[Computing $\snsvi^\ast(G,T,u,u')$ when $\din(u)>0$:]
Since $u'$ is shocked and $u$ is not shocked, the following cases arise:
\begin{itemize}
\item
If $u\not\in\iz(u')$ then then $u$ does not fail. 
Thus, there is no feasible solution for the subgraph induced by the nodes in $\Delta(u)$ under this
condition, and $\snsvi^\ast(G,T,u,u')=\infty$. 

\item
Otherwise, $u\in\iz(u')$, and therefore $u$ fails when $u'$ is shocked.
For each child $v_i$ of $u$, there are two options: $v_i$ is shocked and fails, or $v_i$ is not shocked.
Thus, in this case we have 
$\snsvi^\ast(G,T,u,u') = \sum_{i=1}^{\din(u)} \min \Big\{\, \ssvi^\ast(G,T,v_i),\, \snsvi^\ast(G,T,v_i,u')\,\Big\}$.
\end{itemize}
\end{description}
Let $\ell_{\max}$ be the maximum level number of any node in $G$.
Based on the above observations, we can design the dynamic programming algorithm as shown in \FI{dyn1} to compute an 
optimal solution of \nam\ on $G$.
It is easy to check that the running time of our algorithm is $O\left(n^2\right)$.

\begin{figure}
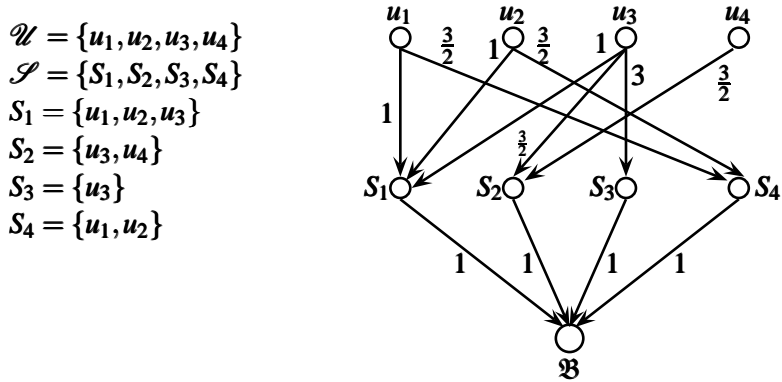

\begin{tabular}{p{6.5in}} 
\toprule
(* preprocessing *)
\\
$\forall u\in V\colon$ compute $\iz(u)$
\\
(* dynamic programming *)
\\
{\bf for} $\ell=\ell_{\max},\ell_{\max}\!-\!1,\dots,1$ {\bf do} 
\\
\hspace*{0.2in}
{\bf for} each node $u$ at level $\ell$ {\bf do}
\\
\hspace*{0.4in}
{\bf if} $\din(u)=0$ {\bf then} 
\\
\hspace*{0.6in}
$\ssvi^\ast(G,T,u)=1$ 
\\
\hspace*{0.6in}
$\forall u'\in \nabla(u)\colon$ {\bf if} $u\in\iz(u')$ {\bf then} $\snsvi^\ast(G,T,u,u')=0$ {\bf else} $\snsvi^ast(G,T,u,u')=\infty$ 
\\
\hspace*{0.4in}
{\bf else}  \hspace*{0.2in} (* $\din(u)>0$ *)
\\
\hspace*{0.6in}
$\ssvi^\ast(G,T,u) = 1 + \sum_{i=1}^{\din(u)} \min \Big\{\, \ssvi^\ast(G,T,v_i),\, \snsvi^\ast(G,T,v_i,u)\, \Big\}$
\\
\hspace*{0.6in}
$\forall u'\in \nabla(u)\colon$ {\bf if} $u\notin\iz(u')$ {\bf then} $\snsvi^\ast(G,T,u,u') = \infty$ 
\\
\hspace*{1.7in}
{\bf else} 
\\
\hspace*{1.8in}
$\snsvi^\ast(G,T,u,u') = \sum_{i=1}^{\din(u)} \min \Big\{\, \ssvi^\ast(G,T,v_i),\, \snsvi^\ast(G,T,v_i,u')\,\Big\}$
\\
\hspace*{1.5in}
{\bf endif} 
\\
\hspace*{0.4in}
{\bf endif} 
\\
\hspace*{0.2in}
{\bf endfor} 
\\
{\bf endfor}
\\
{\bf return} $\ssvi^\ast(G,T,r)$ as the solution \\
\bottomrule
\end{tabular}
\caption{\label{dyn1}A polynomial time algorithm to compute $\vi^\ast(G,T)$ when $G$ is a rooted in-arborescence and 
each node of $G$ fails individually when shocked.}
\end{figure}

\section{Heterogeneous Networks, $\pmb{\mbox{\nam}}$, Any $T$, Logarithmic Inapproximability}

\begin{theorem}\label{hetero-thm1}
Assuming $\NP\not\subset\mathsf{DTIME}\left(n^{\log\log n}\right)$, 
for any constant $0<\eps<1$ and any $T$, it is impossible to approximate $\vi^\ast(G,T)$
within a factor of $\,\,(1-\eps)\ln n\,\,$ in polynomial time even if $G$ is a DAG.
\end{theorem}

\begin{figure}[htbp]
\begin{center}
\begin{minipage}[c]{4in}
\begin{pspicture}(0,-2.5)(5,1.5)
\psset{xunit=1.5cm,yunit=1cm}
\rput[ml](0,1){$\pmb{\cU=\{u_1,u_2,u_3,u_4\}}$}
\rput[ml](0,0.5){$\pmb{\cS=\{S_1,S_2,S_3,S_4\}}$}
\rput[ml](0,0){$\pmb{S_1=\{u_1,u_2,u_3\}}$}
\rput[ml](0,-0.5){$\pmb{S_2=\{u_3,u_4\}}$}
\rput[ml](0,-1){$\pmb{S_3=\{u_3\}}$}
\rput[ml](0,-1.5){$\pmb{S_4=\{u_1,u_2\}}$}
\hspace*{-1.5in}
\pscircle[linewidth=1pt,origin={6,-1},fillstyle=none,fillcolor=lightgray](0,0){0.15}
\rput(5.8,-1){$\pmb{S_1}$}
\pscircle[linewidth=1pt,origin={7.5,-3},fillstyle=none,fillcolor=lightgray](0,0){0.2}
\rput(7.5,-3.4){$\pmb{\mathfrak{B}}$}
\psline[linewidth=1pt,arrowsize=1.5pt 4,linecolor=black]{->}(6,-1.15)(7.45,-2.85)
\rput[mr](6.6,-2){$\pmb{1}$}
\pscircle[linewidth=1pt,origin={7,-1},fillstyle=none,fillcolor=lightgray](0,0){0.15}
\rput(6.8,-1){$\pmb{S_2}$}
\psline[linewidth=1pt,arrowsize=1.5pt 4,linecolor=black]{->}(7,-1.15)(7.5,-2.85)
\rput[mr](7.2,-2){$\pmb{1}$}
\pscircle[linewidth=1pt,origin={8,-1},fillstyle=none,fillcolor=lightgray](0,0){0.15}
\rput(7.8,-1){$\pmb{S_3}$}
\psline[linewidth=1pt,arrowsize=1.5pt 4,linecolor=black]{->}(8,-1.15)(7.5,-2.85)
\rput[mr](7.95,-2){$\pmb{1}$}
\pscircle[linewidth=1pt,origin={9,-1},fillstyle=none,fillcolor=lightgray](0,0){0.15}
\rput[ml](9.15,-1){$\pmb{S_4}$}
\psline[linewidth=1pt,arrowsize=1.5pt 4,linecolor=black]{->}(9,-1.15)(7.55,-2.85)
\rput[mr](8.55,-2){$\pmb{1}$}
\pscircle[linewidth=1pt,origin={6,1},fillstyle=none,fillcolor=lightgray](0,0){0.15}
\rput(6,1.3){$\pmb{u_1}$}
\psline[linewidth=1pt,arrowsize=1.5pt 4,linecolor=black]{->}(6,0.85)(6,-0.85)
\rput[mr](5.95,0){$\pmb{1}$}
\psline[linewidth=1pt,arrowsize=1.5pt 4,linecolor=black]{->}(6,0.85)(8.9,-0.9)
\rput[mr](6.5,0.9){$\pmb{\frac{3}{2}}$}
\pscircle[linewidth=1pt,origin={7,1},fillstyle=none,fillcolor=lightgray](0,0){0.15}
\rput(7,1.3){$\pmb{u_2}$}
\psline[linewidth=1pt,arrowsize=1.5pt 4,linecolor=black]{->}(7,0.85)(6.05,-0.9)
\rput[mr](6.9,0.85){$\pmb{1}$}
\psline[linewidth=1pt,arrowsize=1.5pt 4,linecolor=black]{->}(7,0.85)(9.05,-0.85)
\rput[ml](7.2,0.9){$\pmb{\frac{3}{2}}$}
\pscircle[linewidth=1pt,origin={8,1},fillstyle=none,fillcolor=lightgray](0,0){0.15}
\rput(8,1.3){$\pmb{u_3}$}
\psline[linewidth=1pt,arrowsize=1.5pt 4,linecolor=black]{->}(8,0.85)(6.1,-1)
\rput[mr](7.85,0.93){$\pmb{1}$}
\psline[linewidth=1pt,arrowsize=1.5pt 4,linecolor=black]{->}(8,0.85)(7,-0.85)
\rput[mr](7.15,-0.4){$\scriptscriptstyle\pmb{\frac{3}{2}}$}
\psline[linewidth=1pt,arrowsize=1.5pt 4,linecolor=black]{->}(8,0.85)(8,-0.85)
\rput[ml](8.05,0.5){$\pmb{3}$}
\pscircle[linewidth=1pt,origin={9,1},fillstyle=none,fillcolor=lightgray](0,0){0.15}
\rput(9,1.3){$\pmb{u_4}$}
\psline[linewidth=1pt,arrowsize=1.5pt 4,linecolor=black]{->}(8.95,0.85)(7.1,-0.9)
\rput[ml](8.8,0.4){$\pmb{\frac{3}{2}}$}
\end{pspicture}
\end{minipage}
\hspace*{0.1in}
\begin{minipage}[c]{2in}
\caption{\label{fig-sc1}An instance $\langle\cU,\cS\rangle$ of \SC\ and its corresponding banking network $G=(V,F)$.}
\end{minipage}
\end{center}
\end{figure}

\begin{proof}
The (unweighted) \SC\ problem is defined as follows. We have an universe $\cU$ of $n$ elements, 
a collection of $m$ sets $\cS$ over $\cU$. 
The goal is to pick a sub-collection $\cS'\subseteq\cS$ containing a {\em minimum} number of sets 
such that these sets ``cover'' $\cU$, \IE, $\cup_{S\in\cS'}S=\cU$. 
It is known that there exists instances of \SC\
that cannot be approximated within a factor of $(1-\delta)\ln n$, for any constant $0<\delta<1$, unless 
$\NP\subseteq\mathsf{DTIME}\left(n^{\log\log n}\right)$~\cite{F98}.
Without any loss of generality, one may assume that every element $u\in\cU$ belongs to at least two sets in $\cS$ 
since otherwise the only set containing $u$ must be selected in any solution.

Given such an instance $\langle \cU,\cS\rangle$ of \SC, we now construct an instance of 
the banking network $G=(V,F)$ as follows:
\begin{itemize}
\item
We have a special node $\mathfrak{B}$.

\item
For every set $S\in\cS$, we have a node $S$, and a directed edge $(S,\mathfrak{B})$.  

\item
For every element $u\in\cU$, we have a node $u$, and directed edges $(u,S)$
for every set $S$ that contains $u$.
\end{itemize}
Thus, $|V|=n+m+1$, and $|F|<n\,m+m$. See~\FI{fig-sc1} for an illustration.
We set the shares of internal assets for each bank as follows: 
\begin{itemize}
\item
For each set $S\in\cS$, if $S$ contains $k>1$ elements then, for each element $u\in S$,
we set the weight of the edge $e=(u,S)$ as $w(e)=\frac{3}{k}$. 

\item
For each set $S\in\cS$, we set the weight of the edge $(S,\mathfrak{B})$ as $1$. 
\end{itemize}
Thus, $I=4m$. Also, observe that:
\begin{itemize}
\item
For any $S\in\cS$, $b_S=3$, and $\iota_S=1$.

\item
For any $u\in\cU$, $b_u=0$. Also, since $u$ belongs to at least two sets in $\cS$ and any set has at most $n-1$ elements, 
$\frac{2}{n}\leq\iota_u<\frac{3n}{2}$.

\item
$b_{\mathfrak{B}}=m$ and $\iota_{\mathfrak{B}}=0$.

\item
Since $\din(u)=0$ for any element $u\in\cU$, if a node $u$ is shocked, no part of the shock is propagated to any other node in the network.

\item
Since the longest path in $G$ has $2$ edges, by Proposition~\ref{obs1}(b) no new node in $G$ fails for $T>3$.
\end{itemize}
Let the share of external assets for a node (bank) $y$ be denoted by $E_y$ (thus, $\sum_{y\in V}E_y=E$).
We will select the remaining network parameters, namely $\gamma$, $\Phi$ and the $E_y$ values, based on the following properties.

\vspace*{0.1in}
\noindent
{\bf (I)}
If the node $\mathfrak{B}$ is shocked at $t=1$, it fails: 
\begin{gather}
\Phi\,( b_{\mathfrak{B}}-\iota_{\mathfrak{B}}+E_{\mathfrak{B}}\,) > \gamma\,( b_{\mathfrak{B}}+E_{\mathfrak{B}}\,)
\,\,\,\equiv\,\,\,
\Phi\,(m+E_{\mathfrak{B}}\,) > \gamma\,(m+E_{\mathfrak{B}}\,) 
\,\,\,\equiv\,\,\,
\Phi > \gamma
\label{eq1-new} 
\end{gather}

\noindent
{\bf (II)}
For any $S\in\cS$, 
if node $S$ is shocked at $t=1$, then $S$ fails at $t=1$, and,  
for every $u\in S$, node $u$ fails at time $t=2$: 
\begin{gather*}
\frac{\min\big\{\,\Phi\,\left( b_S-\iota_S+E_S\right) - \gamma\,\left( b_S+E_S \right),\,b_S\,\big\}}{\din(S)} 
> 
\gamma\,\left( b_u+E_u \right) \\
\,\,\,\equiv\,\,\,
\frac{\min\big\{\,\Phi\,(2+E_S) - \gamma\,(3+E_S),\,3\,\big\}}{|S|} > \gamma\,E_u
\end{gather*}
The above inequality is satisfied if:
\begin{gather}
\Phi\,(2+E_S) > \gamma\,(\,3+E_S+|S|\,E_u\,)
\label{eq2-1-new} 
\\
\Phi\,(2+E_S) - \gamma\,(3+E_S)\leq 3
\label{eq2-2-new} 
\end{gather}

\noindent
{\bf (III)}
For any $u\in\cU$, consider the node $u$, and let $S_1,S_2,\dots,S_p\in\cS$ be the $p$ sets that contain $u$. Then, we require that 
if the node $\mathfrak{B}$ is shocked at $t=1$ then $\mathfrak{B}$ fails at $t=1$, 
every node among the set of nodes $\{\,S_1,S_2,\dots,S_p\,\}$ that was not shocked at $t=1$ fails at $t=2$, 
but the node $u$ does not fail if the none of the nodes $u,S_1,S_2,\dots,S_p$ were shocked, 
This is satisfied provided the following inequalities hold:
\begin{description}
\item[(III-1)]
Any node among the set of nodes $\{\,S_1,S_2,\dots,S_p\,\}$ that was not shocked at $t=1$ fails at $t=2$. This is satisfies provided for any set $S\in\cS$ the following holds:
\begin{multline*}
\frac{\min\big\{\,\Phi\,\left( b_{\mathfrak{B}}-\iota_{\mathfrak{B}}+E_{\mathfrak{B}} \right) - \gamma\,\left( b_{\mathfrak{B}} +E_{\mathfrak{B}}\,\right),\,b_{\mathfrak{B}}\,\big\}}{\din(\mathfrak{B})} 
> 
\gamma\,\left( b_{S}+E_{S} \right) 
\\
\,\,\equiv\,\,
\min\left\{\,(\Phi-\gamma)\left( 1+\frac{E_{\mathfrak{B}}}{m} \right),\,1\,\right\} > \gamma\,(3+E_S)
\end{multline*}
The above inequality is satisfied provided:
\begin{gather}
(\Phi-\gamma) \left( 1+\frac{E_{\mathfrak{B}}}{m} \right) > \gamma\,(3+E_S)
\,\,\,\equiv\,\,\,
\Phi \left( 1+ \frac{E_{\mathfrak{B}}}{m} \right) > \gamma\, \left( 4+E_S+  \frac{E_{\mathfrak{B}}}{m} \right)
\label{eq3-new} 
\\
1 > \gamma\,(3+E_S) 
\,\,\,\equiv\,\,\,
\gamma < \frac{1}{3+E_S} 
\label{eq4-new} 
\end{gather}

\item[(III-2)]
$u$ does not fail if the none of the nodes $u,S_1,S_2,\dots,S_p$ were shocked: 
\begin{multline*}
\min \left\{ (\Phi-\gamma) \left(1+ \frac{E_{\mathfrak{B}}}{m} \right),\,1 \right\} -\gamma\,(3+E_S)
\leq
\frac{\gamma\,E_u}{n} 
\\
\,\,\,\equiv\,\,\,
\min \left\{ (\Phi-\gamma) \left(1+ \frac{E_{\mathfrak{B}}}{m} \right),\,1 \right\}
\leq
\gamma\,\left(3+E_{S}+\frac{E_u}{n}\right) 
\end{multline*}
The above inequality is satisfied provided:
\begin{gather}
(\Phi-\gamma) \left(1+ \frac{E_{\mathfrak{B}}}{m} \right)
\leq
\gamma\,\left(3+E_{S}+\frac{E_u}{n}\right) 
\,\,\,\equiv\,\,\,
\Phi\, \left(1+ \frac{E_{\mathfrak{B}}}{m} \right)
\leq
\gamma\,\left(4+E_S+ \frac{E_{\mathfrak{B}}}{m} +\frac{E_u}{n}\right) 
\label{eq5-new} 
\\
(\Phi-\gamma) \left(1+ \frac{E_{\mathfrak{B}}}{m} \right) \leq 1 
\,\,\,\equiv\,\,\,
\gamma
\geq
\Phi\,-\frac{1}{1+ \frac{E_{\mathfrak{B}}}{m} }
\label{eq6-new} 
\end{gather}
\end{description}
There are many choices of parameters $\gamma$, $\Phi$ and $E_y$'s satisfying 
Equations~\eqref{eq1-new}--\eqref{eq6-new}; we exhibit just one:
\[
\forall S\in\cS\colon E_S=0
\,\,\,\,\,\,\,\,\,\,\,\,\,\,\,\,\,\,\,
E_{\mathfrak{B}}=0
\,\,\,\,\,\,\,\,\,\,\,\,\,\,\,\,\,\,\,
\forall u\in\cU\colon E_u=\frac{1}{100n}
\,\,\,\,\,\,\,\,\,\,\,\,\,\,\,\,\,\,\,
\gamma=0.1
\,\,\,\,\,\,\,\,\,\,\,\,\,\,\,\,\,\,\,
\Phi=0.4+\frac{1}{n^{10000}} 
\]
Suppose that $\cS'\subset \cS$ is a solution of \SC. 
Then, we shock the node $\mathfrak{B}$ and the nodes $S$ for each $S\in\cS'$. 
By {\bf (I)} and {\bf (II)} the node $\mathfrak{B}$ and the nodes $S$ for each $S\in\cS'$ 
fails at $t=1$, and by {\bf (II)} the nodes $u$ for every $u\in\cU$ fails $t=2$.
Thus, we obtain a solution of $G$ by shocking $|\cS'|+1$ nodes.

Conversely, consider a solution of the \nam\ problem on $G$.
If a node $u$ for some $u\in\cU$ was shocked, we can instead shock the node $S$ for any set $S$ that contains 
$a$, which by {\bf (II)} still fails all the nodes in the network and does not increase the number of shocked nodes.
Thus, after such normalizations, we may assume that the shocked nodes consist of $\mathfrak{B}$ and 
a subset $\cS'\subseteq\cS$ of nodes.
By {\bf (II)} and {\bf (III)} for every node $u\in\cU$ at least one set that contains $u$ must be in $\cS'$.
Thus, the collection of sets in $\cS'$ form a cover of $\cU$ of size $|cS'|$.
\end{proof}

\section{Heterogeneous Networks, $\pmb{\mbox{{{{\sc Stab$_{2,\Phi}$}}}}}$, Logarithmic Approximation} 

For any positive real $x>0$, let $\overline{x}=\max\left\{x,\nicefrac{1}{x}\right\}$ and $\underline{x}=\min\left\{x,\nicefrac{1}{x}\right\}$.
Let 
$w_{\min}=\min_{e\colon w(e)>0} \big\{w(e)\big\}$, 
\linebreak
$w_{\max}=\max_e \big\{w(e)\big\}$, 
$\alpha_{\min}=\min_{v\colon \alpha_v>0} \big\{\alpha_v\big\}$, and 
$\alpha_{\max}=\max_v \big\{\alpha_v\big\}$.

\begin{theorem}\label{log2}
$\mathbf{\mbox{{{{\sc Stab$_{2,\Phi}$}}}}}$ admits 
a poly-time algorithm with 
approximation ratio
\linebreak
$\mathrm{O}\left( \log \dfrac{n\,\,\,\overline{E}\,\,\,\overline{w_{\max}}\,\,\,\overline{w_{\min}}\,\,\,\overline{\alpha_{\max}}}
{\Phi\,\,\gamma\,\, (\Phi-\gamma)\,\,\underline{E}\,\,\underline{w_{\min}}\,\,\underline{\alpha_{\min}}\,\,\underline{w_{\max}} } \, \right)$.
\end{theorem}

\begin{proof}
We can reuse the proof of the corresponding approximation for homogeneous networks in Theorem~\ref{log1} to obtain an approximation ratio of 
$2 + \ln n + \ln \left( \max_{v\in V} \left\{ \sum_{u\in V} \frac{\delta_{\,v,u}}{\zeta} \right\} \right)$, 
where $\displaystyle\zeta=\min_{u\in V} \big\{ \, \min_{v\in V} \{\delta_{\,u,v}\},\, c_u\,  \big\}$,
provided we recalculate $\max_{v\in V} \left\{ \sum_{u\in V} \frac{\delta_{\,v,u}}{\zeta} \right\}$.
Then,
\begin{gather*}
\min_{\substack{u\in V \\ \delta_{\,u,u}>0}} \left\{ \delta_{\,u,u} \right\}
=
\min_{\substack{u\in V \\ \delta_{u,u}>0}} \left\{  \Phi\, \left( \sum_{e=(v',u)\in F} \hspace*{-0.15in}w(e) - \hspace*{-0.2in}\sum_{e=(u,v')\in F} \hspace*{-0.15in}w(e) \,\,+ \alpha_v E \right) \right\}
=
\Omega\Big(\,\mathrm{poly} \left( s, \Phi, \underline{E}, \underline{\alpha_{\min}} \right)\,\Big)
\end{gather*}
\begin{multline*}
\min_{u\in V} \min_{\substack{v\in V \\ \delta_{\,u,v}>0 } } \left\{ \delta_{\,u,v} \right\} 
=
\min_{u\in V} \min_{\substack{v\in V \\ \Phi\,e_v>c_v} } 
\left\{\! (\Phi-\gamma) \left( \!\!\! 1+\frac{\alpha_v\,E}{\hspace*{-0.2in}\displaystyle \sum_{e=(v',v)\in F} \hspace*{-0.2in}w(e)} \right) 
- \Phi \,\frac{\hspace*{-0.3in}\displaystyle \sum_{\hspace*{0.3in}e=(v,v')\in F} \hspace*{-0.4in}w(e)}{\hspace*{-0.3in}\displaystyle \sum_{\hspace*{0.3in}e=(v',v)\in F} \hspace*{-0.4in}w(e) } \,\,\right\} 
\\
=
\Omega \Big( \mathrm{poly} \left( n^{-1}, \Phi-\gamma, \Phi , \underline{E} , \underline{w_{\max}} , \underline{w_{\min}}, \underline{\alpha_{\min}}\, \right) \Big)
\end{multline*}
\begin{gather*}
\min_{u\in V} \{c_u\}
=
\min_{u\in V} \left\{ \gamma\,\left(\displaystyle \sum_{e=(v',u)\in F} \hspace*{-0.2in}w(e)
+\alpha_u\,E \right) \right\}
=
\Omega \Big( \mathrm{poly} \left( n^{-1} , \gamma , \underline{E} ,\underline{\alpha_{\min}} , \underline{w_{\min}}\, \right) \, \Big)
\\
\zeta=\min \big\{ \, \min_{u\in V}\min_{v\in V} \{\delta_{\,u,v}\},\, \min_{u\in V} \{c_u\}\,  \big\}
=
\Omega \Big( \mathrm{poly} \left( n^{-1} , \Phi-\gamma , \Phi , \gamma , \underline{E}, \underline{w_{\min}} , \underline{\alpha_{\min}}, \underline{w_{\max}}\, \right) \Big)
\\
\max_{v\in V} \sum_{u\in V} \delta_{\,v,u} 
\leq n\,
\max_{u\in V}
\left\{\! (\Phi-\gamma) \left( \!\!\! 1+\frac{\alpha_v\,E}{\hspace*{-0.2in}\displaystyle \sum_{e=(v',v)\in F} \hspace*{-0.2in}w(e)} \right) 
- \Phi \,\frac{\hspace*{-0.3in}\displaystyle \sum_{\hspace*{0.3in}e=(v,v')\in F} \hspace*{-0.4in}w(e)}{\hspace*{-0.3in}\displaystyle \sum_{\hspace*{0.3in}e=(v',v)\in F} \hspace*{-0.4in}w(e) } \,\,\right\} 
=
\text{O} \Big(\,\mathrm{poly} \left( n , \overline{E} , \overline{w_{\max}} , \overline{w_{\min}} , \overline{\alpha_{\max}}\, \right) \, \Big)
\end{gather*}
and thus, 
\[
\max_{v\in V} \left\{ \sum_{u\in V} \frac{\delta_{\,v,u}}{\zeta} \right\} = \mathrm{O} \left( \mathrm{poly} \left( n ,  {\Phi}^{-1}, {\gamma}^{-1}, {(\Phi-\gamma)}^{-1} , \overline{E}, 
{\underline{E}}^{-1} , \overline{w_{\max}} , \overline{w_{\min}}, \overline{\alpha_{\max}} , {\underline{w_{\min}}}^{-1}, {\underline{\alpha_{\min}}}^{-1} , {\underline{w_{\max}}}^{-1} \right) \right)
\]
giving the desired approximation bound.
\end{proof}

\section{Heterogeneous Networks, $\pmb{\mbox{\nam}}$, $T>3$, Poly-logarithmic Inapproximability}

\begin{theorem}\label{hetero-thm2}
Assuming $\NP\not\subseteq\mathsf{DTIME}\left(n^{\mathrm{poly}(\log n)}\right)$, 
for any constant $0<\eps<1$ and any $T>3$, it is impossible to approximate $\vi^\ast(G,T)$
within a factor of $2^{\log^{1-\eps}n}$ in polynomial time even if $G$ is a DAG.
\end{theorem}

\begin{proof}
The \minrep problem (with minor modifications from the original setup) is defined as follows.
We are given a bipartite graph $G=(\vl,\vr,F)$ such that the degree of every node of $G$ 
is at least $10$, a partition of
$\vl$ into $\frac{|\vl|}{\alpha}$ equal-size subsets $\vl_1,\vl_2,\dots,\vl_\alpha$, and a partition of $\vr$
into $\frac{|\vr|}{\beta}$ equal-size subsets $\vr_1,\vr_2,\dots,\vr_\beta$. 

These partitions define a natural ``bipartite super-graph'' $\GS=(\VS,\ES)$ in the following manner.
$\GS$ has a ``super-node'' for every $\vl_i$ (for $i=1,2,\dots,\alpha$) and for every $\vr_j$ (for $j=1,2,\dots,\beta$). 
There exists an ``super-edge'' $h_{i,j}$ between the super-node for $\vl_i$ and the
super-node for $\vr_j$ if and only if there exists $u\in \vl_i$ and $v\in \vr_j$ such that
$\{u,v\}$ is an edge of $G$. 
A pair of nodes $u$ and $v$ of $G$ ``witnesses'' the super-edge $h_{i,j}$ of $H$ provided 
$u$ is in $\vl_i$, $v$ is in $\vr_j$ and the edge $\{u,v\}$ exists in $G$, and a set of nodes $V'\subseteq V$ of $G$
witnesses a super-edge if and only if there exists at least one pair of nodes in $S$ that witnesses
the super-edge. 

The goal of \minrep is to find $V_1\subseteq\vl$ and $V_2\subseteq\vr$
such that $V_1\cup V_2$ witnesses {\em every} super-edge of $H$ and the {\em size} of the solution, 
namely $|V_1|+|V_2|$, is {\em minimum}.
For notational simplicity, let $n=|\vl|+|\vr|$.
The following result is a consequence of Raz's parallel repetition theorem~\cite{R98,KKL04}. 

\begin{theorem}~{\rm\cite{KKL04}}
Let $L$ be any language in $\NP$ and $0<\delta<1$ be any constant. Then, there exists a reduction
running in $n^{\mathrm{poly}(\log n)}$ time that, given an input instance $x$ of $L$, produces an instance of \minrep such that:
\begin{itemize}
\item
if $x\in L$ then \minrep has a solution of size
$\alpha+\beta$; 

\item
if $x\not\in L$ then \minrep has a solution of size at least 
$(\alpha+\beta)\cdot 2^{\log^{1-\delta} n}$.
\end{itemize}
\end{theorem}
Thus, the above theorem provides a $2^{\log^{1-\delta} n}$-inapproximability 
for \minrep under the complexity-theoretic assumption of 
$\NP\not\subseteq\mathsf{DTIME}\left(n^{\rm polylog(n)}\right)$.

\begin{figure}[ht]
\begin{pspicture}(-0.5,-7.8)(32,5)
\psset{xunit=0.45cm,yunit=0.5cm}
\rput(30,8){\framebox{\parbox{1.5in}{$\pmb{G=(\vl,\vr,F)}$ \\ $\pmb{\GS=(\VS,\ES)}$}}}
\rput(30,9.7){\bf\minrep instance}
\rput(30,5.5){\scalebox{2}[2]{\rotatebox{225}{$\pmb{\Longrightarrow}$}}}
\pscircle[linewidth=0.5pt,origin={0,-11},fillstyle=none,fillcolor=white](0,0){0.15}
\rput[ml](0.6,-11.2){$\displaystyle\pmb{\vsi}$}
\rput(3,-5.5){$\pmb{\fgerightB}$}
\psframe[origin={0,0},linewidth=0.5pt,linestyle=dashed,linecolor=black](2.1,-6)(4,-16)
\rput(5.15,-5.5){$\pmb{\fgeleftB}$}
\psframe[origin={0,0},linewidth=0.5pt,linestyle=dashed,linecolor=black](4.3,-6)(6,-16)
\pscircle[linewidth=0.5pt,origin={3,-8},fillstyle=none,fillcolor=white](0,0){0.15}
\rput[ml](2.5,-7.4){$\pmb{\fgerightB_1}$}
\pscircle[linewidth=0.5pt,origin={3,-9},fillstyle=none,fillcolor=white](0,0){0.15}
\rput[ml](2.5,-9.7){$\pmb{\fgerightB_2}$}
\rput[mc](3,-11){$\displaystyle\pmb{\vdots}$}
\rput[mc](3,-12){$\displaystyle\pmb{\vdots}$}
\pscircle[linewidth=0.5pt,origin={3,-14},fillstyle=none,fillcolor=white](0,0){0.15}
\rput[ml](2.5,-14.7){$\pmb{\fgerightB_{|F|}}$}
\pscircle[linewidth=0.5pt,origin={5,-8},fillstyle=none,fillcolor=white](0,0){0.15}
\rput[ml](4.5,-7.4){$\pmb{\fgeleftB_1}$}
\pscircle[linewidth=0.5pt,origin={5,-9},fillstyle=none,fillcolor=white](0,0){0.15}
\rput[ml](4.5,-9.7){$\pmb{\fgeleftB_2}$}
\rput[mc](5,-11){$\displaystyle\pmb{\vdots}$}
\rput[mc](5,-12){$\displaystyle\pmb{\vdots}$}
\pscircle[linewidth=0.5pt,origin={5,-14},fillstyle=none,fillcolor=white](0,0){0.15}
\rput[ml](4.5,-14.7){$\pmb{\fgeleftB_{|F|}}$}
\psline[linewidth=0.5pt,arrowsize=1.5pt 6,linecolor=black]{<-}(0.1,-10.8)(2.75,-7.9)
\psline[linewidth=0.5pt,arrowsize=1.5pt 6,linecolor=black]{<-}(0.25,-11)(2.75,-9)
\psline[linewidth=0.5pt,arrowsize=1.5pt 6,linecolor=black]{<-}(0.1,-11.2)(2.75,-14.1)
\psline[linewidth=0.5pt,arrowsize=1.5pt 6,linecolor=black]{<-}(3.25,-8)(4.75,-8)
\psline[linewidth=0.5pt,arrowsize=1.5pt 6,linecolor=black]{<-}(3.25,-9)(4.75,-9)
\psline[linewidth=0.5pt,arrowsize=1.5pt 6,linecolor=black]{<-}(3.25,-14)(4.75,-14)
\pscurve[linewidth=0.5pt,arrowsize=1.5pt 6,linecolor=black]{<-}(5.25,-7.9)(7,-8)(8.8,-5)
\pscurve[linewidth=0.5pt,arrowsize=1.5pt 6,linecolor=black]{<-}(5.25,-8.9)(7,-9)(9.8,-5)
\pscurve[linewidth=0.5pt,arrowsize=1.5pt 6,linecolor=black]{<-}(5.2,-13.9)(7,-14.5)(10,-14.3)(23.8,-5)
\pscircle[linewidth=0.5pt,origin={16,8},fillstyle=none,fillcolor=white](0,0){0.4}
\rput(16,8){$\mathbf{\vt}$}
\rput(16,9.1){\bf Super-node}
\pscurve[linewidth=0.5pt,arrowsize=1.5pt 4,linecolor=black]{<-}(15.3,8.4)(3,4)(0,0.25)
\pscurve[linewidth=0.5pt,arrowsize=1.5pt 4,linecolor=black]{<-}(15.3,8.1)(4,4)(1,0.25)
\pscurve[linewidth=0.5pt,arrowsize=1.5pt 4,linecolor=black]{<-}(15.3,7.8)(5,4)(3,0.25)
\psline[linewidth=0.5pt,arrowsize=1.5pt 4,linecolor=black]{<-}(15.4,7.6)(7,0.35)
\psline[linewidth=0.5pt,arrowsize=1.5pt 4,linecolor=black]{<-}(15.5,7.5)(12,0.3)
\psline[linewidth=0.5pt,arrowsize=1.5pt 4,linecolor=black]{<-}(15.7,7.4)(13,0.3)
\psline[linewidth=0.5pt,arrowsize=1.5pt 4,linecolor=black]{<-}(15.9,7.3)(15,0.3)
\psline[linewidth=0.5pt,arrowsize=1.5pt 4,linecolor=black]{<-}(16.1,7.3)(17,0.3)
\psline[linewidth=0.5pt,arrowsize=1.5pt 4,linecolor=black]{<-}(16.3,7.3)(18,0.3)
\psline[linewidth=0.5pt,arrowsize=1.5pt 4,linecolor=black]{<-}(16.5,7.4)(20,0.3)
\psline[linewidth=0.5pt,arrowsize=1.5pt 4,linecolor=black]{<-}(16.5,7.4)(24,0.35)
\pscurve[linewidth=0.5pt,arrowsize=1.5pt 4,linecolor=black]{<-}(16.7,7.8)(26,4)(29,0.25)
\pscurve[linewidth=0.5pt,arrowsize=1.5pt 4,linecolor=black]{<-}(16.7,8.1)(27,4)(30,0.25)
\pscurve[linewidth=0.5pt,arrowsize=1.5pt 4,linecolor=black]{<-}(16.7,8.4)(29,4)(32,0.25)
\pscircle[linewidth=0.5pt,origin={0,0},fillstyle=none,fillcolor=lightgray](0,0){0.15}
\psline[linewidth=0.5pt,arrowsize=1.5pt 4,linecolor=black](-0.25,-0.15)(-1.5,-0.75)
\psline[linewidth=0.5pt,arrowsize=1.5pt 4,linecolor=black](-0.15,-0.2)(-0.75,-1.25)
\psline[linewidth=0.5pt,arrowsize=1.5pt 4,linecolor=black](0,-0.3)(0,-1.25)
\pscircle[linewidth=0.5pt,origin={1,0},fillstyle=none,fillcolor=lightgray](0,0){0.15}
\psline[linewidth=0.5pt,arrowsize=1.5pt 4,linecolor=black,origin={1,0}](0.25,-0.15)(0.75,-1.25)
\psline[linewidth=0.5pt,arrowsize=1.5pt 4,linecolor=black,origin={1,0}](-0.15,-0.2)(-0.75,-1.25)
\psline[linewidth=0.5pt,arrowsize=1.5pt 4,linecolor=black,origin={1,0}](0,-0.3)(0,-1.25)
\rput(2,0){$\pmb{\cdots}$}
\pscircle[linewidth=0.5pt,origin={3,0},fillstyle=none,fillcolor=lightgray](0,0){0.15}
\psline[linewidth=0.5pt,arrowsize=1.5pt 4,linecolor=black,origin={3,0}](0.25,-0.15)(0.75,-1.25)
\psline[linewidth=0.5pt,arrowsize=1.5pt 4,linecolor=black,origin={3,0}](-0.15,-0.2)(-0.75,-1.25)
\psline[linewidth=0.5pt,arrowsize=1.5pt 4,linecolor=black,origin={3,0}](0,-0.3)(0,-1.25)
\psframe[origin={0,0},linewidth=1pt,linestyle=dashed,linecolor=gray](-0.5,-0.75)(3.5,0.75)
\rput(-0.6,1){$\pmb{\ovr{\vl_1}}$}
\rput(4,0){$\pmb{\cdots}$}
\rput(5,0){$\pmb{\cdots}$}
\rput(6,0){$\pmb{\cdots}$}
\pscircle[linewidth=0.5pt,origin={7,0},fillstyle=none,fillcolor=lightgray](0,0){0.2}
\rput(6.45,0.5){$\pmb{\ovr{u}}$}
\rput(8,0){$\pmb{\cdots}$}
\rput(9,0){$\pmb{\cdots}$}
\psframe[origin={0,0},linewidth=1pt,linestyle=dashed,linecolor=gray](5.5,-0.75)(9.5,1.25)
\rput(7,1.85){$\pmb{\ovr{\vl_i}}$}
\rput(10,0){$\pmb{\cdots}$}
\rput(11,0){$\pmb{\cdots}$}
\pscircle[linewidth=0.5pt,origin={12,0},fillstyle=none,fillcolor=lightgray](0,0){0.15}
\psline[linewidth=0.5pt,arrowsize=1.5pt 4,linecolor=black,origin={12,0}](-0.25,-0.15)(-1.5,-0.75)
\psline[linewidth=0.5pt,arrowsize=1.5pt 4,linecolor=black,origin={12,0}](-0.15,-0.2)(-0.75,-1.25)
\psline[linewidth=0.5pt,arrowsize=1.5pt 4,linecolor=black,origin={12,0}](0,-0.3)(0,-1.25)
\pscircle[linewidth=0.5pt,origin={13,0},fillstyle=none,fillcolor=lightgray](0,0){0.15}
\psline[linewidth=0.5pt,arrowsize=1.5pt 4,linecolor=black,origin={13,0}](0.25,-0.15)(0.75,-1.25)
\psline[linewidth=0.5pt,arrowsize=1.5pt 4,linecolor=black,origin={13,0}](-0.15,-0.2)(-0.75,-1.25)
\psline[linewidth=0.5pt,arrowsize=1.5pt 4,linecolor=black,origin={13,0}](0,-0.3)(0,-1.25)
\rput(14,0){$\pmb{\cdots}$}
\pscircle[linewidth=0.5pt,origin={15,0},fillstyle=none,fillcolor=lightgray](0,0){0.15}
\psline[linewidth=0.5pt,arrowsize=1.5pt 4,linecolor=black,origin={15,0}](0.25,-0.15)(0.75,-1.25)
\psline[linewidth=0.5pt,arrowsize=1.5pt 4,linecolor=black,origin={15,0}](-0.15,-0.2)(-0.75,-1.25)
\psline[linewidth=0.5pt,arrowsize=1.5pt 4,linecolor=black,origin={15,0}](0,-0.3)(0,-1.25)
\psframe[origin={0,0},linewidth=1pt,linestyle=dashed,linecolor=gray](11.5,-0.75)(15.5,0.75)
\rput(14.3,1.3){$\pmb{\ovr{\vl_\alpha}}$}
\rput(7.4,3){\scalebox{18.3}[3]{\rotatebox{-90}{\color{gray}$\mathbf{\{}$}}}
\rput(7.5,3.8){$\pmb{\ovr{\vl}}$}
\pscircle[linewidth=0.5pt,origin={17,0},fillstyle=none,fillcolor=darkgray](0,0){0.15}
\psline[linewidth=0.5pt,arrowsize=1.5pt 4,linecolor=black,origin={17,0}](-0.25,-0.15)(-1.1,-0.8)
\psline[linewidth=0.5pt,arrowsize=1.5pt 4,linecolor=black,origin={17,0}](-0.15,-0.2)(-0.75,-1.25)
\psline[linewidth=0.5pt,arrowsize=1.5pt 4,linecolor=black,origin={17,0}](0,-0.3)(0,-1.25)
\pscircle[linewidth=0.5pt,origin={18,0},fillstyle=none,fillcolor=darkgray](0,0){0.15}
\psline[linewidth=0.5pt,arrowsize=1.5pt 4,linecolor=black,origin={18,0}](0.25,-0.15)(0.75,-1.25)
\psline[linewidth=0.5pt,arrowsize=1.5pt 4,linecolor=black,origin={18,0}](-0.15,-0.2)(-0.75,-1.25)
\psline[linewidth=0.5pt,arrowsize=1.5pt 4,linecolor=black,origin={18,0}](0,-0.3)(0,-1.25)
\rput(19,0){$\pmb{\cdots}$}
\pscircle[linewidth=0.5pt,origin={20,0},fillstyle=none,fillcolor=darkgray](0,0){0.15}
\psline[linewidth=0.5pt,arrowsize=1.5pt 4,linecolor=black,origin={20,0}](0.25,-0.15)(0.75,-1.25)
\psline[linewidth=0.5pt,arrowsize=1.5pt 4,linecolor=black,origin={20,0}](-0.15,-0.2)(-0.75,-1.25)
\psline[linewidth=0.5pt,arrowsize=1.5pt 4,linecolor=black,origin={20,0}](0,-0.3)(0,-1.25)
\psframe[origin={17,0},linewidth=1pt,linestyle=dashed,linecolor=gray](-0.5,-0.75)(3.5,0.75)
\rput(20.5,1.4){$\pmb{\ovr{\vr_1}}$}
\rput(21,0){$\pmb{\cdots}$}
\rput(22,0){$\pmb{\cdots}$}
\rput(23,0){$\pmb{\cdots}$}
\pscircle[linewidth=0.5pt,origin={24,0},fillstyle=none,fillcolor=darkgray](0,0){0.2}
\rput(24.55,0.5){$\pmb{\ovr{v}}$}
\rput(25,0){$\pmb{\cdots}$}
\rput(26,0){$\pmb{\cdots}$}
\psframe[origin={17,0},linewidth=1pt,linestyle=dashed,linecolor=gray](5.5,-0.75)(9.5,1.25)
\rput(25,2){$\pmb{\ovr{\vr_j}}$}
\rput(27,0){$\pmb{\cdots}$}
\rput(28,0){$\pmb{\cdots}$}
\pscircle[linewidth=0.5pt,origin={29,0},fillstyle=none,fillcolor=darkgray](0,0){0.15}
\psline[linewidth=0.5pt,arrowsize=1.5pt 4,linecolor=black,origin={29,0}](-0.25,-0.15)(-1.5,-0.75)
\psline[linewidth=0.5pt,arrowsize=1.5pt 4,linecolor=black,origin={29,0}](-0.15,-0.2)(-0.75,-1.25)
\psline[linewidth=0.5pt,arrowsize=1.5pt 4,linecolor=black,origin={29,0}](0,-0.3)(0,-1.25)
\pscircle[linewidth=0.5pt,origin={30,0},fillstyle=none,fillcolor=darkgray](0,0){0.15}
\psline[linewidth=0.5pt,arrowsize=1.5pt 4,linecolor=black,origin={30,0}](0.25,-0.15)(0.75,-1.25)
\psline[linewidth=0.5pt,arrowsize=1.5pt 4,linecolor=black,origin={30,0}](-0.15,-0.2)(-0.75,-1.25)
\psline[linewidth=0.5pt,arrowsize=1.5pt 4,linecolor=black,origin={30,0}](0,-0.3)(0,-1.25)
\rput(31,0){$\pmb{\cdots}$}
\pscircle[linewidth=0.5pt,origin={32,0},fillstyle=none,fillcolor=darkgray](0,0){0.15}
\psline[linewidth=0.5pt,arrowsize=1.5pt 4,linecolor=black,origin={32,0}](0.25,-0.15)(0.75,-1.25)
\psline[linewidth=0.5pt,arrowsize=1.5pt 4,linecolor=black,origin={32,0}](-0.15,-0.2)(-0.75,-1.25)
\psline[linewidth=0.5pt,arrowsize=1.5pt 4,linecolor=black,origin={32,0}](0,-0.3)(0,-1.25)
\psframe[origin={17,0},linewidth=1pt,linestyle=dashed,linecolor=gray](11.5,-0.75)(15.5,0.75)
\rput(32.5,1.4){$\pmb{\ovr{\vr_\beta}}$}
\rput(24.4,3){\scalebox{18.3}[3]{\rotatebox{-90}{\color{gray}$\mathbf{\{}$}}}
\rput(24,3.8){$\pmb{\ovr{\vr}}$}
\pscircle[linewidth=0.5pt,origin={9,-5},fillstyle=none,fillcolor=lightgray](0,0){0.15}
\psline[linewidth=0.5pt,arrowsize=1.5pt 4,linecolor=black,origin={9,-5}](-0.25,0.15)(-1.5,0.75)
\psline[linewidth=0.5pt,arrowsize=1.5pt 4,linecolor=black,origin={9,-5}](-0.15,0.2)(-0.75,1.25)
\psline[linewidth=0.5pt,arrowsize=1.5pt 4,linecolor=black,origin={9,-5}](0,0.3)(0,1.25)
\pscircle[linewidth=0.5pt,origin={10,-5},fillstyle=none,fillcolor=lightgray](0,0){0.15}
\psline[linewidth=0.5pt,arrowsize=1.5pt 4,linecolor=black,origin={10,-5}](0.25,0.15)(1,1.25)
\psline[linewidth=0.5pt,arrowsize=1.5pt 4,linecolor=black,origin={10,-5}](-0.15,0.2)(-0.75,1.25)
\psline[linewidth=0.5pt,arrowsize=1.5pt 4,linecolor=black,origin={10,-5}](0,0.3)(0,1.25)
\rput(11,-5){$\pmb{\cdots}$}
\pscircle[linewidth=0.5pt,origin={12,-5},fillstyle=none,fillcolor=lightgray](0,0){0.15}
\psline[linewidth=0.5pt,arrowsize=1.5pt 4,linecolor=black,origin={12,-5}](0.25,0.15)(1,1.25)
\psline[linewidth=0.5pt,arrowsize=1.5pt 4,linecolor=black,origin={12,-5}](-0.15,0.2)(-0.75,1.25)
\psline[linewidth=0.5pt,arrowsize=1.5pt 4,linecolor=black,origin={12,-5}](0,0.3)(0,1.25)
\psframe[origin={9,-5},linewidth=1pt,linestyle=dashed,linecolor=gray](-0.5,-0.75)(3.5,0.75)
\rput(13,-6.4){$\pmb{\ovr{F_{1,1}}}$}
\rput(13,-5){$\pmb{\cdots}$}
\rput(14,-5){$\pmb{\cdots\cdot\cdot\cdot}$}
\pscircle[linewidth=0.5pt,origin={16,-5},fillstyle=none,fillcolor=lightgray](0,0){0.45}
\rput(16,-5.1){$\scriptstyle\pmb{{f}_{\,\ovr{u},\ovr{v}}}$}
\psline[linewidth=0.5pt,arrowsize=1.5pt 8,linecolor=black]{<-}(7.1,-0.3)(15.7,-4.2)
\psline[linewidth=0.5pt,arrowsize=1.5pt 8,linecolor=black]{<-}(23.9,-0.3)(16.3,-4.2)
\rput(18,-5){$\pmb{\cdot\cdots}$}
\psframe[origin={9,-5.5},linewidth=1pt,linestyle=dashed,linecolor=gray](5.5,-0.75)(9.5,1.5)
\rput(17.7,-5.7){$\pmb{\ovr{F_{i,j}}}$}
\rput(19,-5){$\pmb{\cdots}$}
\rput(20,-5){$\pmb{\cdots}$}
\pscircle[linewidth=0.5pt,origin={21,-5},fillstyle=none,fillcolor=lightgray](0,0){0.15}
\psline[linewidth=0.5pt,arrowsize=1.5pt 4,linecolor=black,origin={21,-5}](-0.25,0.15)(-1.5,0.75)
\psline[linewidth=0.5pt,arrowsize=1.5pt 4,linecolor=black,origin={21,-5}](-0.15,0.2)(-0.75,1.25)
\psline[linewidth=0.5pt,arrowsize=1.5pt 4,linecolor=black,origin={21,-5}](0,0.3)(0,1.25)
\pscircle[linewidth=0.5pt,origin={22,-5},fillstyle=none,fillcolor=lightgray](0,0){0.15}
\psline[linewidth=0.5pt,arrowsize=1.5pt 4,linecolor=black,origin={22,-5}](0.25,0.15)(1,1.25)
\psline[linewidth=0.5pt,arrowsize=1.5pt 4,linecolor=black,origin={22,-5}](-0.15,0.2)(-0.75,1.25)
\psline[linewidth=0.5pt,arrowsize=1.5pt 4,linecolor=black,origin={22,-5}](0,0.3)(0,1.25)
\rput(23,-5){$\pmb{\cdots}$}
\pscircle[linewidth=0.5pt,origin={24,-5},fillstyle=none,fillcolor=lightgray](0,0){0.15}
\psline[linewidth=0.5pt,arrowsize=1.5pt 4,linecolor=black,origin={24,-5}](0.25,0.15)(1,1.25)
\psline[linewidth=0.5pt,arrowsize=1.5pt 4,linecolor=black,origin={24,-5}](-0.15,0.2)(-0.75,1.25)
\psline[linewidth=0.5pt,arrowsize=1.5pt 4,linecolor=black,origin={24,-5}](0,0.3)(0,1.25)
\psframe[origin={9,-5},linewidth=1pt,linestyle=dashed,linecolor=gray](11.5,-0.75)(15.5,0.75)
\rput(25,-6.25){$\pmb{\ovr{F_{\alpha,\beta}}}$}
%
\pscircle[linewidth=0.5pt,origin={10.5,-13},fillstyle=none,fillcolor=white](0,0){0.4}
\rput(10.5,-13){$\scriptstyle\pmb{\ovr{h_{1,1}}}$}
\rput(12,-13){$\pmb{\cdots}$}
\rput(13,-13){$\pmb{\cdots}$}
\rput(14,-13){$\pmb{\cdots}$}
\rput(15,-13){$\pmb{\cdots}$}
\pscircle[linewidth=0.5pt,origin={16.5,-13},fillstyle=none,fillcolor=white](0,0){0.4}
\rput(16.5,-13){$\scriptstyle\pmb{\ovr{h_{i,j}}}$}
\rput(18,-13){$\pmb{\cdots}$}
\rput(19,-13){$\pmb{\cdots}$}
\rput(20,-13){$\pmb{\cdots}$}
\rput(21,-13){$\pmb{\cdots}$}
\pscircle[linewidth=0.5pt,origin={22.5,-13},fillstyle=none,fillcolor=white](0,0){0.4}
\rput(22.5,-13){$\scriptstyle\pmb{\ovr{h_{\alpha,\beta}}}$}
\rput(16.6,-14.7){\scalebox{15.2}[4]{\rotatebox{90}{\color{gray}$\mathbf{[}$}}}
\rput(16.6,-15.4){$\pmb{\big|\,\ES\,\big|}$}
\psline[linewidth=0.5pt,arrowsize=1.5pt 8,linecolor=black]{<-}(9,-5.15)(10.2,-12.35)
\psline[linewidth=0.5pt,arrowsize=1.5pt 8,linecolor=black]{<-}(10,-5.15)(10.5,-12.2)
\rput(11,-9){$\pmb{\cdots}$}
\psline[linewidth=0.5pt,arrowsize=1.5pt 8,linecolor=black]{<-}(12,-5.15)(10.7,-12.35)
\psline[linewidth=0.5pt,arrowsize=1.5pt 8,linecolor=black,linestyle=solid](15.8,-11.5)(16.2,-12.2)
\psline[linewidth=0.5pt,arrowsize=1.5pt 8,linecolor=black]{<-}(16,-5.8)(16.5,-12.2)
\psline[linewidth=0.5pt,arrowsize=1.5pt 8,linecolor=black,linestyle=solid](17,-11.5)(16.8,-12.2)
\rput(15.7,-9){$\pmb{\cdots}$}
\rput(17,-9){$\pmb{\cdots}$}
\psline[linewidth=0.5pt,arrowsize=1.5pt 8,linecolor=black]{<-}(21,-5.15)(22.2,-12.35)
\psline[linewidth=0.5pt,arrowsize=1.5pt 8,linecolor=black]{<-}(22,-5.15)(22.5,-12.2)
\rput(23,-9){$\pmb{\cdots}$}
\psline[linewidth=0.5pt,arrowsize=1.5pt 8,linecolor=black]{<-}(24,-5.15)(22.7,-12.35)
\end{pspicture}
\caption{\label{minrep}Reduction of an instance of \minrep to \nam\ for heterogeneous networks.}
\end{figure}
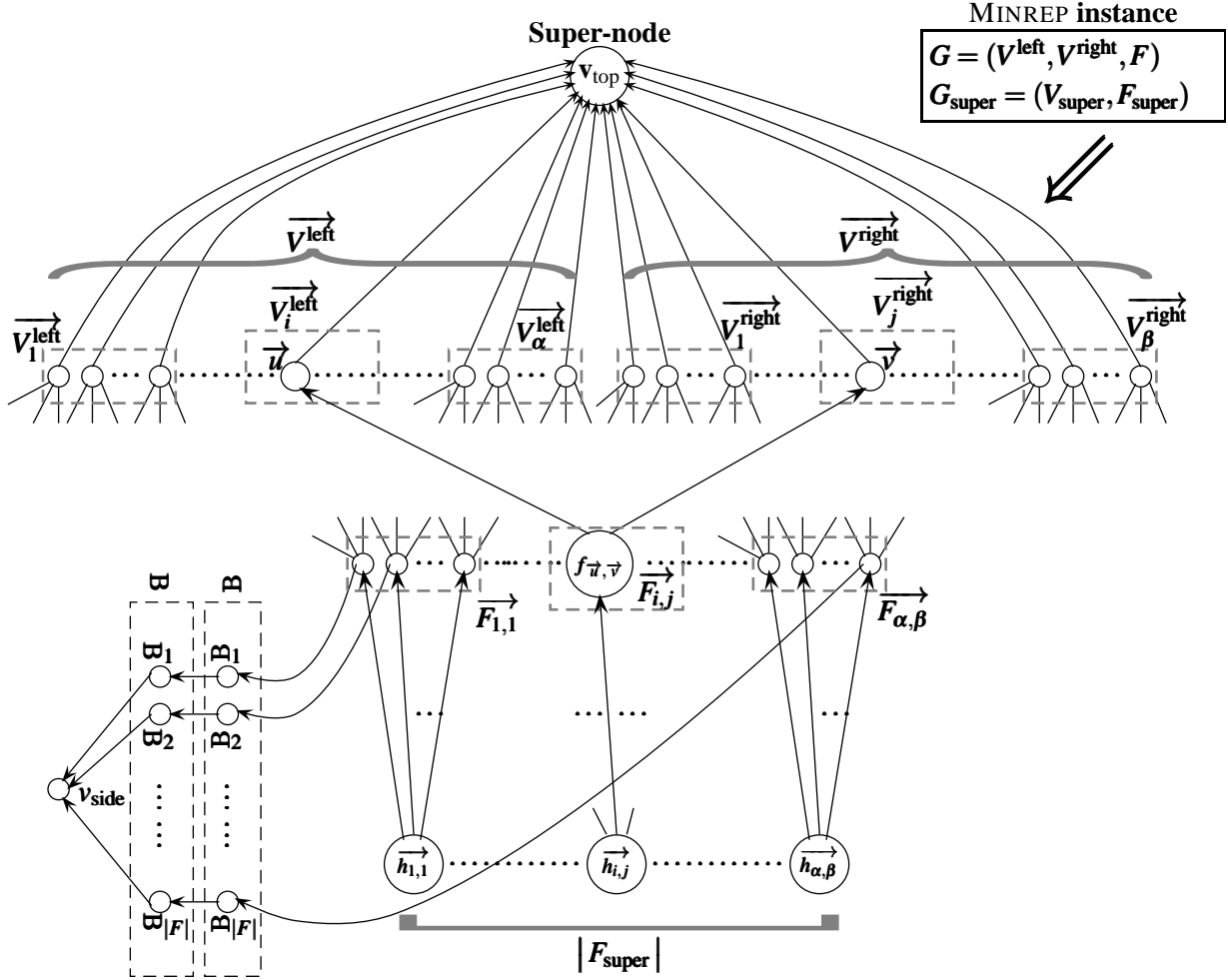

Let $F_{i,j}=\left\{\{u,v\}\,\big|\,u\in\vl_i,\,v\in\vr_j,\,\{u,v\}\in F\right\}$.
We now show our construction of an instance of \nam\ from an instance of \minrep$\!\!\!$.
Our directed graph $\ovr{G}=(\ovr{V},\ovr{F})$ for \nam\ is constructed as follows (see~\FI{minrep} for an illustration):

\vspace*{0.1in}
\noindent
{\sf Nodes}: 
\begin{itemize}
\item
For every node $u\in\vl_i$ of $G$ we have a corresponding node $\ovr{u}$ in the set of nodes $\ovr{\vl_i}$ in $\ovr{G}$, and 
for every node $v\in\vr_j$ of $G$ we have a corresponding node $\ovr{v}$ in the set of nodes $\ovr{\vr_j}$ in $\ovr{G}$.
The total number of such nodes is $n$.

\item
For every edge $\{u,v\}$ of $G$ with $u\in\vl_i$ and $v\in\vr_j$, we have a corresponding node $f_{\ovr{u},\ovr{v}}$ in the set 
of nodes $\ovr{F_{i,j}}$ in $\ovr{G}$. There are $|F|$ such nodes.

\item
For every super-edge $h_{i,j}$ of $\GS$, we have a node $\ovr{h_{i,j}}$ in $\ovr{G}$. There are $|\ES|$ such nodes.

\item
We have one ``top super-node'' $\vt$, 
one ``side super-node'' $\vsi$, and $2\,|F|$ additional nodes 
$\fgerightB_1,\fgerightB_2,\dots,\fgerightB_{|F|}$, $\fgeleftB_1,\fgeleftB_2,\dots,\fgeleftB_{|F|}$.
Let $\fgerightB=\cup_{j=1}^{|F|}\fgerightB_j$ and $\fgeleftB=\cup_{j=1}^{|F|}\fgeleftB_j$.
\end{itemize}
Thus, $n+3|F|+2<|\ovr{V}|=n+|F|+|\ES|+2+2\,|F|<n+4\,|F|+2$. 

\vspace*{0.1in}
\noindent
{\sf Edges}: 
\begin{itemize}
\item
For every node $u$ of $G$, we have an edge $\left(u,\vt\right)$ in $\ovr{G}$. There are $n$ such edges.

\item
For every edge $\{u,v\}$ of $G$, we have two edges $(f_{\ovr{u},\ovr{v}},\ovr{u})$ and $(f_{\ovr{u},\ovr{v}},\ovr{v})$ in $\ovr{G}$. There are $2\,|F|$ such edges.

\item
For every super-edge $h_{i,j}$ of $\GS$ and for every edge $f_{u,v}$ in $F_{i,j}$, we have an edge $\left(\ovr{h_{i,j}},f_{\ovr{u},\ovr{v}}\right)$ in $\ovr{G}$. 
There are $|F|$ such edges.

\item
Let $p_1,p_2,\dots,p_{|F|}$ be any arbitrary ordering of the edges in $F$. Then, 
for every $j=1,2,\dots,|F|$, we have the edges $(\vsi,\fgerightB_j)$, $(\fgerightB_j,\fgeleftB_j)$ and $(\fgeleftB_j,p_j)$.
The total number of such edges is $3|F|$. 
\end{itemize}
Thus, $|\ovr{E}|=n+6\,|F|$.

\vspace*{0.1in}
\noindent
{\sf Distribution of internal assets}: 
We set the weight of every edge to $1$, 
Thus, $I=n+\sum_{u\in\vl\cup\vr}\deg(u)+4\,|F|=n+6\,|F|$. 

\vspace*{0.1in}
\noindent
Let $\deg(u)\geq 10$ be the degree of node $u\in\vl\cup\vr$. Observe that:
\begin{itemize}
\item
$b_{\vt}=n$, and $\iota_{\vt}=0$.
Since $\dout\left(\vt\right)=0$, by Proposition~\ref{obs1}(a) the node $\vt$ must be shocked to make the network fail.

\item
$b_{\vsi}=|F|$, and $\iota_{\vsi}=0$.
Since $\dout\left(\vsi\right)=0$, by Proposition~\ref{obs1}(a) the node $\vsi$ must be shocked to make the network fail.

\item
For any $u\in\vl\cup\vr$, $b_{\ovr{u}}=\deg(u)$ and $\iota_{\ovr{u}}=1$.

\item
For any node $f_{\ovr{u},\ovr{v}}$, $b_{f_{\ovr{u},\ovr{v}}}=1$ and $\iota_{f_{\ovr{u},\ovr{v}}}=2$.

\item
For every node $\ovr{h_{i,j}}$, $b_{\ovr{h_{i,j}}}=0$ and $\iota_{\ovr{h_{i,j}}}=|F_{i,j}|$.
Since $\din\left(\ovr{h_{i,j}}\right)=0$ for any node $\ovr{h_{i,j}}$, if such a node is shocked, no part of the shock is propagated to any other node in the network.

\item
For every $j$, $b_{\fgeleftB_j}=\iota_{\fgeleftB_j}=b_{\fgerightB_j}=\iota_{\fgerightB_j}=1$. 

\item
Since the longest directed path in $G$ has $4$ edges, by Proposition~\ref{obs1}(b) no new node in $G$ fails for $t>4$.
\end{itemize}

{
\begin{table}[p]
\footnotesize
\vspace*{-0.4in}
\hspace*{-0.7in}
\begin{tabular}{rl|rl|rl|rl}
\toprule
$\Phi > \gamma$ & \eqref{eq1-het} & $\displaystyle\Phi > \gamma\, \left( 1 + \frac { \deg(u) + E_{\ovr{u}} } { 1 + \frac{E_{\vt}}{n} } \right)$ & \eqref{eq2-het}
            & $\displaystyle\Phi \leq \gamma + \frac { 1 } { 1 + \frac { E_{\vt} } { n } }$ & \eqref{eq3-het} 
            & $\displaystyle\Phi > \gamma\, \left(\, 1 + \frac {1} {\deg(u) - 1 + E_{\ovr{u}} } \right)$ & \eqref{eq4-het} \\ \midrule
\multicolumn{7}{c}{$\displaystyle \frac { \Phi\, \left( \deg(u) - 1 + E_{\ovr{u}} \right)  - \gamma\,( \deg(u) + E_{\ovr{u}}\,) } { \deg(u) }
\,+\,
\frac { \Phi\, \left( \deg(v) - 1 + E_{\ovr{v}} \right)  - \gamma\,( \deg(v) + E_{\ovr{v}}\,) } { \deg(v) }
>
\gamma \, \left( 1 + E_{f_{\ovr{u},\ovr{v}}} \right)$} & \eqref{eq5-het} \\ \midrule
\multicolumn{7}{c}{$\displaystyle\Phi \leq \gamma \, \left( \frac { \deg(u) + E_{\ovr{u}} } { \deg(u) - 1 + E_{\ovr{u}} }   \right) + \frac { \deg(u) } { \deg(u) - 1 + E_{\ovr{u}} }$} & 
                     \multicolumn{1}{l}{\eqref{eq6-het}} \\ \midrule

\multicolumn{7}{c}{$\displaystyle\frac { \Phi\,( \deg(u) - 1 + E_{\ovr{u}} \,) - \gamma \,( \deg(u) + E_{\ovr{u}} ) } { \deg(u) }
\,+\,
\frac { \Phi\,( \deg(v) - 1 + E_{\ovr{v}} \,) - \gamma \,( \deg(v) + E_{\ovr{v}} ) } { \deg(v) }
\,-\,
\gamma \, \left( 1 + E_{f_{\ovr{u},\ovr{v}}} \right)
>
\gamma \, E_{\ovr{h_{i,j}}}$} & \eqref{eq7-het}  \\ \midrule
$\gamma \, E_{\ovr{h_{i,j}}} < 1$ & \eqref{eq8-het}  & $\displaystyle\Phi \leq \gamma\, \left( 1 +\frac {1} {E_{\fgerightB_j}} \,\right)$ & \eqref{eq9-het} 
           & \multicolumn{2}{c}{$\displaystyle\Phi>\frac {\gamma \, \left( 2 + \frac{E_{\vsi}}{|F|} + E_{\fgerightB_j} \, \right) } { \left(1 + \frac{E_{\vsi}}{|F|} \right) }$} & 
           \multicolumn{1}{l}{\eqref{eq10-het}} &  \\ \midrule
\multicolumn{3}{r}{ $\displaystyle\Phi>\gamma \, 
                     \left( 
                     \frac
                     { ( |F| +E_{\vsi}\,) }
                     { 3\,|F| + |F| E_{\fgerightB_j} + |F| E_{\fgeleftB_j} + E_{\vsi} }
                     \right)$}  & \eqref{eq13-het} & 
{$\displaystyle\Phi \leq \gamma + \frac { 1 } { 1 + \frac{E_{\vsi} } {|F|} } $ } & \eqref{eq11-het} & 
          $\displaystyle\Phi \leq \gamma\, \left( 1 +\frac {1} {E_{\fgeleftB_j}} \,\right)$ & \eqref{eq12-het} \\ \midrule
\multicolumn{3}{r}{$\displaystyle\Phi \leq \gamma \, \left( 1 + \frac { 1 + E_{\fgerightB_j} } {1 + \frac{E_{\vsi}}{|F|} }  \, \right) + \frac{1}{1+\frac{E_{\vsi}}{|F|} }$}
          & \eqref{eq14-het}  & \multicolumn{3}{c}{
            $\displaystyle\Phi\leq \gamma \left( \frac {3 + \frac {E_{\vsi}} {|F|} + E_{\fgerightB_j} + E_{\fgeleftB_j}} { 1 + \frac {E_{\vsi}} {|F|} } \right) + \frac{1}{ 1 + \frac {E_{\vsi}} {|F|} }$
} 
& 
\eqref{eq16-het}
\\ \midrule
\multicolumn{7}{r}{$\displaystyle\Phi> \gamma \left(
                     \frac 
                     {
                     6 + \frac{1}{\deg(u)}  + \frac { E_{\vt} } {n\,\deg(u)} + \frac {E_{\ovr{u}} }{\deg(u)} + \frac{1}{\deg(v)}  + \frac { E_{\vt} } {n\,\deg(v)} + \frac {E_{\ovr{v}} }{\deg(v)} + E_{f_{\ovr{u},\ovr{v}}} + \frac {E_{\vsi}} {|F|} + E_{\fgerightB_j} + E_{\fgeleftB_j}
                     }
                     { 
                     1 + \frac {E_{\vsi}} {|F|} + 
                     \frac{1}{\deg(u)} + \frac {E_{\vt} } {n\,\deg(u)} + \frac{1}{\deg(v)} + \frac {E_{\vt} } {n\,\deg(v)} } 
                     \right)$
                     } & \eqref{eq15-het} \\ \midrule
\multicolumn{3}{r}{$\displaystyle\Phi \leq \gamma \left( \frac{1  + \frac { E_{\vt} } {n} + \deg(u) + E_{\ovr{u}}} {1 + \frac {E_{\vt} } {n} } \right) + \frac {\deg(u)} { 1 + \frac {E_{\vt} } {n}}$} & 
                     \eqref{eq17-het} 
                     & \multicolumn{3}{r}{
                                       $\displaystyle\Phi 
                                       > 
                                       \gamma 
                                       \left( 
                                       \frac 
                                       {6 + \frac {E_{\vsi}} {|F|} + E_{\fgerightB_j} + E_{\fgeleftB_j} + \frac { E_{\vt} } {n\,\deg(u)} + \frac {E_{\ovr{u}} + 1 } { \deg(u) } + \frac {E_{\ovr{v}} } {\deg(v)} + E_{f_{\ovr{u},\ovr{v}}} 
                                       }
                                       {
                                       2 + \frac {E_{\vsi}} {|F|} + \frac{1}{\deg(u)}  + \frac {E_{\vt} } {n\,\deg(u)} + \frac {E_{\ovr{v}} -1 } {\deg(v)} 
                                       }
                                       \right) 
                                       $} & \eqref{eq18-het} \\ \midrule
\multicolumn{5}{r}{$\displaystyle\Phi 
\leq
\gamma \left(
\frac {3 + \frac {E_{\vsi}} {|F|} + E_{\fgerightB_j} + E_{\fgeleftB_j} } { 1 + \frac {E_{\vsi}} {|F|} - E_{\fgerightB_j} + E_{\fgeleftB_j} }
\right)
+ 
\frac {1} { 1 + \frac {E_{\vsi}} {|F|} - E_{\fgerightB_j} + E_{\fgeleftB_j} }$} & \eqref{eq19-het} \\ \midrule
\multicolumn{8}{l}{$\displaystyle\Phi 
\leq
\gamma \left( 
\frac
{
6 + \frac{1}{\deg(u)}  + \frac {E_{\vt}}{n\,\deg(u)} + \frac{E_{\ovr{u}}}{\deg(u)} + \frac{1}{\deg(v)}  + \frac {E_{\vt}}{n\,\deg(v)} + \frac{E_{\ovr{v}}}{\deg(v)} + E_{f_{\ovr{u},\ovr{v}}} + \frac {E_{\vsi}} {|F|} + E_{\fgerightB_j} + E_{\fgeleftB_j}  
}
{
1 + \frac{1}{\deg(u)} + \frac {E_{\vt} }{n\,\deg(u)} + \frac{1}{\deg(v)} + \frac {E_{\vt} }{n\,\deg(v)} + \frac {E_{\vsi}} {|F|} - E_{\fgerightB_j} + E_{\fgeleftB_j}
}
\right)$} \\ 
\multicolumn{7}{r}{$\displaystyle +
\frac { 1 }
{
1 + \frac{1}{\deg(u)} + \frac {E_{\vt} }{n\,\deg(u)} + \frac{1}{\deg(v)} + \frac {E_{\vt} }{n\,\deg(v)} + \frac {E_{\vsi}} {|F|} - E_{\fgerightB_j} + E_{\fgeleftB_j}
}$} & \eqref{eq20-het} \\ \midrule
\multicolumn{7}{c}{$\displaystyle\Phi 
\leq 
\gamma 
\left( 
\frac 
{
6 + \frac{1}{\deg(u)}  + \frac {E_{\vt}}{n\,\deg(u)} + \frac{E_{\ovr{u}}}{\deg(u)} + \frac{1}{\deg(v)}  + \frac {E_{\vt}}{n\,\deg(v)} + \frac{E_{\ovr{v}}}{\deg(v)} + E_{f_{\ovr{u},\ovr{v}}} + \frac {E_{\vsi}} {|F|} + E_{\fgerightB_j} + E_{\fgeleftB_j} + \frac{E_{\ovr{h_{i,j}}} }{ |F_{i,j}|} 
}
{
1 + \frac{1}{\deg(u)} + \frac {E_{\vt} }{n\,\deg(u)} + \frac{1}{\deg(v)} + \frac {E_{\vt} }{n\,\deg(v)} + \frac {E_{\vsi}} {|F|} - E_{\fgerightB_j} + E_{\fgeleftB_j}
}
\right)$} & \eqref{eq21-het} \\ \midrule
\multicolumn{3}{r}{$\displaystyle\Phi 
\leq \gamma \left( 
\frac {2 + \frac {E_{\vsi}} {|F|} + E_{\fgerightB_j} }
{1 + \frac {E_{\vsi}} {|F|} - E_{\fgerightB_j}}
\right) + 
\frac {1} 
{1 + \frac {E_{\vsi}} {|F|} - E_{\fgerightB_j}}$} & \eqref{eq22-het} & 
\multicolumn{3}{r}{$\displaystyle\Phi 
\leq 
\gamma \left( 
\frac { 2 + \frac {E_{\vsi}} {|F|} + E_{\fgerightB_j} + 1 + E_{\fgeleftB_j}  } 
{ 1 + \frac {E_{\vsi}} {|F|} - E_{\fgerightB_j} +  E_{\fgeleftB_j} }
\right)  
+
\frac { 1 } 
{ 1 + \frac {E_{\vsi}} {|F|} - E_{\fgerightB_j} +  E_{\fgeleftB_j} }$} & \eqref{eq23-het} \\ \midrule
\\
\multicolumn{7}{c}{$\displaystyle\Phi 
\leq 
\gamma \left( 
\frac
{
6 + \frac{1}{\deg(u)}  + \frac {E_{\vt}}{n\,\deg(u)} + \frac {E_{\ovr{u}}}{\deg(u)} + \frac {E_{\ovr{v}}}{\deg(v)}
+
\frac {E_{\vsi}} {|F|} + E_{\fgerightB_j} + E_{f_{\ovr{u},\ovr{v}}} + E_{\fgeleftB_j} 
+
\frac { \gamma \, E_{\ovr{h_{i,j}}} }{ |F_{i,j}|} 
}
{ 2 + \frac{1}{\deg(u)} + \frac {E_{\vt}}{n\,\deg(u)} - \frac {1}{\deg(v)} + \frac {E_{\ovr{v}}}{\deg(v)} + \frac {E_{\vsi}} {|F|} - E_{\fgerightB_j} +  E_{\fgeleftB_j} } 
\right)$} & \eqref{eq24-het} \\ \midrule
\multicolumn{7}{r}{
\hspace*{-0.2in}
$\displaystyle\Phi 
\leq
\gamma \left( 
\frac 
{
6 + \frac{1}{\deg(u)}  + \frac {E_{\vt}}{n\,\deg(u)} + \frac {E_{\ovr{u}}}{\deg(u)} + \frac {E_{\ovr{v}}}{\deg(v)}
+
\frac {E_{\vsi}} {|F|} + E_{\fgerightB_j} + E_{f_{\ovr{u},\ovr{v}}} + E_{\fgeleftB_j}  
}
{ 2 + \frac{1}{\deg(u)} + \frac {E_{\vt}}{n\,\deg(u)} - \frac {1}{\deg(v)} + \frac {E_{\ovr{v}}}{\deg(v)} + \frac {E_{\vsi}} {|F|} - E_{\fgerightB_j} +  E_{\fgeleftB_j} }
\right)
\!+\!
\frac { 1 }
{ 2 + \frac{1}{\deg(u)} + \frac {E_{\vt}}{n\,\deg(u)} - \frac {1}{\deg(v)} + \frac {E_{\ovr{v}}}{\deg(v)} + \frac {E_{\vsi}} {|F|} - E_{\fgerightB_j} +  E_{\fgeleftB_j} }$
} & \eqref{eq25-het} \\
\bottomrule
\end{tabular}
\caption{\label{summ-cons}List of all inequalities to be satisfied in the proof of Theorem~\ref{hetero-thm2}.}
\end{table}
}

\noindent
Let the share of external assets for a node (bank) $y$ be denoted by $E_y$ (thus, $\sum_{y\in V}E_y=E$).
We will select the remaining network parameters, namely $\gamma$, $\Phi$ and the set of $E_y$ values, based on the following desirable properties and events.
For the convenience of the readers, all the relevant constraints are also summarized in Table~\ref{summ-cons}.
Assume that no nodes in $\left(\cup_{i,j}\,\ovr{F_{i,j}}\right)\bigcup\left(\cup_{i,j}\left\{\ovr{h_{i,j}}\right\}\right)$ were shocked at $t=1$.
\begin{description}
\item[(I)]
Suppose that the node $\vt$ is shocked at $t=1$. Then, the following happens.
\begin{description}
\item[(I-a)]
$\vt$ fails at $t=1$:
\begin{gather}
\hspace*{-0.5in}
\Phi\,( b_{\vt}-\iota_{\vt}+E_{\vt}\,) > \gamma\,( b_{\vt}+E_{\vt}\,)
\,\,\,\equiv\,\,\,
\Phi\,(n+E_{\vt}\,) > \gamma\,(n+E_{\vt}\,) 
\,\,\,\equiv\,\,\,
\boxed { \Phi > \gamma }
\label{eq1-het} 
\end{gather}

\item[(I-b)]
Each node $\ovr{u}\in\ovr{\vl}\cup\ovr{\vr}$ that was not shocked at $t=1$ fails at $t=2$: 
\begin{multline*}
\frac { \min \left\{ \Phi\,( b_{\vt}-\iota_{\vt}+E_{\vt}\,) - \gamma\,( b_{\vt}+E_{\vt}\,),\,b_{\vt} \right\} }  { \din(\vt) }
> 
\gamma\, \left( b_{\ovr{u}} + E_{\ovr{u}} \right) 
\\
\,\,\,\equiv\,\,\,
\frac { \min \left\{ \Phi\,( n +E_{\vt}\,) - \gamma\,( n + E_{\vt}\,),\,n \right\} }  { n }
> 
\gamma\, \left( \deg(u) + E_{\ovr{u}} \right) 
\end{multline*}
These constraints are satisfied provided:
\begin{gather}
\displaystyle
\hspace*{-0.5in}
\frac { \Phi\,( n +E_{\vt}\,) - \gamma\,( n + E_{\vt}\,) }  { n } > \gamma\, \left( \deg(u) + E_{\ovr{u}} \right) 
\,\,\,\pmb{\equiv}\,\,\,
\boxed {
\Phi > \gamma\, \left( 1 + \frac { \deg(u) + E_{\ovr{u}} } { 1 + \frac{E_{\vt}}{n} } \right) 
}
\label{eq2-het} 
\\
\displaystyle
\Phi\,( n +E_{\vt}\,) - \gamma\,( n + E_{\vt}\,) \leq n  
\,\,\,\pmb{\equiv}\,\,\,
\boxed {
\Phi \leq \gamma + \frac { 1 } { 1 + \frac { E_{\vt} } { n } }
}
\label{eq3-het} 
\end{gather}

\item[(I-c)]
If the nodes $\ovr{u}$, $\ovr{v}$ and $f_{\ovr{u},\ovr{v}}$ were {\em not} shocked at $t=1$, then 
the part of the shock, say $\sigma_1$, given to $\vt$ that is received by node $f_{\ovr{u},\ovr{v}}$ at $t=3$ is: 
\[
\hspace*{-0.9in}
\begin{array}{lll}
\sigma_1  & = & 
\dfrac { \displaystyle \min \left\{ \frac { \min \left\{ \Phi\,( b_{\vt}-\iota_{\vt}+E_{\vt}\,) - \gamma\,( b_{\vt}+E_{\vt}\,),\,b_{\vt} \right\} }  { \din(\vt) } - \gamma\, \left( b_{\ovr{u}} + E_{\ovr{u}} \right),\, b_{\ovr{u}} \right\} } { \din\left(\ovr{u}\right) }
\\
& & \,+\,
\dfrac { \displaystyle \min \left\{ \frac { \min \left\{ \Phi\,( b_{\vt}-\iota_{\vt}+E_{\vt}\,) - \gamma\,( b_{\vt}+E_{\vt}\,),\,b_{\vt} \right\} }  { \din(\vt) } - \gamma\, \left( b_{\ovr{v}} + E_{\ovr{v}} \right),\, b_{\ovr{v}} \right\} } { \din\left(\ovr{v}\right) } 
\\
& = & 
\dfrac { \displaystyle \min \left\{ \frac { \min \left\{ \Phi\,( n +E_{\vt}\,) - \gamma\,( n  +E_{\vt}\,),\,n \right\} }  { n } - \gamma\, \left( \deg(u) + E_{\ovr{u}} \right),\, \deg(u) \right\} } { \deg(u) }
\\
& & \,+\,
\dfrac { \displaystyle \min \left\{ \frac { \min \left\{ \Phi\,( n +E_{\vt}\,) - \gamma\,( n  +E_{\vt}\,),\,n \right\} }  { n } - \gamma\, \left( \deg(v) + E_{\ovr{v}} \right),\, \deg(v) \right\} } { \deg(v) }
\end{array}
\]
On the other hand, if the node $f_{\ovr{u},\ovr{v}}$ and {\em exactly} one of the nodes  $\ovr{u}$ and $\ovr{v}$, say $\ovr{u}$, were {\em not} shocked at $t=1$, then 
the part of the shock, say $\sigma_1'$, given to $\vt$ that is received by node $f_{\ovr{u},\ovr{v}}$ at $t=3$ is: 
\begin{gather*}
\sigma_1' 
=
\dfrac { \displaystyle \min \left\{ \frac { \min \left\{ \Phi\,( n +E_{\vt}\,) - \gamma\,( n  +E_{\vt}\,),\,n \right\} }  { n } - \gamma\, \left( \deg(u) + E_{\ovr{u}} \right),\, \deg(u) \right\} } { \deg(u) }
\end{gather*}
\end{description}
\item[(II)]
Suppose that some node $\ovr{u}$ is shocked at $t=1$. Then, the following happens.
\begin{description}
\item[(II-a)]
Node $\ovr{u}$ fails at $t=1$:
\begin{gather}
\Phi\,( b_{\ovr{u}}-\iota_{\ovr{u}}+E_{\ovr{u}}\,) > \gamma\,( b_{\ovr{u}}+E_{\ovr{u}}\,)
\,\,\,\equiv\,\,\,
\boxed {
\Phi > \gamma\, \left(\, 1 + \frac {1} {\deg(u) - 1 + E_{\ovr{u}} } \right)
}
\label{eq4-het} 
\end{gather}

\item[(II-b)]
Node $f_{\ovr{u},\ovr{v}}\in\ovr{F_{i,j}}$ fails at $t=2$ and node $\ovr{h_{i,j}}$ fails at $t=3$ if both $\ovr{u}$ and $\ovr{v}$ were shocked at $t=1$: 
\begin{multline*}
\hspace*{-1in}
\frac { \min \left\{ \, \Phi\,( b_{\ovr{u}}-\iota_{\ovr{u}}+E_{\ovr{u}}\,) - \gamma\,( b_{\ovr{u}}+E_{\ovr{u}}\,), \, b_{\ovr{u}} \, \right\} } { \din(\ovr{u}) }
\\
\, + \,
\frac { \min \left\{ \, \Phi\,( b_{\ovr{v}}-\iota_{\ovr{v}}+E_{\ovr{v}}\,) - \gamma\,( b_{\ovr{v}}+E_{\ovr{v}}\,), \, b_{\ovr{v}} \, \right\} } { \din(\ovr{v}) }
>
\gamma \, \left( b_{f_{\ovr{u},\ovr{v}}} + E_{f_{\ovr{u},\ovr{v}}} \right)
\end{multline*}
\begin{multline*}
\hspace*{-1in}
\,\,\pmb{\equiv} \,\,
\frac { \min \left\{ \, \Phi\, \left( \deg(u) - 1 + E_{\ovr{u}} \right) \,-\, \gamma\,( \deg(u) + E_{\ovr{u}}\,), \, \deg(u) \, \right\} } { \deg(u) }
\\
\, + \,
\frac { \min \left\{ \, \Phi\, \left( \deg(v) - 1 + E_{\ovr{v}} \right) \,-\, \gamma\,( \deg(v) + E_{\ovr{v}}\,), \, \deg(v) \, \right\} } { \deg(v) }
>
\gamma \, \left( 1 + E_{f_{\ovr{u},\ovr{v}}} \right)
\end{multline*}
\begin{multline*}
\hspace*{-1.5in}
\scalebox{0.85}{
$\displaystyle 
\dfrac
{
\min
\left\{
\frac { \displaystyle \min \left\{ \, \Phi\,( b_{\ovr{u}}-\iota_{\ovr{u}}+E_{\ovr{u}}\,) - \gamma\,( b_{\ovr{u}}+E_{\ovr{u}}\,), \, b_{\ovr{u}} \, \right\} } { \din(\ovr{u}) }
\,+\,
\frac { \displaystyle \min \left\{ \, \Phi\,( b_{\ovr{v}}-\iota_{\ovr{v}}+E_{\ovr{v}}\,) - \gamma\,( b_{\ovr{v}}+E_{\ovr{v}}\,), \, b_{\ovr{v}} \, \right\} } { \din(\ovr{v}) }
\,-\,
\gamma \, \left( b_{f_{\ovr{u},\ovr{v}}} + E_{f_{\ovr{u},\ovr{v}}} \right) , \,
b_{f_{\ovr{u},\ovr{v}}}
\right\}
}
{ \din\left( f_{\ovr{u},\ovr{v}} \right) }
$
}
\\
>\,\gamma \, \left( b_{\ovr{h_{i,j}}} + E_{\ovr{h_{i,j}}} \right) 
\end{multline*}
\vspace*{-0.3in}
\begin{gather*}
{\mbox{\Large$\pmb{\equiv}$}}
\end{gather*}
\begin{multline*}
\hspace*{-1.5in}
\scalebox{0.8}
{
$\displaystyle
\min
\left\{
\dfrac { \displaystyle \min \left\{ \, \Phi\,( \deg(u) - 1 + E_{\ovr{u}} \,) - \gamma \,( \deg(u) + E_{\ovr{u}} \,), \, \deg(u) \, \right\} } { \deg(u) }
+
\dfrac { \displaystyle \min \left\{ \, \Phi\,( \deg(v) - 1 + E_{\ovr{v}} \,) - \gamma \,( \deg(v) + E_{\ovr{v}} \,), \, \deg(v) \, \right\} } { \deg(v) }
\,-\,
\gamma \, \left( 1 + E_{f_{\ovr{u},\ovr{v}}} \right) , \, 1
\right\}
$
}
\\
>\, \gamma \,\, E_{\ovr{h_{i,j}}}
\end{multline*}
These constraints are satisfied provided the inequalities~\eqref{eq1-het}--\eqref{eq4-het} are satisfied, and the following holds:
\begin{multline}
\hspace*{-1.3in}
\boxed {
\frac { \Phi\, \left( \deg(u) - 1 + E_{\ovr{u}} \right)  - \gamma\,( \deg(u) + E_{\ovr{u}}\,) } { \deg(u) }
+
\frac { \Phi\, \left( \deg(v) - 1 + E_{\ovr{v}} \right)  - \gamma\,( \deg(v) + E_{\ovr{v}}\,) } { \deg(v) }
>
\gamma \, \left( 1 + E_{f_{\ovr{u},\ovr{v}}} \right)
}
\label{eq5-het} 
\end{multline}
\begin{gather}
\hspace*{-1.3in}
\Phi\, \left( \deg(u) - 1 + E_{\ovr{u}} \right) \,-\, \gamma \, \left( \deg(u) + E_{\ovr{u}} \, \right) \leq \deg(u) 
\,\,\pmb{\equiv}\,\,
\boxed {
\Phi \leq \gamma \, \left( \frac { \deg(u) + E_{\ovr{u}} } { \deg(u) - 1 + E_{\ovr{u}} }   \right) + \frac { \deg(u) } { \deg(u) - 1 + E_{\ovr{u}} }
}
\label{eq6-het} 
\end{gather}
\begin{multline}
\hspace*{-1.5in}
\boxed {
\frac { \Phi\,( \deg(u) - 1 + E_{\ovr{u}} \,) - \gamma \,( \deg(u) + E_{\ovr{u}} ) } { \deg(u) }
\,+\,
\frac { \Phi\,( \deg(v) - 1 + E_{\ovr{v}} \,) - \gamma \,( \deg(v) + E_{\ovr{v}} ) } { \deg(v) }
\,-\,
\gamma \, \left( 1 + E_{f_{\ovr{u},\ovr{v}}} \right)
>
\gamma \, E_{\ovr{h_{i,j}}}
}
\label{eq7-het} 
\end{multline}
\begin{equation}
\boxed {
\gamma \, E_{\ovr{h_{i,j}}} < 1 
}
\label{eq8-het} 
\end{equation}
\end{description}

\item[(III)]
When the node $\vsi$ is shocked at $t=1$, the following happens.
\begin{description}
\item[(III-a)]
$\vsi$ fails at $t=1$:
\begin{gather*}
\Phi\,( b_{\vsi}-\iota_{\vsi}+E_{\vsi}\,) > \gamma\,( b_{\vsi}+E_{\vsi}\,)
\,\,\,\equiv\,\,\,
\Phi\,(|F|+E_{\vsi}\,) > \gamma\,(|F|+E_{\vsi}\,) 
\,\,\,\equiv\,\,\,
\Phi > \gamma
\end{gather*}
which is same as~\eqref{eq1-het}. 

\item[(III-b)]
If a node $\fgerightB_j\in\fgerightB$ is shocked at $t=1$, it does not fail: 
\begin{gather}
\Phi\,( b_{\fgerightB_j}-\iota_{\fgerightB_j}+E_{\fgerightB_j}\,) \leq \gamma\,( b_{\fgerightB_j}+E_{\fgerightB_j}\,)
\,\,\,\equiv\,\,\,
\boxed{
\Phi \leq \gamma\, \left( 1 +\frac {1} {E_{\fgerightB_j}} \,\right)
}
\label{eq9-het} 
\end{gather}

\item[(III-c)]
Any node $\fgerightB_j\in\fgerightB$ fails at $t=2$ irrespective of whether $\fgerightB_j$ was shocked or not: 
\begin{gather*}
\frac { \min \left\{ \Phi\,( b_{\vsi}-\iota_{\vsi}+E_{\vsi}\,) \,-\, \gamma\,( b_{\vsi}+E_{\vsi}\,), \, b_{\vsi} \right\}  } { \din(\vsi) }
>
\gamma \,( b_{\fgerightB_j}+E_{\fgerightB_j} \, )
\end{gather*}
These constraints are satisfied provided: 
\begin{gather}
\hspace*{-0.5in}
\frac { \Phi\,( b_{\vsi}-\iota_{\vsi}+E_{\vsi}\,) \,-\, \gamma\,( b_{\vsi}+E_{\vsi}\,) } { \din(\vsi) }
>
\gamma \,( b_{\fgerightB_j}+E_{\fgerightB_j} \, )
\,\,\,\equiv\,\,\,
\boxed{
\Phi >
\frac {\gamma \, \left( 2 + \frac{E_{\vsi}}{|F|} + E_{\fgerightB_j} \, \right) } { \left(1 + \frac{E_{\vsi}}{|F|} \right) }
}
\label{eq10-het} 
\end{gather}
\vspace*{-0.3in}
\begin{gather}
\Phi\,( b_{\vsi}-\iota_{\vsi}+E_{\vsi}\,) \,-\, \gamma\,( b_{\vsi}+E_{\vsi}\,) \leq b_{\vsi}
\,\,\,\equiv\,\,\,
\boxed{
\Phi \leq \gamma + \frac { 1 } { 1 +  \frac {E_{\vsi}} {|F|} }
}
\label{eq11-het} 
\end{gather}

\item[(III-d)]
If a node $\fgeleftB_j\in\fgeleftB$ is shocked at $t=1$, it does not fail (and thus, by {\bf (III-b)}, it does not fail at $t=2$ also): 
\vspace*{-0.2in}
\begin{gather}
\Phi\,( b_{\fgeleftB_j}-\iota_{\fgeleftB_j}+E_{\fgeleftB_j}\,) \leq \gamma\,( b_{\fgeleftB_j}+E_{\fgeleftB_j}\,)
\,\,\,\equiv\,\,\,
\boxed{
\Phi \leq \gamma\, \left( 1 +\frac {1} {E_{\fgeleftB_j}} \,\right)
}
\label{eq12-het} 
\end{gather}

\item[(III-e)]
Any node $\fgeleftB_j\in\fgeleftB$ fails at $t=3$ irrespective of whether $\fgeleftB_j$ was shocked or not: 
\begin{gather*}
\frac
{
\min \left\{
\frac { \min \left\{ \Phi\,( b_{\vsi}-\iota_{\vsi}+E_{\vsi}\,) \,-\, \gamma\,( b_{\vsi}+E_{\vsi}\,), \, b_{\vsi} \right\}  } { \din(\vsi) }
\,-\,
\gamma \,( b_{\fgerightB_j}+E_{\fgerightB_j} \, ), \, b_{\fgerightB_j}
\right\}
}
{ \din(\fgerightB_j) }
>
\gamma \, \left( b_{\fgeleftB_j} + E_{\fgeleftB_j} \right)
\,\,\,\equiv\,\,\,
\\
\min \left\{
\frac { \min \left\{ \Phi\,( |F| +E_{\vsi}\,) \,-\, \gamma\,( |F| +E_{\vsi}\,), \, |F| \right\}  } { |F| }
\,-\,
\gamma \,( 1 +E_{\fgerightB_j} \, ), \, 1
\right\}
>
\gamma \, \left( 1 + E_{\fgeleftB_j} \right)
\end{gather*}
These constraints are satisfied provided all the previous constraints hold and the following holds: 
\begin{gather}
\hspace*{-1.4in}
\frac { \Phi\,( |F| +E_{\vsi}\,) \,-\, \gamma\,( |F| +E_{\vsi}\,) } { |F| }
\,-\,
\gamma \,( 1 +E_{\fgerightB_j} \, )
>
\gamma \, \left( 1 + E_{\fgeleftB_j} \right)
\,\,\,\equiv\,\,\,
\boxed{
\Phi
>
\gamma \, 
\left( 
\frac
{ ( |F| +E_{\vsi}\,) }
{ 3\,|F| + |F| E_{\fgerightB_j} + |F| E_{\fgeleftB_j} + E_{\vsi} }
\right)
}
\label{eq13-het} 
\end{gather}
\vspace*{-0.3in}
\begin{gather}
\hspace*{-0.5in}
\frac { \Phi\,( |F| +E_{\vsi}\,) \,-\, \gamma\,( |F| +E_{\vsi}\,) } { |F| }
\,-\,
\gamma \,( 1 +E_{\fgerightB_j} \, )
\leq 1
\,\,\,\equiv\,\,\,
\boxed{
\Phi
\leq
\gamma \, \left( 1 + \frac { 1 + E_{\fgerightB_j} } {1 + \frac{E_{\vsi}}{|F|} }  \, \right)
+ \frac{1}{1+\frac{E_{\vsi}}{|F|} }
}
\label{eq14-het} 
\end{gather}

\item[(III-f)]
Consider a directed path 
\pscirclebox[boxsep=false,framesep=0pt]{$\scriptstyle\,\vsi\,$}$\,\longleftarrow$\pscirclebox[boxsep=false,framesep=0pt]{$\fgerightB_j$}$\,\longleftarrow$\pscirclebox[boxsep=false,framesep=0pt]{$\fgeleftB_j$}$\,\longleftarrow$\pscirclebox[boxsep=false,framesep=0pt]{$p_j$}
from $p_j=f_{\ovr{u},\ovr{v}}$ to $\vsi$. 
The maximum value of its proportion of shock receive by $p_j$ from this path, say $\sigma_2$, is obtained by shocking all the nodes $\vsi,\fgerightB_j,\fgeleftB_j$ and is given by (assuming
all previous inequalities hold): 
\begin{gather*}
\hspace*{-1.5in}
\scalebox{0.75}{
$ \displaystyle 
\sigma_2 = 
\dfrac{ \displaystyle
\min \left\{
\frac { \displaystyle
\min \left\{
\frac {\Phi \, \left( b_{\vsi}-\iota_{\vsi} + E_{\vsi} \right) - \gamma \, \left( b_{\vsi} + E _{\vsi} \right) } { \din(\vsi) }
\,-\,
\left(   \gamma \left( b_{\fgerightB_j} + E_{\fgerightB_j}  \right)  -  \Phi \left( b_{\fgerightB_j} - \iota_{\fgerightB_j} + E_{\fgerightB_j} \right) \right), \, 
b_{\fgerightB_j}
\right\}
}
{ \displaystyle \din(\fgerightB_j) }
\,-\,
\left(   \gamma \left( b_{\fgeleftB_j} + E_{\fgeleftB_j}  \right)  -  \Phi \left( b_{\fgeleftB_j} - \iota_{\fgeleftB_j} + E_{\fgeleftB_j} \right) \right), \, 
b_{\fgeleftB_j}
\right\}
}
{ \displaystyle \din(\fgeleftB_j) }
$
}
\\
\hspace*{-0.5in}
\,=\,\min \left\{
\min \left\{
\Phi \left( 1 + \frac {E_{\vsi}} {|F|} - E_{\fgerightB_j} \right) 
\,-\,
\gamma \left( 2 + \frac {E_{\vsi}} {|F|} + E_{\fgerightB_j} \right), \,
1
\right\}
\,-\,
\left(   \gamma \left( 1 + E_{\fgeleftB_j}  \right)  -  \Phi E_{\fgeleftB_j} \right), \, 
1
\right\}
\end{gather*}
Similarly, the minimum value of its proportion of shock receive by $p_j$ from this path, say $\sigma_2$, is obtained by shocking only the node $\vsi$ and is given by (assuming
all previous inequalities hold): 
\begin{gather*}
\hspace*{-1in}
\sigma_2' = 
\dfrac{
\min \left\{
\frac { \displaystyle
\frac {\displaystyle  \Phi \, \left( b_{\vsi}-\iota_{\vsi} + E_{\vsi} \right) - \gamma \, \left( b_{\vsi} + E _{\vsi} \right) } { \din(\vsi) }
\,-\,
\gamma \left( b_{\fgerightB_j} + E_{\fgerightB_j}  \right)
}
{ \displaystyle \din(\fgerightB_j) }
\,-\,
\gamma \left( b_{\fgeleftB_j} + E_{\fgeleftB_j}  \right), \,
b_{\fgeleftB_j}
\right\}
}
{ \displaystyle  \din(\fgeleftB_j) }
\\
\,=\,\min \left\{
\Phi \left( 1 + \dfrac {E_{\vsi}} {|F|}  \right) 
\,-\,
\gamma \left( 2 + \dfrac {E_{\vsi}} {|F|} + E_{\fgerightB_j} \right)
\,-\,
\gamma \left( 1 + E_{\fgeleftB_j} \right),
1
\right\}
\end{gather*}
We want node $f_{\ovr{u},\ovr{v}}$ to fail at $t=4$ assuming it did not fail already. Since $f_{\ovr{u},\ovr{v}}$  did not fail at $t=2$, 
at most one of the nodes $\ovr{u}$ and $\ovr{v}$ was shocked. There are two cases to consider: when neither $\ovr{u}$ nor $\ovr{v}$ was shocked, 
or when exactly one of these nodes, say $\ovr{v}$, was shocked (assuming all previous inequalities hold): 
\hspace*{-1in}
\begin{longtable}{rl}
$\displaystyle  \sigma_2' + \sigma_1 =$ & 
         $\displaystyle  
          \min \left\{
                         \Phi \left( 1 + \dfrac {E_{\vsi}} {|F|}  \right) 
                      \,-\,
                         \gamma \left( 2 + \dfrac {E_{\vsi}} {|F|} + E_{\fgerightB_j} \right)
                      \,-\,
                         \gamma \left( 1 + E_{\fgeleftB_j} \right),\,
                         1
                \right\}
         $
\\
&
         $\displaystyle  
              \,+\, 
                 \dfrac { \min 
                             \left\{ \frac 
                                          { \displaystyle  \min \left\{ \Phi\,( n +E_{\vt} \right) - \gamma \left( n  +E_{\vt}\,),\,n \right\} }  
                                          { \displaystyle  n } 
                                     \,-\, 
                                     \gamma\, \left( \deg(u) + E_{\ovr{u}} \right),\, \deg(u) 
                             \right\} 
                        } 
                        { \deg(u) }
          $
\\
& 
          $\displaystyle  
               \,+\,
                   \dfrac { \min \left\{ 
                                     \frac 
                                           { 
                                             \displaystyle  \min \left\{ \Phi\,( n +E_{\vt} \right) - \gamma \left( n  +E_{\vt}\,),\,n \right\} 
                                           }  
                                           { \displaystyle  n } 
                                      \,-\, \gamma\, \left( \deg(v) + E_{\ovr{v}} \right),\, \deg(v) 
                                 \right\} 
                          }
                          { \deg(v) }
$
\\
$\displaystyle >$ & 
          $\displaystyle \gamma \, \left(  b_{f_{\ovr{u},\ovr{v}}} + E_{f_{\ovr{u},\ovr{v}}} \right)$
\\
\\
& \hspace*{2in} \mbox{$\displaystyle \pmb{\equiv}$} 
\\
\\
& \hspace*{0.15in} 
          $\displaystyle  
                 \min \left\{
                          \Phi \left( 1 + \dfrac {E_{\vsi}} {|F|}  \right) 
                          \,-\,
                          \gamma \left( 2 + \dfrac {E_{\vsi}} {|F|} + E_{\fgerightB_j} \right)
                           \,-\,
                          \gamma \left( 1 + E_{\fgeleftB_j} \right),\,
                          1
                       \right\}
          $
\\
&
          $\displaystyle  
                         +\, 
                            \dfrac { 
                                    \min \left\{     \Phi \left( 1 + \frac {\displaystyle  E_{\vt} } {\displaystyle  n} \right) 
                                               \,-\, \gamma \left( 1  + \frac {\displaystyle   E_{\vt} } {\displaystyle  n} \right)  
                                               \,-\, \gamma\, \left( \deg(u) + E_{\ovr{u}} \right),\, 
                                               \deg(u) 
                                         \right\} 
                                   }
                                   { \deg(u) }
          $
\\
& 
          $\displaystyle  
                        +\,
                            \dfrac { 
                                     \min \left\{     \Phi \left( 1 + \frac {\displaystyle  E_{\vt} } {\displaystyle  n} \right) 
                                                \,-\, \gamma \left( 1  + \frac {\displaystyle   E_{\vt} } {\displaystyle  n} \right)  
                                                \,-\, \gamma\, \left( \deg(v) + E_{\ovr{v}} \right),\, 
                                                \deg(v) 
                                          \right\} 
                                   }
                                   { \deg(v) }
          $
\\
$\displaystyle >$ & 
          $\displaystyle \gamma \, \left(  1 + E_{f_{\ovr{u},\ovr{v}}} \right)$
\end{longtable}
%
\begin{gather*}
\displaystyle
\hspace*{-2in}
\sigma_2' + \sigma_1' + 
        \dfrac { 
                 \min \left\{ \, 
                              \Phi\,( b_{\ovr{v}}-\iota_{\ovr{v}}+E_{\ovr{v}}\,) 
                              \,-\, 
                              \gamma\,( b_{\ovr{v}}+E_{\ovr{v}}\,), \, b_{\ovr{v}} \, 
                      \right\} 
               }
               { \din(\ovr{v}) }
> 
\gamma \, \left(  b_{f_{\ovr{u},\ovr{v}}} + E_{f_{\ovr{u},\ovr{v}}} \right) 
\end{gather*}
\vspace*{-0.2in}
\begin{gather*}
\displaystyle
\hspace*{-2in}
\pmb{\equiv}
\end{gather*}
\vspace*{-0.3in}
\begin{longtable}{rl}
 & $\displaystyle 
         \min \left\{ 
            \Phi \left( 1 + \dfrac {E_{\vsi}} {|F|}  \right) 
             \,-\, 
            \gamma \left( 2 + \dfrac {E_{\vsi}} {|F|} + E_{\fgerightB_j} \right) 
             \,-\, \gamma \left( 1 + E_{\fgeleftB_j} \right),\,
             1 
         \right\}$
\\
$\displaystyle +$ & 
   $\displaystyle 
      \dfrac { \displaystyle 
               \min \left\{ 
                    \frac { \min \left\{ \Phi\,( n +E_{\vt}\,) - \gamma\,( n  +E_{\vt}\,),\,n \right\} }  { n } 
                      - 
                    \gamma\, \left( \deg(u) + E_{\ovr{u}} \right),\, 
                    \deg(u) 
                    \right\} 
             }
             { \displaystyle \deg(u) }$
\\
$\displaystyle +$ & 
   $\displaystyle 
        \dfrac { \min \left\{ \, 
                              \Phi\,( b_{\ovr{v}}-\iota_{\ovr{v}}+E_{\ovr{v}}\,) 
                                 - 
                              \gamma\,( b_{\ovr{v}}+E_{\ovr{v}}\,), \, b_{\ovr{v}} \, 
                      \right\} 
               } 
               { \displaystyle \din(\ovr{v}) }$
\\
$\displaystyle >$ & 
          $\gamma \, \left(  b_{f_{\ovr{u},\ovr{v}}} + E_{f_{\ovr{u},\ovr{v}}} \right)$
\end{longtable}
\vspace*{-0.5in}
\begin{gather*}
\hspace*{-2in}
\displaystyle
\pmb{\equiv}
\end{gather*}
\vspace*{-0.3in}
\begin{longtable}{rl}
   &  $\displaystyle 
         \min \left\{ 
              \Phi \left( 1 + \dfrac {E_{\vsi}} {|F|}  \right) 
                 \,-\, 
              \gamma \left( 2 + \dfrac {E_{\vsi}} {|F|} + E_{\fgerightB_j} \right) 
                 \,-\, 
              \gamma \left( 1 + E_{\fgeleftB_j} \right), \, 1 
             \right\}$
\\
$\displaystyle +$ & 
       $\dfrac { \min \left\{ \displaystyle
                           \Phi \left( 1 + \frac {E_{\vt} } {n} \right) 
                               \,-\, 
                           \gamma \left( 1  + \frac { E_{\vt} } {n} \right)  
                               \,-\, 
                           \gamma\, \left( \deg(u) + E_{\ovr{u}} \right),\, \deg(u) 
                       \right\} 
               } 
               { \displaystyle \deg(u) }$
\\
$\displaystyle +$ & 
         $\dfrac { \min \left\{ \, \Phi\,( \deg(\ovr{v}) -1+E_{\ovr{v}}\,) - \gamma\,( \deg(\ovr{v})+E_{\ovr{v}}\,), \, \deg(\ovr{v}) \, \right\} } { \din(\ovr{v}) }$
\\
$\displaystyle >$ & 
        $\gamma \, \left( 1 + E_{f_{\ovr{u},\ovr{v}}} \right)$
\end{longtable}
%
These constraints are satisfied provided all the previous constraints hold and the following holds: 
\begin{longtable}{rl}
  &  $\displaystyle 
       \Phi \left( 1 + \dfrac {E_{\vsi}} {|F|}  \right) 
         \,-\,
       \gamma \left( 2 + \dfrac {E_{\vsi}} {|F|} + E_{\fgerightB_j} \right)
         \,-\,
      \gamma \left( 1 + E_{\fgeleftB_j} \right)$
\\
$\displaystyle +$ & 
       $\dfrac { \displaystyle 
                 \Phi \left( 1 + \frac {E_{\vt} } {n} \right) 
                    - 
                 \gamma \left( 1  + \frac { E_{\vt} } {n} \right)  
                    - 
                 \gamma\, \left( \deg(u) + E_{\ovr{u}} \right) 
               }
               { \displaystyle \deg(u) }$
\\
$\displaystyle +$ & $\displaystyle
        \dfrac {  \displaystyle
                    \Phi \left( 1 + \frac {E_{\vt} } {n} \right) 
                        - 
                  \gamma \left( 1  + \frac { E_{\vt} } {n} \right)  
                        - 
                  \gamma\, \left( \deg(v) + E_{\ovr{v}} \right) 
               }
               { \displaystyle \deg(v) } \,>\, \gamma \, \left(  1 + E_{f_{\ovr{u},\ovr{v}}} \right)$
\end{longtable}
\vspace*{-0.3in}
\begin{gather}
\displaystyle
\hspace*{-1.3in}
\,\,\pmb{\equiv}\,\, 
\boxed{
\Phi 
> 
\gamma \left(
\dfrac 
{\displaystyle
6 + \frac{1}{\deg(u)}  + \frac { E_{\vt} } {n\,\deg(u)} + \frac {E_{\ovr{u}} }{\deg(u)} + \frac{1}{\deg(v)}  + \frac { E_{\vt} } {n\,\deg(v)} + \frac {E_{\ovr{v}} }{\deg(v)} + E_{f_{\ovr{u},\ovr{v}}} + \frac {E_{\vsi}} {|F|} + E_{\fgerightB_j} + E_{\fgeleftB_j}
}
{\displaystyle 
1 + \frac {E_{\vsi}} {|F|} + 
\frac{1}{\deg(u)} + \frac {E_{\vt} } {n\,\deg(u)} + \frac{1}{\deg(v)} + \frac {E_{\vt} } {n\,\deg(v)} } 
\right)
}
\label{eq15-het} 
\end{gather}
\begin{gather}
\displaystyle
\hspace*{-1.2in}
\Phi \left( 1 + \dfrac {E_{\vsi}} {|F|}  \right) 
\,-\,
\gamma \left( 2 + \dfrac {E_{\vsi}} {|F|} + E_{\fgerightB_j} \right)
\,-\,
\gamma \left( 1 + E_{\fgeleftB_j} \right)
\leq 
1
\,\,\pmb{\equiv}\,\, 
\boxed{
\Phi 
\leq
\gamma \left( \dfrac {\displaystyle 3 + \frac {E_{\vsi}} {|F|} + E_{\fgerightB_j} + E_{\fgeleftB_j}} {\displaystyle  1 + \frac {E_{\vsi}} {|F|} } \right) + \frac{1}{ 1 + \frac {E_{\vsi}} {|F|} }
}
\label{eq16-het} 
\\
\notag
\\
\displaystyle
\hspace*{-1.3in}
\Phi \left( 1 + \dfrac {E_{\vt} } {n} \,\right) - \gamma\,\left( 1  + \dfrac { E_{\vt} } {n} \right)  - \gamma\, \left( \deg(u) + E_{\ovr{u}} \right) \leq \deg(u)
\,\,\pmb{\equiv}\,\, 
\boxed{
\Phi \leq \gamma \left( \dfrac{\displaystyle 1  + \frac { E_{\vt} } {n} + \deg(u) + E_{\ovr{u}}} {\displaystyle 1 + \frac {E_{\vt} } {n} } \right) + \frac {\displaystyle \deg(u)} {\displaystyle  1 + \frac {E_{\vt} } {n}}
}
\label{eq17-het} 
\\
\notag
\\
\displaystyle
\hspace*{-1in}
\Phi \left( 1 + \dfrac {E_{\vsi}} {|F|}  \right) \,-\, \gamma \left( 2 + \frac {\displaystyle E_{\vsi}} {\displaystyle |F|} + E_{\fgerightB_j} \right) \,-\, \gamma \left( 1 + E_{\fgeleftB_j} \right)
+
\dfrac { \Phi \left( 1 + \frac {\displaystyle E_{\vt} } {\displaystyle n} \right) - \gamma \left( 1  + \frac {\displaystyle  E_{\vt} } {\displaystyle n} \right)  - \gamma\, \left( \deg(u) + E_{\ovr{u}} \right) } { \deg(u) }
\notag
\\
\displaystyle
+\,
\dfrac { \Phi\,( \deg(\ovr{v}) -1+E_{\ovr{v}}\,) - \gamma\,( \deg(\ovr{v})+E_{\ovr{v}}\,))} { \din(\ovr{v}) }
\,>\, 
\gamma \, \left( 1 + E_{f_{\ovr{u},\ovr{v}}} \right) 
\notag
\\
\displaystyle
\pmb{\equiv}
\notag
\\
\hspace*{-1in}
\boxed{
\displaystyle
\Phi 
> 
\gamma 
\left( 
\dfrac 
{6 + \frac {\displaystyle E_{\vsi}} {|F|} + E_{\fgerightB_j} + E_{\fgeleftB_j} + \frac { \displaystyle E_{\vt} } {\displaystyle n\,\deg(u)} + \frac {\displaystyle E_{\ovr{u}} + 1 } {\displaystyle  \deg(u) } + \frac {\displaystyle E_{\ovr{v}} } {\displaystyle \deg(v)} + E_{f_{\ovr{u},\ovr{v}}} 
}
{\displaystyle
2 + \frac {\displaystyle E_{\vsi}} {|F|} + \frac{\displaystyle 1}{\displaystyle \deg(u)}  + \frac {\displaystyle E_{\vt} } {\displaystyle n\,\deg(u)} + \frac {\displaystyle E_{\ovr{v}}\, -1 } {\displaystyle \deg(v)} 
}
\right) 
}
\label{eq18-het} 
\end{gather}

\item[(IV)]
By {\bf (II-b)} node $\ovr{h_{i,j}}$ fails at $t=3$ provided both the nodes $\ovr{u}$ and $\ovr{v}$ were shocked at $t=1$. 
Our goal is to make sure that node $\ovr{h_{i,j}}$ does not fail in any other condition (assuming the node itself was not shocked). 
Assuming the nodes $\ovr{u}$, $\ovr{v}$ and $f_{\ovr{u},\ovr{v}}$ were not shocked, the maximum amount of shock that $f_{\ovr{u},\ovr{v}}\in \ovr{F_{i,j}}$ 
can receive is when all the nodes before $f_{\ovr{u},\ovr{v}}$ in the path  
\pscirclebox[boxsep=false,framesep=0pt]{$\scriptstyle\,\vsi\,$}$\,\longleftarrow$\pscirclebox[boxsep=false,framesep=0pt]{$\fgerightB_j$}$\,\longleftarrow$\pscirclebox[boxsep=false,framesep=0pt]{$\fgeleftB_j$}$\,\longleftarrow$\pscirclebox[boxsep=false,framesep=0pt]{$p_j$}$\,\longleftarrow$\pscirclebox[boxsep=false,framesep=0pt]{$\scriptscriptstyle\,f_{\ovr{u},\ovr{v}}\,$}$\,$
were shocked and no more than one of the nodes $\ovr{u}$ or $\ovr{v}$ was shocked.
Based on this, the following constraints must hold for $\ovr{h_{i,j}}$ not to fail.
\begin{gather*}
\displaystyle
\dfrac 
{ \min \left\{ \sigma_1 + \sigma_2 - \gamma \left( b_{f_{\ovr{u},\ovr{v}}} + E_{f_{\ovr{u},\ovr{v}}} \right), \, b_{f_{\ovr{u},\ovr{v}}} \right \} }
{ \din (f_{\ovr{u},\ovr{v}}) } 
\,\,\leq \,\,
\dfrac { \gamma \left( b_{\ovr{h_{i,j}}} + E_{\ovr{h_{i,j}}} \right) } { |F_{i,j}| }
\\
\displaystyle
\pmb{\equiv}
\end{gather*}
\vspace*{-0.3in}
\begin{longtable}{l}
$\displaystyle 
    \min \scalebox{1.7}[1.4]{\Bigg \{} 
            \dfrac { 
                    \min \left\{ \frac { \displaystyle \min \left\{ \Phi\,( n +E_{\vt}\,) - \gamma\,( n  +E_{\vt}\,),\,n \right\} }  { \displaystyle n } 
                                    - \gamma\, \left( \deg(u) 
                                    + E_{\ovr{u}} \right),\, \deg(u) 
                          \right\} 
                   }
                   { \displaystyle \deg(u) }$
\\
                     $\hspace*{0.3in}+\,\,\,
            \dfrac { 
                    \min \left\{ \frac { \displaystyle \min \left\{ \Phi\,( n +E_{\vt}\,) - \gamma\,( n  +E_{\vt}\,),\,n \right\} }  { \displaystyle n } 
                                 - \gamma\, \left( \deg(v) 
                                 + E_{\ovr{v}} \right),\, \deg(v) 
                         \right\} 
                   }
                   { \displaystyle \deg(v) }$
\\
$\hspace*{0.3in} \displaystyle +\, 
         \min \left\{ 
             \min \left\{
                      \Phi \left( 1 + \dfrac {E_{\vsi}} {\displaystyle |F|} - E_{\fgerightB_j} \right) 
                         \,-\,
                      \gamma \left( 2 + \dfrac {E_{\vsi}} {|F|} + E_{\fgerightB_j} \right), \,
                       1
                  \right\}
                \,-\,
             \left(\gamma \left( 1 + E_{\fgeleftB_j}  \right)  -  \Phi E_{\fgeleftB_j} \right), \, 
             1
            \right\}$
\\
       $\hspace*{0.3in} -\, 
         \gamma \left( 1 + E_{f_{\ovr{u},\ovr{v}}} \right), \, 1 \scalebox{1.7}[1.4]{\Bigg \}} \,\,\leq\,\, \dfrac { \gamma \, E_{\ovr{h_{i,j}}} }{ |F_{i,j}|}$
\end{longtable}
\vspace*{-0.6in}
\begin{gather*}
\displaystyle
\hspace*{-3in}
\pmb{\equiv}
\end{gather*}
\vspace*{-0.3in}
\begin{longtable}{l}
 $\displaystyle 
     \min \scalebox{1.7}[1.4]{\Bigg\{} 
           \Phi \left( 
                        \dfrac{1}{\deg(u)} + \dfrac {E_{\vt} }{n\,\deg(u)} + \dfrac{1}{\deg(v)} + \dfrac {E_{\vt} }{n\,\deg(v)} 
                 \right)$
\\
 \hspace*{0.3in} $\displaystyle 
           -\, \gamma \left( 
                         3 + \dfrac{1}{\deg(u)}  + \dfrac {E_{\vt}}{n\,\deg(u)} + \dfrac{E_{\ovr{u}}}{\deg(u)} + \dfrac{1}{\deg(v)}  + \dfrac {E_{\vt}}{n\,\deg(v)} + \dfrac{E_{\ovr{v}}}{\deg(v)} + E_{f_{\ovr{u},\ovr{v}}} 
                      \right)$
\\
 \hspace*{0.3in} $\displaystyle 
           +\, \min \left\{
                      \Phi \left( 
                           1 + \dfrac {E_{\vsi}} {|F|} - E_{\fgerightB_j} + E_{\fgeleftB_j} 
                           \right) 
                     \,-\,
                     \gamma \left( 
                           3 + \dfrac {E_{\vsi}} {|F|} + E_{\fgerightB_j} + E_{\fgeleftB_j}  
                            \right), \, 
                      1
                   \right\},\, 
           1 
         \scalebox{1.7}[1.4]{\Bigg\} } \,\,\leq\,\, \dfrac { \gamma \, E_{\ovr{h_{i,j}}} }{ |F_{i,j}|}$
\end{longtable}
These constraints are satisfied provided all the previous constraints hold and the following holds: 
\begin{multline}
\displaystyle
\Phi \left( 1 + \dfrac {E_{\vsi}} {|F|} - E_{\fgerightB_j} + E_{\fgeleftB_j} \right) - \gamma \left( 3 + \dfrac {E_{\vsi}} {|F|} + E_{\fgerightB_j} + E_{\fgeleftB_j}  \right) \leq 1
\\
\displaystyle
\pmb{\equiv}\,\,
\boxed{
\Phi 
\leq
\gamma \left(
\dfrac {3 + \dfrac {E_{\vsi}} {|F|} + E_{\fgerightB_j} + E_{\fgeleftB_j} } { 1 + \dfrac {E_{\vsi}} {|F|} - E_{\fgerightB_j} + E_{\fgeleftB_j} }
\right)
+ 
\dfrac {1} {\displaystyle 1 + \frac {E_{\vsi}} {|F|} - E_{\fgerightB_j} + E_{\fgeleftB_j} }
}
\label{eq19-het} 
\end{multline}

\hspace*{-1in}
\fbox{
 \addtolength{\linewidth}{-2\fboxsep}%
 \addtolength{\linewidth}{-2\fboxrule}%
 \begin{minipage}{1.25\linewidth}
\begin{multline}
\displaystyle
\Phi 
\leq
\gamma \left( 
\dfrac
{ \displaystyle
6 + \frac{1}{\deg(u)}  + \frac {E_{\vt}}{n\,\deg(u)} + \frac{E_{\ovr{u}}}{\deg(u)} + \frac{1}{\deg(v)}  + \frac {E_{\vt}}{n\,\deg(v)} + \frac{E_{\ovr{v}}}{\deg(v)} + E_{f_{\ovr{u},\ovr{v}}} + \frac {E_{\vsi}} {|F|} + E_{\fgerightB_j} + E_{\fgeleftB_j}  
}
{ \displaystyle
1 + \frac{1}{\deg(u)} + \frac {E_{\vt} }{n\,\deg(u)} + \frac{1}{\deg(v)} + \frac {E_{\vt} }{n\,\deg(v)} + \frac {E_{\vsi}} {|F|} - E_{\fgerightB_j} + E_{\fgeleftB_j}
}
\right)
\\
\displaystyle
+
\dfrac { 1 }
{ \displaystyle
1 + \frac{1}{\deg(u)} + \frac {E_{\vt} }{n\,\deg(u)} + \frac{1}{\deg(v)} + \frac {E_{\vt} }{n\,\deg(v)} + \frac {E_{\vsi}} {|F|} - E_{\fgerightB_j} + E_{\fgeleftB_j}
}
\label{eq20-het} 
\end{multline}
 \end{minipage}
}

\begin{gather}
\hspace*{-1.3in}
\displaystyle
\boxed{
\Phi 
\leq 
\gamma 
\left( 
\dfrac 
{
\displaystyle
6 + \frac{1}{\deg(u)}  + \frac {E_{\vt}}{n\,\deg(u)} + \frac{E_{\ovr{u}}}{\deg(u)} + \frac{1}{\deg(v)}  + \frac {E_{\vt}}{n\,\deg(v)} + \frac{E_{\ovr{v}}}{\deg(v)} + E_{f_{\ovr{u},\ovr{v}}} + \frac {E_{\vsi}} {|F|} + E_{\fgerightB_j} + E_{\fgeleftB_j} + \frac{E_{\ovr{h_{i,j}}} }{ |F_{i,j}|} 
}
{
\displaystyle
1 + \frac{1}{\deg(u)} + \frac {E_{\vt} }{n\,\deg(u)} + \frac{1}{\deg(v)} + \frac {E_{\vt} }{n\,\deg(v)} + \frac {E_{\vsi}} {|F|} - E_{\fgerightB_j} + E_{\fgeleftB_j}
}
\right)
}
\label{eq21-het} 
\end{gather}
On the other hand, if exactly one of the nodes $\ovr{u}$ or $\ovr{v}$, say $\ovr{u}$, was shocked at $t=1$, then 
the maximum amount of shock that $f_{\ovr{u},\ovr{v}}\in \ovr{F_{i,j}}$ can receive is is modified, and the new conditions for our desired goal become as follows.
\begin{gather*}
\hspace*{-1.2in}
\displaystyle
\dfrac 
{ \displaystyle  \min \left\{ \sigma_1' + 
\frac { \min \left\{ \, \Phi\,( b_{\ovr{v}}-\iota_{\ovr{v}}+E_{\ovr{v}}\,) - \gamma\,( b_{\ovr{v}}+E_{\ovr{v}}\,), \, b_{\ovr{v}} \, \right\} } { \din(\ovr{v}) }
+ \sigma_2 
- \gamma \left( b_{f_{\ovr{u},\ovr{v}}} + E_{f_{\ovr{u},\ovr{v}}} \right), \, b_{f_{\ovr{u},\ovr{v}}} \right \} }
{ \displaystyle  \din (f_{\ovr{u},\ovr{v}}) } 
\,\leq \,
\dfrac { \gamma \left( b_{\ovr{h_{i,j}}} + E_{\ovr{h_{i,j}}} \right) } { |F_{i,j}| }
\\
\pmb{\equiv}
\end{gather*}
\begin{longtable}{l}
 $\displaystyle 
     \min \scalebox{1.7}[1.4]{\Bigg\{ }
             \dfrac { \min 
                        \left\{ \displaystyle 
                                 \frac { \displaystyle 
                                                   \min \left\{ 
                                                             \Phi\,\left( n +E_{\vt}\,\right) \,-\, \gamma\,\left( n  +E_{\vt}\,\right), \,n 
                                                        \right\} 
                                       }  
                                       { \displaystyle 
                                          n 
                                       } 
                                  - 
                                  \gamma\, \left( \deg(u) + E_{\ovr{u}} \right),\, \deg(u) 
                         \right\} 
                    } 
                    { \deg(u) }$
\\
 $\hspace*{0.5in}\displaystyle     + \,\,
            \dfrac { \min 
                       \left\{ \, 
                          \Phi\,\left( b_{\ovr{v}}-\iota_{\ovr{v}}+E_{\ovr{v}}\,\right) - \gamma\,\left( b_{\ovr{v}}+E_{\ovr{v}}\,\right), \, b_{\ovr{v}} \, 
                       \right\} 
                   } 
                   { \din(\ovr{v}) }$
\\
 $\displaystyle\hspace*{0.5in} + \,
     \min \left\{
        \min \left\{
               \Phi \left( 
                       1 + \dfrac {E_{\vsi}} {|F|} - E_{\fgerightB_j} 
                    \right) 
               \,-\,
               \gamma \left( 
                         2 + \dfrac {E_{\vsi}} {|F|} + E_{\fgerightB_j} 
                      \right), \,
                1
             \right\}
        \,-\,
        \Big(   
            \gamma \left( 1 + E_{\fgeleftB_j}  \right)  -  \Phi E_{\fgeleftB_j} 
        \Big), \, 
         1
         \right\}$
\\
 $\hspace*{0.4in}\displaystyle
    \,-\, \gamma \left( 1 + E_{f_{\ovr{u},\ovr{v}}} \right), \, 1 \scalebox{1.5}[1.4]{\Bigg \} }
    \,\,\leq \,\,
    \dfrac { \gamma \, E_{\ovr{h_{i,j}}} }{ |F_{i,j}|}$
\end{longtable}
\vspace*{-0.6in}
\begin{gather*}
\displaystyle
\pmb{\equiv}
\end{gather*}
\vspace*{-0.3in}
\begin{longtable}{l}
$\displaystyle
\min \scalebox{1.7}[1.4]{\Bigg \{ }
              \Phi \left( 1 + \dfrac{1}{\deg(u)} + \dfrac {E_{\vt}}{n\,\deg(u)} - \dfrac {1}{\deg(v)} + \dfrac {E_{\ovr{v}}}{\deg(v)} \right) 
              - 
              \gamma \left( 2 + \dfrac{1}{\deg(u)}  + \dfrac {E_{\vt}}{n\,\deg(u)} + \dfrac {E_{\ovr{u}}}{\deg(u)} + \dfrac {E_{\ovr{v}}}{\deg(v)} \right)$
\\
$\displaystyle
\hspace*{0.4in}
+\, 
\min \left\{
\min \left\{
\Phi \left( 1 + \dfrac {E_{\vsi}} {|F|} - E_{\fgerightB_j} \right) 
\,-\,
\gamma \left( 2 + \dfrac {E_{\vsi}} {|F|} + E_{\fgerightB_j} \right), \,
1
\right\}
\,-\,
\left(   \gamma \left( 1 + E_{\fgeleftB_j}  \right)  -  \Phi E_{\fgeleftB_j} \right), \, 
1
\right\}
$
\\
$\displaystyle
\hspace*{0.4in}
- \gamma \left( 1 + E_{f_{\ovr{u},\ovr{v}}} \right), \, 1 
\scalebox{1.7}[1.4]{\Bigg \} }
\,\,\leq \,\,
\dfrac { \gamma \, E_{\ovr{h_{i,j}}} }{ |F_{i,j}|}
$
\end{longtable}
These constraints are satisfied provided all the previous constraints hold and the following holds: 
\begin{gather}
\displaystyle
\hspace*{-1.5in}
\Phi \left( 1 + \dfrac {E_{\vsi}} {|F|} - E_{\fgerightB_j} \right) - \gamma \left( 2 + \dfrac {E_{\vsi}} {|F|} + E_{\fgerightB_j} \right) \leq 1
\,\pmb{\equiv}\,
\boxed{
\Phi 
\leq \gamma \left( 
\dfrac {2 + \frac {E_{\vsi}} {|F|} + E_{\fgerightB_j} }
{\displaystyle 1 + \frac {E_{\vsi}} {|F|} - E_{\fgerightB_j}}
\right) + 
\dfrac {1} 
{\displaystyle 1 + \frac {E_{\vsi}} {|F|} - E_{\fgerightB_j}}
}
\label{eq22-het} 
\end{gather}
\begin{multline}
\displaystyle
\Phi \left( 1 + \dfrac {E_{\vsi}} {|F|} - E_{\fgerightB_j} \right) 
-
\gamma \left( 2 + \dfrac {E_{\vsi}} {|F|} + E_{\fgerightB_j} \right)
-
\left(   \gamma \left( 1 + E_{\fgeleftB_j}  \right)  -  \Phi E_{\fgeleftB_j} \right) 
\leq
1 
\\
\displaystyle
\,\,\pmb{\equiv}\,\,
\boxed{
\Phi 
\leq 
\gamma \left( 
\dfrac { \displaystyle 2 + \frac {E_{\vsi}} {|F|} + E_{\fgerightB_j} + 1 + E_{\fgeleftB_j}  } 
{ \displaystyle 1 + \frac {E_{\vsi}} {|F|} - E_{\fgerightB_j} +  E_{\fgeleftB_j} }
\right)  
+
\dfrac { \displaystyle 1 } 
{ \displaystyle 1 + \frac {E_{\vsi}} {|F|} - E_{\fgerightB_j} +  E_{\fgeleftB_j} }
}
\label{eq23-het} 
\end{multline}
\begin{gather}
\hspace*{-1.5in}
\Phi \left( 1 + \frac{1}{\deg(u)} + \frac {E_{\vt}}{n\,\deg(u)} - \frac {1}{\deg(v)} + \frac {E_{\ovr{v}}}{\deg(v)} \right) 
- 
\gamma \left( 2 + \frac{1}{\deg(u)}  + \frac {E_{\vt}}{n\,\deg(u)} + \frac {E_{\ovr{u}}}{\deg(u)} + \frac {E_{\ovr{v}}}{\deg(v)} \right)
\notag
\\
\hspace*{-1in}
+ 
\Phi \left( 1 + \frac {E_{\vsi}} {|F|} - E_{\fgerightB_j} \right) 
\,-\,
\gamma \left( 2 + \frac {E_{\vsi}} {|F|} + E_{\fgerightB_j} \right)
\,-\,
\left(   \gamma \left( 1 + E_{\fgeleftB_j}  \right)  -  \Phi E_{\fgeleftB_j} \right)
- \gamma \left( 1 + E_{f_{\ovr{u},\ovr{v}}} \right)
\leq 
\frac { \gamma \, E_{\ovr{h_{i,j}}} }{ |F_{i,j}|}
\notag
\\
\pmb{\equiv}
\notag
\\
\hspace*{-1.5in}
\boxed{
\Phi 
\leq 
\gamma \left( 
\dfrac
{ \displaystyle
6 + \frac{1}{\deg(u)}  + \frac {E_{\vt}}{n\,\deg(u)} + \frac {E_{\ovr{u}}}{\deg(u)} + \frac {E_{\ovr{v}}}{\deg(v)}
+
\frac {E_{\vsi}} {|F|} + E_{\fgerightB_j} + E_{f_{\ovr{u},\ovr{v}}} + E_{\fgeleftB_j} 
+
\frac { \gamma \, E_{\ovr{h_{i,j}}} }{ |F_{i,j}|} 
}
{  \displaystyle 2 + \frac{1}{\deg(u)} + \frac {E_{\vt}}{n\,\deg(u)} - \frac {1}{\deg(v)} + \frac {E_{\ovr{v}}}{\deg(v)} + \frac {E_{\vsi}} {|F|} - E_{\fgerightB_j} +  E_{\fgeleftB_j} } 
\right)
}
\label{eq24-het} 
\end{gather}
\begin{gather}
\hspace*{-1.5in}
\Phi \left( 1 + \frac{1}{\deg(u)} + \frac {E_{\vt}}{n\,\deg(u)} - \frac {1}{\deg(v)} + \frac {E_{\ovr{v}}}{\deg(v)} \right) 
- 
\gamma \left( 2 + \frac{1}{\deg(u)}  + \frac {E_{\vt}}{n\,\deg(u)} + \frac {E_{\ovr{u}}}{\deg(u)} + \frac {E_{\ovr{v}}}{\deg(v)} \right)
\notag
\\
\hspace*{-1in}
+ 
\Phi \left( 1 + \frac {E_{\vsi}} {|F|} - E_{\fgerightB_j} \right) 
\,-\,
\gamma \left( 2 + \frac {E_{\vsi}} {|F|} + E_{\fgerightB_j} \right)
\,-\,
\left(   \gamma \left( 1 + E_{\fgeleftB_j}  \right)  -  \Phi E_{\fgeleftB_j} \right)
- \gamma \left( 1 + E_{f_{\ovr{u},\ovr{v}}} \right)
\leq 
1
\end{gather}
$\displaystyle\pmb{\equiv}\,\,$\fbox{
 \addtolength{\linewidth}{-2\fboxsep}%
 \addtolength{\linewidth}{-2\fboxrule}%
 \begin{minipage}{1\linewidth}
\vspace*{-0.1in}
\begin{multline}
\Phi 
\leq
\gamma \left( 
\dfrac 
{ \displaystyle
6 + \frac{1}{\deg(u)}  + \frac {E_{\vt}}{n\,\deg(u)} + \frac {E_{\ovr{u}}}{\deg(u)} + \frac {E_{\ovr{v}}}{\deg(v)}
+
\frac {E_{\vsi}} {|F|} + E_{\fgerightB_j} + E_{f_{\ovr{u},\ovr{v}}} + E_{\fgeleftB_j}  
}
{ \displaystyle 2 + \frac{1}{\deg(u)} + \frac {E_{\vt}}{n\,\deg(u)} - \frac {1}{\deg(v)} + \frac {E_{\ovr{v}}}{\deg(v)} + \frac {E_{\vsi}} {|F|} - E_{\fgerightB_j} +  E_{\fgeleftB_j} }
\right)
\\
+\,
\dfrac { 1 }
{ \displaystyle 2 + \frac{1}{\deg(u)} + \frac {E_{\vt}}{n\,\deg(u)} - \frac {1}{\deg(v)} + \frac {E_{\ovr{v}}}{\deg(v)} + \frac {E_{\vsi}} {|F|} - E_{\fgerightB_j} +  E_{\fgeleftB_j} }
\label{eq25-het} 
\end{multline}
 \end{minipage}
}
\end{description}
\end{description}
There are many choices of parameters $\gamma$, $\Phi$ and $E_y$'s satisfying 
inequalities~\eqref{eq1-het}--\eqref{eq10-het}; we exhibit just one:
\begin{gather*}
\gamma=4\,n^{-1000}
\,\,\,\,\,\,\,\,\,\,\,\,\,\,\,\,\,\,\,
\Phi=n^{-1000}
\,\,\,\,\,\,\,\,\,\,\,\,\,\,\,\,\,\,\,
\forall\, u\in\vl\cup\vr \colon E_{\ovr{u}}=1 
\,\,\,\,\,\,\,\,\,\,\,\,\,\,\,\,\,\,\,
E_{\vt}=n^3
\,\,\,\,\,\,\,\,\,\,\,\,\,\,\,\,\,\,\,
E_{\vsi}=n^2\,|F|
\\
\forall\, u\in\vl \, \forall\, v\in\vr \colon E_{f_{\ovr{u},\ovr{v}}}=1
\,\,\,\,\,\,\,\,\,\,\,\,\,\,\,\,\,\,\,
\forall \, h_{i,j}\in\ES \, \forall \, f_{u,v}\in F_{i,j} \colon E_{\ovr{h_{i,j}}}=1
\,\,\,\,\,\,\,\,\,\,\,\,\,\,\,\,\,\,\,
\forall\,j\colon E_{\fgerightB_j} = E_{\fgerightB_j} = \frac{1}{4} 
\end{gather*}
Remembering that $10\leq\deg(u)<n$ for any node $u\in\vl\cup\vr$ and $|F_{i,j}|<|F|$, 
it is relatively straightforward to verify that all the inequalities are satisfied for all sufficiently large $n$.
Note that 
\begin{multline*}
E=E_{\vt} + E_{\vsi} + \!\!\!\!\!\!\!\!\!\!\!\sum_{u\,\in\,\vl\,\cup\,\vr} \!\!\!\!\!\!\!\!\!\!\!\!\!\!\!E_{\ovr{u}} + \sum_{\,\,\{u,v\}\in F} \!\!\!\!\!\!\!\!E_{f_{\ovr{u},\ovr{v}}} + \!\!\!\!\!\!\sum_{h_{i,j}\in\,\ES} \!\!\!\!\!\!\!\!\!\!E_{\ovr{h_{i,j}}} \,\,\,+ \,\,\,\sum_{j=1}^{|F|} \left( E_{\fgerightB_j} + \sum E_{\fgeleftB_j}\, \right)
\\
=
n^3+n^2\,|F|+ n + \frac{3}{2}\,|F| + |\ES|
\notag
\end{multline*}
and thus the ratio of total external assets to total internal assets $E/I$ is large.
We can now finish our proof by selecting $\delta$ such that $\log^{1-\delta} n = \log^{1-\eps} |\ovr{V}|-1$ and showing the following:
\begin{description}
\item[\hspace*{0.2in}(completeness)]
If \minrep has a solution of size $\alpha+\beta$ on $G$ then 
$\vi^\ast\left(\ovr{G},T\right)\leq \alpha+\beta+2$.

\item[\hspace*{0.2in}(soundness)]
If {\em every} solution of \minrep on $G$ is of size at least $(\alpha+\beta) 2^{\log^{1-\delta} n}$
then 
\linebreak
$\vi^\ast\left(\ovr{G},T\right)\geq \dfrac{\alpha+\beta}{2}\,2^{\log^{1-\delta} n}$.
\end{description}

\subsection*{Proof of Completeness (\minrep has a solution of size $\alpha+\beta$)}

Let $V_1\subseteq\vl$ and $V_2\subseteq\vr$ be a solution of \minrep such that $|V_1|+|V_2|=\alpha+\beta$.
We shock the nodes $\vt$ and $\vsi$, and every node $\ovr{u}$ for every $u\in\vl\cup\vr$. By {\bf (I-a)} $\vt$ fails
at $t=1$, and by {\bf (I-b)} and {\bf (II-a)} every node in $\cup_{i=1}^\alpha \vl_i \bigcup \cup_{j=1}^\beta \vr_j$ fails on or before $t=2$.
By {\bf (III-a)}, {\bf (III-b)} and {\bf (III-c)} every node in $\left\{\vs\right\}\cup\fgerightB\cup\fgeleftB$ fails on or before $t=3$.
Since $V_1$ and $V_2$ are a valid solution of \minrep, for every super-edge $h_{i,j}$ there exists $u\in V_1$ and $v\in V_2$ 
such that $u\in\vl_i$, $v\in\vr_j$ and $\{u,v\}\in F$; since we shock the nodes $\ovr{u}$ and $\ovr{v}$, by {\bf (II-a)} both 
$\ovr{u}$ and $\ovr{v}$ fail at $t=1$, by {\bf (II-b)} the node $f_{\ovr{u},\ovr{v}}$ fails at $t=2$, and by {\bf (II-c)} 
the node $\ovr{h_{i,j}}$ fails at $t=3$. Thus, the network $\ovr{G}$ fails at $t=3$ and 
$\vi^\ast\left(\ovr{G},T\right)=\alpha+\beta+1$ for $t\geq 4$.

\subsection*{Proof of Soundness (every solution of \minrep is of size at least $(\alpha+\beta) 2^{\log^{1-\delta} n}$)}

We will prove the logically equivalent contrapositive of our claim, \IE, we will show that if 
$\vi^\ast\left(\ovr{G},T\right)<\frac{\alpha+\beta}{2}\,2^{\log^{1-\delta} n}$ then 
\minrep has a solution of size strictly less than $(\alpha+\beta) 2^{\log^{1-\delta} n}$.
Consider a solution of \nam\ on $\ovr{G}$ that shocks at most $z=\frac{\alpha+\beta}{2}\,2^{\log^{1-\delta} n}$ nodes.
Note that the nodes $\vt$ and $\vsi$ must be shocked at $t=1$ by Proposition~\ref{obs1}(a). 
By {\bf (I-a)} and {\bf (III-a)}, the nodes $\vt$ and $\vsi$ fails at $t=1$, by {\bf (I-b)} and {\bf (III-c)} 
every node in $\ovr{\vl}\cup\ovr{\vr}\cup\fgerightB$ fails at $t=2$, by {\bf (III-e)} every node in $\fgeleftB$
fails at $t=3$, by {\bf (III-f)} every node $f_{\ovr{u},\ovr{u}}$ fails at $t=4$ unless it was shocked at $t=1$ and 
by {\bf (IV)} a node $\ovr{h_{i,j}}$ fails only if $\ovr{h_{i,j}}$, $f_{ovr{u},ovr{u}}\in \ovr{F_{i,j}}$ or {\em both} the nodes 
$\ovr{u}$ and $\ovr{v}$ were shocked at $t=1$.
We ``normalize'' this given solution in the following manner (each step of the normalization assumes that the previous steps
have been already carried out): 
\begin{itemize}
\item
If a node from $\fgeleftB\cup\fgerightB$ was shocked at $t=1$, we do not shock it. By {\bf (III)} this has no effect on the failure
of the network.

\item
If a node $f_{\ovr{u},\ovr{v}}\in \ovr{F_{i,j}}$ was shocked, we do not shock it but instead shock the nodes $\ovr{u}$ and $\ovr{v}$ if they were not already shocked in 
the given solution. This at most doubles the number of nodes shocked and, by {\bf (II-b)}, the node $f_{u,v}$ fails at $t=2$ and the node
$\ovr{h_{i,j}}$ fails at $t=3$ if it was not shocked at $t=1$. Thus, after this sequence of normalization steps, we may assume that no $f_{\ovr{u},\ovr{v}}$ node was shocked.

\item
If a node $\ovr{h_{i,j}}$ was shocked at $t=1$, we do not shock it
but instead shock the nodes $\ovr{u}$ and $\ovr{v}$ (for some $u$ and $v$ such that $\{u,v\}\in F_{i,j}$) if they were not already shocked in 
the given solution. This at most doubles the number of nodes shocked and, by {\bf (II-b)}, the node $f_{u,v}$ fails at $t=2$ and the node
$\ovr{h_{i,j}}$ fails at $t=3$. Thus, after this sequence of normalization steps, we may assume that no $\ovr{h_{i,j}}$ node was shocked.
\end{itemize}
These normalizations result in a solution of \nam\ of size at most $2\,z$ in which the nodes $\vt$, $\vsi$, a subset 
$\ovr{V_1}\subseteq \ovr{\vl}$ and a subset $\ovr{V_2}\subseteq \ovr{\vr}$ of nodes.
Our solution of \minrep is $V_1=\{\,v\,|\,\ovr{v}\in\ovr{V_1}\,\}\subseteq\vl$ and $V_2=\{\,v\,|\,\ovr{v}\in\ovr{V_2}\,\}\subseteq\vr$
of size $2z-2<2z$.
Since failure of every $\ovr{h_{i,j}}$ is attributed to shocking two nodes $\ovr{u}$ and $\ovr{v}$ such that $f_{\ovr{u},\ovr{v}}\in \ovr{F_{i,j}}$,
every super-edge $h_{i,j}$ of $G$ is witnessed by the two nodes $u$ and $v$. 
\end{proof}

\section{Homogeneous Networks, $\pmb{\mbox{\dnam}}$, any $T$, hardness and exact algorithm}

\begin{theorem}\label{dual-np}~\\
\noindent
{\bf (a)}
Assuming $\mathsf{P}\neq\NP$, 
$\dvi^\ast(G,T,\kappa)$ cannot be approximated within a factor of ${\left(1-{\e}^{-1}+\delta\right)}^{-1}$, for any $\delta>0$,
even if $G$ is a DAG ($\e$ is the base of natural logarithm).

\vspace*{5pt}
\noindent
{\bf (b)}
If $G$ is a rooted in-arborescence then $\dvi^\ast(G,T,\kappa)<\dfrac{\kappa}{n}\,\left(1+\mathrm{deg}_{\mathrm{in}}^{\max} \left( \frac{\Phi}{\gamma}-1 \right) \right)$,
where $\mathrm{deg}_{\mathrm{in}}^{\max}=\max_{v\in V}\{\din(v)\}$ is the maximum in-degree over all nodes of $G$. 
Moreover, under the assumption that any individual node of the network can be failed by shocking, 
$\dvi^\ast(G,T,\kappa)$ can be computed exactly in $O(n^3)$ time.
\end{theorem}

\begin{proof}~\\
{\bf (a)}
The max $\kappa$-cover problem is defined as follows. An instance of the problem is an universe $\cU$ of $n$ elements, a collection of $m$ sets $\cS$ over $\cU$, 
and a positive integer $\kappa$. The goal is to pick a sub-collection $\cS'\subseteq\cS$ of $\kappa$ sets such that the number of elements covered, namely 
$\big|\cup_{S\in\cS'} S\big|$, is {\em maximized}. Let $\opt$ denote the maximum number of elements covered by an optimal solution of the max $\kappa$-cover problem.
It was shown in~\cite{F98} that, assuming $\mathsf{P}\neq\NP$, the max $\kappa$-cover problem cannot be 
approximated within a factor of $\dfrac{1}{\left(1-\frac{1}{\e}+\delta\right)}$ for any constant $\delta>0$. More precisely, \cite{F98} provides a polynomial-time 
reduction for a restricted but still $\NP$-hard version of the Boolean satisfiability problem (3-CNF5)
instances of max $\kappa$-cover with $\kappa=\big|\cU\big|^{\alpha}$, for some constant $0<\alpha<1$, and shows that
\begin{description}
\item[(1)]
if the CNF formula is satisfiable, then $\opt=\big|\cU\big|$;

\item[(2)]
if the CNF formula is not satisfiable, then $\opt<\left( 1 - \dfrac{1}{\e} + g(\kappa) \right)\big|\cU\big|$, where $g(\kappa)\to 0$ as $\kappa\to\infty$.
\end{description}
Our reduction from max $\kappa$-cover to {\sc Dual-Stab$_{T,\kappa}$} is as follows\footnote{However,
this exact construction will not work in the proof of Theorem~\ref{th1} since the entire network needs to fail in that proof.}.
In our graph $G=(V,F)$, we have an element node $\tilde{u}$ for every element $u\in\cU$, a set node $\tilde{S}$ for every set $S\in\cS$, and directed edges 
$(\tilde{u},\tilde{S})$ for every element $u\in\cU$ and set $S\in\cS$ such that $u\in S$. 
Thus, $n=|V|=|\cU|+|\cS|$ and $|F|=\sum_{S\in\cS} |S|$. We now set the remaining parameters as follows: $E=n$, $\gamma=n^{-2}$ and $\Phi=1$.
Now, we observe the following:
\begin{itemize}
\item
If an element node $\tilde{u}$ is shocked, it does not fail since $\Phi\left(\din\left(\tilde{u}\right)-\dout\left(\tilde{u}\right)+\frac{E}{n}\right)\leq 0$ whereas 
$\gamma \left(\din\left(\tilde{u}\right)+\frac{E}{n} \right)=n^{-2}>0$.

\item
If a set node $\tilde{S}$ is shocked, it fails since $\Phi\left(\din\left(\tilde{S}\right)-\dout\left(\tilde{S}\right)+\frac{E}{n}\right)\geq 2$
whereas 
\linebreak
$\gamma \left(\din\left(\tilde{S}\right)+\frac{E}{n} \right)\leq\frac{n+1}{n^2}<1$.

\item
If a set node $\tilde{S}$ is shocked, then every element node $\tilde{u}$ for $u\in S$ fails at $t=2$. To observe this, note that 
\begin{gather*}
\frac {\min \left\{\Phi\left(\din\left(\tilde{S}\right)-\dout\left(\tilde{S}\right)+\frac{E}{n}\right)-\gamma \left(\din\left(\tilde{u}\right)+\frac{E}{n} \right),\,\din\left(\tilde{S}\right)\right\}} {\din\left(\tilde{S}\right)}
\\
\geq
\frac{2-\frac{n+1}{n^2}}{n}
>
\frac{n+1}{n^2}
\geq
\gamma \left(\din\left(\tilde{S}\right)+\frac{E}{n} \right)
\end{gather*}

\item
Since the longest directed path in $G$ has one edge, no new nodes fails during $t>2$.
\end{itemize}
Based on the above observations, one can identify the sets selected in max $k$-cover with the set nodes selected for shocking in 
{\sc Dual-Stab$_{T,\kappa}$} on $G$ to conclude that $\dvi^\ast(G,T,\kappa)=\opt+\kappa$.
Thus, using {\bf (1)} and {\bf (2)}, inapproximability gap is
\[
\frac {\big|\cU\big| + \kappa} {\left( 1 - \frac{1}{\e} + g(\kappa) \right)\big|\cU\big| + \kappa}
=
\frac {\big|\cU\big| + \big|\cU\big|^\alpha} {\left( 1 - \frac{1}{\e} + g(\kappa) \right)\big|\cU\big| + \big|\cU\big|^\alpha}
\to
\frac{1}{1-\frac{1}{\e}+\delta}\,\mbox{ as } \big|\cU\big|\to\infty\,\mbox{ for any } \delta>0
\]

\noindent
{\bf (b)}
Using Lemma~\ref{bound1}, we have  
\[
\dvi^\ast(G,T,\kappa)
<
\frac{\displaystyle  \kappa \, \left( { \displaystyle \max_{u\in V} \Big\{\, \iz(u)\, \Big\}  } \right) }{n}
<
\dfrac{\kappa}{n}\,\left(1+\mathrm{deg}_{\mathrm{in}}^{\max} \left( \frac{\Phi}{\gamma}-1 \right) \right)
\]
To provide a polynomial time algorithm for $\dvi^\ast(G,T,\kappa)$, we suitably modify the algorithm described 
in the proof of Theorem~\ref{poly1}. We redefine $\snsvi^\ast(G,T,u,v)$ and $\ssvi^\ast(G,T,u)$ in the following manner:
\begin{itemize}
\item
For every node $u'\in\nabla(u)$ and every integer $0\leq k\leq\kappa$, $\snsdvi^\ast(G,T,u,u',k)$ is the number of nodes in an optimal solution of \dnam\ 
(or $\infty$ if there is no feasible solution of \dnam) 
for the subgraph induced by the nodes in $\Delta(u)$ assuming the following:
\begin{itemize}
\item
$u'$ was shocked,

\item
$u$ was not shocked,

\item
no node in the path $u'\leadsto u$ except $u'$ was shocked, and 

\item
total number of shocked nodes in $\Delta(u)$ is exactly $k$.
\end{itemize}

\item
For every integer $0\leq k\leq\kappa$, $\ssdvi^\ast(G,T,u,k)$ is the number of nodes in an optimal solution of \dnam\ 
for the subgraph induced by the nodes in $\Delta(u)$ 
(or $\infty$, if there is no feasible solution of \nam\ under the stated conditions)
assuming that the node $u$ was shocked (and therefore failed), and the number 
of shocked nodes in $\Delta(u)$ is exactly $k$.
\end{itemize}
Computing these quantities becomes slightly more computationally involved as shown below.
\begin{description}
\item[Computing $\ssdvi^\ast(G,T,u,k)$ when $\din(u)=0$:]~\\

\vspace*{-0.1in}
$\ssdvi^\ast(G,T,u,1)=1$ and $\ssdvi^\ast(G,T,u,k)=-\infty$ for any $k\neq 1$.

\item[Computing $\snsdvi^\ast(G,T,u,u',k)$ when $\din(u)=0$:]
$\,$

\begin{itemize}
\item
If $u\in\iz(u')$ then shocking node $v$ makes node $u$ fail. Thus, $\snsvi^\ast(G,T,u,u',1)=1$ and 
$\snsvi^\ast(G,T,u,u',k)=-\infty$ for any $k\neq 1$.

\item
Otherwise, node $u$ does not fail. Thus, $\snsdvi^\ast(G,T,u,u')=-\infty$.
\end{itemize}

\item[Computing $\ssdvi^\ast(G,T,u)$ when $\din(u)>0$:]
In this case we have 
\begin{multline*}
\ssdvi^\ast(G,T,u,k) \,\,=\,\, 1 \,\,+ 
\\
\min_{k_1+k_2+\dots+k_{\din(u)}=k-1}\,\left\{\,\sum_{i=1}^k \min \Big\{\, \ssdvi^\ast(G,T,v_i,k_i),\, \snsdvi^\ast(G,T,v_i,u,k_i)\, \Big\}\, \right\}
\end{multline*}

\item[Computing $\snsdvi^\ast(G,T,u,u',k)$ when $\din(u)>0$:]
Since $u'$ is shocked and $u$ is not shocked, the following cases arise:
\begin{itemize}
\item
If $u\not\in\iz(u')$ then then $u$ does not fail. Then, 
\begin{multline*}
\snsdvi^\ast(G,T,u,u',k) \,\,= 
\\
\min_{k_1+k_2+\dots+k_{\din(u)}=k}\,\left\{\, \sum_{i=1}^{\din(u)} \min \Big\{\, \ssdvi^\ast(G,T,v_i,k_i),\, \snsvi^\ast(G,T,v_i,u',k_i)\,\Big\}\,\right\}
\end{multline*}

\item
Otherwise, $u\in\iz(u')$, and therefore $u$ fails when $u'$ is shocked. Then, 
\begin{multline*}
\snsdvi^\ast(G,T,u,u',k) \,\,=\,\, 1\,\,+ 
\\
\min_{k_1+k_2+\dots+k_{\din(u)}=k}\,\left\{\, \sum_{i=1}^{\din(u)} \min \Big\{\, \ssdvi^\ast(G,T,v_i,k_i),\, \snsdvi^\ast(G,T,v_i,u',k_i)\,\Big\}\,\right\}
\end{multline*}
\end{itemize}
\end{description}
It only remains to show how we compute 
\\
$\displaystyle\min_{k_1+k_2+\dots+k_{\din(u)}=\digamma}\, \left\{\,\sum_{i=1}^{\din(u)} \min \Big\{\, \ssdvi^\ast(G,T,v_i,k_i),\, \snsdvi^\ast(G,T,v_i,u',k_i)\,\Big\}\,\right\}$ 
for $\digamma\in\{k-1,k\}$ in polynomial time.
It is easy to cast this problem as an instance of the unbounded integral knapsack problem in the following manner:
\begin{itemize}
\item
We have $\din(u)$ {\em objects} $\O_1,\O_2,\dots,\O_{\din(u)}$, each of {\em unlimited} supply and {\em weight} $1$. 

\item
The {\em cost} of selecting $k_i$ objects of the type $\O_i$ is 
\[
\min \Big\{\, \ssdvi^\ast(G,T,v_i,k_i),\, \snsdvi^\ast(G,T,v_i,u',k_i)\,\Big\}
\]

\item
The {\em goal} is to select a total of {\em exactly} $\digamma$ objects such that the total cost is {\em minimum}.
\end{itemize}
The standard pseudo-polynomial time dynamic programming algorithm for Knapsack can be used to solve the above instance in 
$\mathrm{O}\big(k\,\din(u)\big)=\mathrm{O}\left(n^2\right)$ time.
Thus, the total running time of our algorithm is $\mathrm{O}\left(n^3\right)$.
\end{proof}

\section{Heterogeneous Networks, $\pmb{\mbox{{{{\sc Dual-Stab$_{2,\Phi,\kappa}$}}}}}$, Stronger Inapproximability} 

We show that $\dvi^\ast(G,2,\kappa)$ cannot be approximated within a large approximation factor
provided a complexity-theoretic assumption is satisfied. To understand this assumption, we recall the following definitions from~\cite{A11}.

A random $(m,n,d)$ hyper-graph $H$ is a random hyper-graph of $n$ nodes, $m$ hyper-edges each having having exactly $d$ nodes obtained by choosing
each hyper-edge independently and uniformly at random. For our purpose, assume that $d$ is a constant, and $m\geq n^c$ for some constant $c>3$.
Let $Q\colon \{0,1\}^d \mapsto \{0,1\}$ denote a $d$-ary predicate, and 
let $\mathcal{F}_{Q,m}$ be a distribution over $d$-local functions from $\{0,1\}^n$ to $\{0,1\}^m$ by defining the random 
$d$-local function $f_{H,Q}\colon \{0,1\}^n\mapsto \{0,1\}^m$ to be the function
whose $i\tx$ output is computed by applying the predicate $Q$ to the $d$ inputs that are indexed by the $i\tx$ hyper-edge of $H$.
Finally, the $\kappa$ densest sub-hypergraph problem ($\mathsf{DS}_\kappa$) is defined as follows: {\em given an hyper-graph $G=(V,F)$ with $n=|V|$ and $m=|F|$ 
such that every hyper-edge contains {\em exactly} $d$ nodes and an integer $\kappa>0$, 
select a subset $V'\subseteq V$ of exactly $\kappa$ nodes which maximizes 
$\big|\,\big\{\,\{u_1,u_2,\dots,u_d\}\in F\,|\, u_1,u_2,\dots,u_d\in V'\big\}\,\big|$}.

The essence of the complexity-theoretic assumption is that if, for a suitable choice of $Q$, $\mathcal{F}_{Q,m}$ is a collection of one-way 
functions, then $\mathsf{DS}_\kappa$ is hard to approximate. More precisely, the assumption is: 
\label{star-ref}

\hangindent=10pt
\hangafter=0
\noindent
($\displaystyle\pmb{\star}$)
If $\mathcal{F}_{Q,m}$ is $\displaystyle\nicefrac{1}{\mathrm{o}\left(1/\sqrt{n}\,\log n\right)}$-pseudorandom, then for 
$\displaystyle\kappa=n^{1-{\frac{c-3}{2d}}}$ for some constant $c>3$ there exists 
instances $G=(V,F)$ of $\mathsf{DS}_\kappa$ with $m\geq n^c$ such that it is not possible to decide in polynomial time if there is a solution of $\mathsf{DS}_\kappa$ 
with at least $\dfrac{(1+o(1))\,m}{n^{{\frac{(c-3)}{2}\left(1-\frac{1}{d}\right)}}}$ edges (the ``yes'' instance), or if every solution of $\mathsf{DS}_\kappa$ has at 
most $\dfrac{(1-o(1))\,m}{n^{{\frac{c-3}{2}}}}$ edges (the ``no'' instance).

\begin{theorem}\label{dual-np-2}
Under the technical assumption {\bf ($\pmb{\star}$)}, 
$\dvi^\ast(G,2,\kappa)$ cannot be approximated within a ratio of $n^{\delta}$ for some constant $\delta>0$ even if $G$ is a DAG.
\end{theorem}

\begin{proof}
Given an instance $G=(V,F)$ of $\mathsf{DS}_\kappa$ as stated in {\bf ($\pmb{\star}$)}, we construct an instance graph 
$\ovr{G}=(\ovr{V},\ovr{F})$ as follows:
\begin{itemize}
\item
For every node $u\in V$, we have a node $\ovr{u}\in\ovr{V}$, and for every edge $e=\{u_1,u_2,\dots,u_d\}\in F$, we have a node $\ovr{e}$ (also denoted by $\ovr{\{u_1,u_2,\dots,u_d\}}\,$) 
in $\ovr{V}$. Thus, the total number of nodes of $\ovr{G}$ is $|\ovr{V}| = m+n$.

\item
For every hyper-edge $e=(u_1,u_2,\dots,u_d)\in F$, we have $d$ edges $(e,u_1),(e,u_2),\dots,(e,u_d)\in\ovr{F}$. 
We set the weight (share of internal asset) of every edge $(e,u_i)$ to $2$. Thus, $|I|=2dm$.
\end{itemize}
Let the share of external assets for a node (bank) $\ovr{y}\in\ovr{V}$ be denoted by $E_{\ovr{y}}$ (thus, $\sum_{\ovr{y}\in \ovr{V}} E_{\ovr{y}}=E$).
We will select the remaining network parameters as follows.
For each $e\in F$, $E_{\ovr{e}}=1.99d$, and for each $u\in V$, $E_{\ovr{u}}=0$. Thus, $E=1.99dm$.
Finally, we set $\Phi=1$ and $\gamma=\nicefrac{1}{2}$. We prove the following:
\begin{description}
\item[(completeness)]
If $\DS_\kappa$ has a solution with $\alpha\geq\dfrac{\big(1+o(1)\big)\,m}{n^{\frac{c-3}{2}\left(1-\frac{1}{d}\right)}}$ hyper-edges then then 
\\
$\dvi^\ast\left(\ovr{G},2,\kappa\right)\geq \kappa+\alpha$.

\item[(soundness)]
If {\em every} solution of $\DS_\kappa$ has at most $\beta=\dfrac{\big(1-o(1)\big)\,m}{n^{\frac{c-3}{2}}}$ hyper-edges then 
\\
$\dvi^\ast\left(\ovr{G},2,\kappa\right)\leq \kappa+\beta$.
\end{description}
Note that with $c=5$ (and, thus $m\geq n^5$), and sufficiently large $d$ and $n$, we have 
\[
\frac{\kappa+\alpha}{\kappa+\beta}
=
\frac
{  n^{1-\frac{c-3}{2d}}+\frac{\big(1+o(1)\big)\,m}{n^{\frac{c-3}{2}\left(1-\frac{1}{d}\right)}} }
{  n^{1-\frac{c-3}{2d}}+\frac{\big(1-o(1)\big)\,m}{n^{\frac{c-3}{2}}} }
=
\frac
{  n^{1-\frac{1}{d}}+\frac{\big(1+o(1)\big)\,m}{n^{1-\frac{1}{d}}} }
{  n^{1-\frac{1}{d}}+\frac{\big(1-o(1)\big)\,m}{n} }
\geq
\big(1-\mathrm{o}(1)\big)\,n^{\nicefrac{1}{d}}
\]
which proves the theorem with $\delta=\nicefrac{1}{d}$. 

\subsection*{Proof of Completeness ($\DS_\kappa$ has a solution with $\alpha$ hyper-edges)}

Let $V'\subseteq V$ be a solution of $\DS_\kappa$ with at least $\alpha$ hyper-edges. 
We shock all the nodes in $\vs=\left\{\ovr{u}\,|\,u\in V' \right\}$. Every shocked node $\ovr{u}$ fails at 
$t=1$ since $\Phi\,\left(b_{\ovr{u}}-\iota_{\ovr{u}}+E_{\ovr{u}} \right)=2\,\din(\ovr{u})>\din(\ovr{u})=\gamma\,\left(b_{\ovr{u}}+E_{\ovr{u}}\right)$.
Now, consider a hyper-edge $e=(u_1,u_2,\dots,u_d)\in F$ such that $u_1,u_2,\dots,u_d\in V'$. Then, the node $\ovr{e}$ fails at $t=2$ since 
\[
\sum_{i=1}^d \frac{ \min \big\{ \Phi\,\left(b_{\ovr{u_i}}-\iota_{\ovr{u_i}}+E_{\ovr{u_i}} \right) - \gamma\,\left(b_{\ovr{u_i}}+E_{\ovr{u_i}}\right)  ,\, b_{\ovr{u_i}} \big\} } {\din(\ovr{u_i})}
=
d > 0.995d= \gamma\,\left(b_{\ovr{e}}+E_{\ovr{e}}\right)
\]

\subsection*{Proof of Soundness (every solution of $\DS_\kappa$ has at most $\beta$ hyper-edges)}

We will prove the logically equivalent contrapositive of our claim, \IE, we will show that if 
$\dvi^\ast\left(\ovr{G},2,\kappa\right)>\beta+\kappa$ then $\DS_\kappa$ has a solution of with strictly more than $\beta$ hyper-edges.
First, note that we can assume without loss of generality that, for any hyper-edge $e\in F$, the node $\ovr{e}$ is not shocked.
Otherwise, if we shock node $\ovr{e}$, then it does not fail since at
$t=1$ since $\Phi\,\left(b_{\ovr{e}}-\iota_{\ovr{e}}+E_{\ovr{e}} \right)=-0.01d<0.995d=\gamma\,\left(b_{\ovr{e}}+E_{\ovr{e}}\right)$, and 
in fact doing so increases its equity to $1.005d$. Since the equity of $\ovr{e}$ increased by shocking it, if this node failed in the
given solution then it would also fail if it was not shocked. So, we can instead shock a node $\ovr{u}$ that was not shocked in the given
solution; such a node must exist since $\kappa<n$. 

Note that we have already shown in the proof of the completeness part that, for any $e=(u_1,u_2,\dots,u_d)\in F$, if the $d$ nodes $\ovr{u_1},\ovr{u_2},\dots,\ovr{u_d}$ are 
shocked then $\ovr{e}$ fails at $t=2$. Thus, our proof is complete provided we show that such a node $\ovr{e}$ does {\em not} fail at $t=2$ if {\em at least} one 
of the nodes $\ovr{u_1},\ovr{u_2},\dots,\ovr{u_d}$ is {\em not} shocked. Let $S\subset \{\ovr{u_1},\ovr{u_2},\dots,\ovr{u_d}\}$ be the set of shocked nodes among these $d$ nodes.
Then, $\ovr{e}$ does not fail at $t=2$ since 
\[
\sum_{u_i\in S} \frac{ \min \big\{ \Phi\,\left(b_{\ovr{u_i}}-\iota_{\ovr{u_i}}+E_{\ovr{u_i}} \right) - \gamma\,\left(b_{\ovr{u_i}}+E_{\ovr{u_i}}\right)  ,\, b_{\ovr{u_i}} \big\} } {\din(\ovr{u_i})}
\leq
d-1
\leq
0.995d
=
\gamma\,\left(b_{\ovr{e}}+E_{\ovr{e}}\right)
\]
for all sufficiently large $d$. 
\end{proof}

\section{Software Availability}

An interactive software (called FIN-STAB) implementing an expanded version of the shock propagation algorithm shown in Table~\ref{t1} is
available from the website 
\\
\url{http://www2.cs.uic.edu/~dasgupta/financial-simulator-files}

\section{Conclusions}

In this paper, we have formalized a model for idiosyncratic propagation of shocks to a banking network, defined two possible
stability measures, provided their computational properties and discussed the implications of our results on the banking system.
We view our work as a necessary first step towards understanding vulnerabilities of banking systems due to 
sudden loss of external assets, and hope that it will generate sufficient interests in both the banking network community and
the network algorithms community to further investigate and refine these stability issues. 

\section*{Acknowledgements}

The authors would like to thank the organizers of the 
Industrial-Academic Workshop on Optimization in Finance and Risk Management at the Fields Institute in Toronto (Canada)
for an opportunity to discuss some of the results in this paper and receive valuable feedbacks.

\appendix

\begin{center}
{\bf APPENDIX}
\end{center}

For the benefit of the reader, we provide explanations for a few finance terminologies frequently used in this paper.
\begin{description}
\item[External asset:]
refers to the case of financial institutions borrowing from investors and similar outside entities.

\item[Interbank exposure:]
refers to the case of financial institutions borrowing from other financial institutions.

\item[Net worth or equity:]
a fixed proportion of the total asset of a bank. In general, higher equity imply better stability for an individual bank.
\end{description}
\end{document}